\documentclass[a4paper]{article}
\usepackage[utf8]{inputenc}
\usepackage{amsfonts,anysize,amsmath,amssymb,amsthm,enumitem,graphicx,epstopdf,adjustbox,multirow,cite}
\usepackage[a4paper,left=2cm,right=2cm,top=2cm,bottom=2cm]{geometry}
\theoremstyle{plain}
\newcommand{\RomanNumeralCaps}[1]
    {\MakeUppercase{\romannumeral #1}}
\newtheorem{theorem}{Theorem}[]

\newtheorem{remark}[theorem]{Remark}

\usepackage[color links=true, urlcolor= blue, citecolor=red]{hyperref}
\usepackage{float}
\usepackage[caption=true]{subfig}
\usepackage{appendix}
\begin{document}
\begin{center}
{\bf 
\Large Ecological system with fear induced group defence and prey refuge}
\end{center}
\begin{center}
	{\bf Shivam Yadav$^{a}$, Jai Prakash Tripathi$^{a}$ \footnote[1]{Corresponding author: Jai Prakash Tripathi, Department of Mathematics, Indian Institute of Technology Central University of Rajasthan, Bandar Sindri, Kishangarh, Ajmer 305817, Rajasthan, India, Email: jtripathi85@gmail.com}, Shrichand Bhuria$^{a}$, Satish Kumar Tiwari$^{b},$ Deepak Tripathi$^a$,  Vandana Tiwari$^c,$ Ranjit Kumar Upadhyay}$^d$, \bf Yun Kang$^e$  \\
	\vspace*{1cm}
	{\small $^a$ Department of Mathematics, Central University of Rajasthan,
		Bandar Sindri, Kishangarh \\ Ajmer 305817, Rajasthan, India}\\
	{\small $^b$	Assam Energy Institute (Centre of Rajiv Gandhi Institute of Petroleum Technology, Jais, Amethi),\\Sivasagar, Assam - 785697, India}\\
 {\small $^c$	Gandhi Smarak PG College, Samodhpur, Jaunpur 223102, U.P., India}\\
 {\small $^d$	Department of Mathematics \& Computing, Indian Institute of Technology (Indian School of Mines),\\ Dhanbad 826004, Jharkhand, India}\\
 {\small $e$ College of Integrative Sciences and Arts, Arizona State University, Mesa, AZ85212, USA.}
\end{center}
\begin{abstract}
In this study, we investigate the dynamics of a spatial and non spatial prey-predator interaction model that includes the following: (\romannumeral 1) fear effect incorporated in prey birth rate; (\romannumeral 2) group defence of prey against predators; and (\romannumeral 3) prey refuge. 
We provide comprehensive mathematical analysis of extinction and persistence scenarios for both prey and predator species. We investigate how the prey and predator equilibrium densities are influenced by the prey birth rate and fear level. To better explore the dynamics of the system, a thorough investigation of bifurcation analysis has been performed using fear level, prey birth rate, and prey's death rate caused by intra-prey competition as bifurcation parameter. 
All potential occurrences of bi-stability dynamics have also been investigated for some relevant sets of parametric values. Our numerical evaluations show that high levels of fear can stabilize the prey-predator system by ruling out the possibility of periodic solutions. Also, our model's Hopf bifurcation is subcritical in contrast to traditional prey-predator models, which ignore the cost of fear and have supercritical Hopf bifurcations in general. In contrast to the general trend, predator species go extinct at higher values of prey birth rates. We have also found that, contrary to the typical tendency for prey species to go extinct, both prey and predator populations may coexist in the system as intra-prey competition level grows noticeably. We have also been obtained that both prey and predator equilibrium densities increase (decrease) as the prey birth rate (fear level in prey) increases. The stability and Turing instability of associated spatial model have also been investigated analytically. We also perform the numerical simulation to observe the effect of different parameters on the density distribution of species. Different types of spatiotemporal patterns like spot, mixture of spots and stripes have been observed via variation of time evolution, diffusion coefficient of predator population, level of fear factor and prey refuge. The fear level parameter (k) has a great impact on the spatial dynamics of model system.
\end{abstract}
\pagebreak
\pagebreak[1]
\section{Introduction}

Traditionally, predators are thought to directly kill prey, reducing prey survival and numbers, and that this is the end of their ecological role in prey-predator interactions. However, numerous elegant experiments   
\cite{nelson2004predators,travers2010indirect,eggers2006predation,zanette2011perceived}  have demonstrated that the role of predators in ecology is more profound than just the direct killing of prey, as it may cause behavioral, psychological, and neurobiological changes in the prey population, ultimately affecting prey demography. In reality, when there is a confrontation between prey and predator, the prey perceives a predation risk. This perceived predation risk leads to fear or stress in the prey. Until relatively recently, it has been considered that predator-induced stress (fear) is necessarily acute, transitory, and does not affect the demographical process much. However, many experiments were conducted at the beginning of the $1990$s to determine whether predator-induced fear in response to predator could truly affect prey population even when direct killing was purposefully avoided \cite{lima1998nonlethal,preisser2005scared}. In such studies, direct killing of prey by predators has often been purposefully prevented by completely shutting (e.g. \cite{peckarsky1996alternative,schmitz1997behaviorally}) or partially dissecting (e.g. \cite{nelson2004predators}) the predator population's mouths (e.g. \textit{spiders, stoneflies}, and \textit{damsel bugs}). These risk predators \cite{schmitz1997behaviorally} are then forcibly enclosed in artificial enclosures with their prey (such as \textit{grasshopper nymphs, mayfly larvae}, and \textit{pea aphids}), as they can not directly kill but only \lq intimidate' any prey \cite{preisser2005scared}. The influence of intimidation alone on prey demography is then measured by comparing prey kept in enclosures with \lq risk' predators to prey housed by themselves. 

In some other investigations, prey have been \lq intimidated' using caged predators or the smell of predators \cite{kats1998scent}. In a meta analysis of such studies, Preisser, Bolnick, and Benard \cite{preisser2005scared} found that, on average, only predator intimidation has an equal or greater impact on prey birth rate, survival, and growth than direct death. Travers et al.\cite{travers2010indirect} conducted an experiment to show that female song sparrows who frequently had their nests preyed upon produced fewer eggs in subsequent clutches and, as a result, had a lower birth rate than song sparrows who didn't.
Eggers et al.\cite{eggers2006predation} also showed that \textit{Siberian jays} produced fewer eggs in their first clutch of the season as a result of predator call playbacks only.
These studies demonstrate that prey reproductive success and long-term survival may be negatively impacted by exposure to predators or predator cues. However, none of these studies quantified the impact of risk predators (fear) on prey demography. However, in 2011, Zanette et al. \cite{zanette2011perceived} did revolutionary work by performing a prolific field experiment in which direct killing was actively prevented to quantify the impact of risk predators or predator's fear on the prey population. In this study, it was claimed that \lq intimidation' (perceived predation risk) alone, which was controlled by predator call playbacks, caused a ~$40\%$~ reduction in the number of song sparrow chicks produced annually. Female song sparrows who were afraid of predators produced fewer eggs, more of which failed to hatch due to interruptions in incubation, and more nestlings who starved to death as a result of intimidation hampering their foraging activities. These experiments show that as a consequence of the fear of predators (predator induced stress), fecundity and survival may be reduced, and the total reduction in the prey number because of exposure to predators may thus far exceed that due to direct killing alone. The most convincing reason to believe that fear will affect prey populations universally is that fearful prey eat less because they can't have their heads down to forage and their heads up for vigilance of predators at the same time\cite{zanette2019ecology, clinchy2013predator,zanette2011perceived, tripathi2021widespread}.

The fear may result in the process of avoiding predators for which prey shows some anti-predator responses. A simple but effective anti-predator response involves the use of refuges. Refuge is a concept in which prey protects themselves from predation by hiding in an area or being in a situation where it is less accessible or difficult to find. Spatial refuges (burrows, dense foliage), temporal refuges (temporarily migrating to a low-risk habitat), forming groups, or simply reducing foraging activity are examples of refuges that reduce the confrontation of prey with predators\cite{sih1987prey,bailey1962interaction}. Generally, prey uses spatial refuges as these are easily available because of environmental heterogeneity. By using refuges, the prey population is not entirely exposed to the predators. A certain proportion of the prey population is always shielded from predators.
The benefit of refuge is that it reduces the predation risk. However, the longer the prey stays in the refuge, the fewer feeding and mating opportunities it has. For example, if the prey uses the refuge in the form of a group formation, this can lead to enhanced competition among the prey population \cite{bertram1978living}; where the refuge is due to reduced foraging activity, this probably also cuts down on the prey’s encounter with feeding opportunity; and where the prey uses the refuge in the form of space or time, it can cause a shift away from high-energy feeding habitats or activity times. As a result, spending more time in refuges not only promotes survival, but it also reduces prey feeding rates. There are two consequences of the cost of refuge use in lowering prey feeding rate: it reduces the possibility that prey will escape predator control and works against any stabilising effects of refuge utilisation. The presence of refuges has a substantial impact on the coexistence of predator and prey in the system\cite{sarwardi2012analysis}. The only two types of prey refuge have been considered in extant models; (\romannumeral 1) a constant number of the prey population in refuges\cite{maynard1974models,st1970mathematics,murdoch1975predation,vance1978predation} (\romannumeral 2) a constant proportion \cite{leslie1960properties,bailey1962interaction,hassell1973stability,maynard1974models,murdoch1975predation} of the prey population in refuges. The investigation of these refuges has persisted in an elementary stage because incorporating these \lq pseudo-refuges' into the prey-predator models has proven fairly challenging. It has been shown that use of refuge enhances the coexistence of prey and predator systems \cite{connell1972community}. 

Other than using refuge, a number of anti-predator responses, including toxin release, forming groups, playing dead, mimicry, herd behaviour, camouflage, and apostatic selection are also triggered in prey by fear of encroaching predators\cite{upadhyay2008chaotic,upadhyay2019global,kim2022prey}. Prey grouping is well established as an anti-predator behaviour in many prey species\cite{creel2014effects}. Additionally, some researchers have looked into how group size changes in response to predator risk \cite{boland2003experimental,caro2004adaptive}. Predation is decreased or even avoided as a result of group defence by prey against predators due to the greater inclination of the prey to better defend or hide themselves when their density is high enough. Predators can sometimes find larger groups easily, though. It is commonly acknowledged that enhancing a prey's protection against a predator is costly and causes a decrease in the amount of food intake of the prey\cite{sasmal2020dynamics}. When the prey is aware of the threat, it displays a particular sensitivity to the predator for physiological reasons. As their sensitivity to predators develops, prey often decreases the amount of time they spend outside and strengthens their group defence, which decreases their chances of becoming prey. In other words, a prey individual has a higher probability of surviving due to their greater sensitivity to the predator. The prey will reduce their reproduction rate as a result of putting more time and effort into group defence. As a result, the prey's sensitivity to the predator may promote the protection of the individual prey while also hastening the reduction in predator population density. Therefore, many prey (such as \textit{ungulates, ungulata}) form huge groups in order to increase their chances of survival through \lq safety in number'  and lower the cost of individual defence by using \lq many eyes'  effects\cite{blank2018vigilance}. The selection process will favour those who strike the optimal balance between risk reduction's benefits and costs. 

The question now is how to incorporate the increased instant survival of prey and the cost due to group defence simultaneously through functional response. All the initial functional responses Holling \RomanNumeralCaps{1}, \RomanNumeralCaps{2}, and \RomanNumeralCaps{3} \cite{holling1965functional}, are monotonically increasing, which means that more the prey is in the environment, the better for the predators\cite{cui2016complex}. However, this cannot be true for functional responses that represent group defense. 
Numerous studies and observations also show that monotonicity is not always true and that when the number of prey is high, a non-monotonic response could emerge\cite{boon1962kinetics,edwards1970influence}. Therefore, taking care above critical issues, Andrews\cite{andrews1968mathematical} proposed a Holling type \RomanNumeralCaps{4} functional response $\left(i.e.~~ p (x) =\frac{\beta x}{a+b x+x^2}\right)$, which is the most frequent and straightforward technique to model the group defence behaviour of prey. 
When prey density is high, the rate of predation per predator $p (x)$ decreases until it reaches zero. Therefore, when functional response of Holling type \RomanNumeralCaps{4} is taken as group defence function, predators may not survive above a certain prey density threshold because of the increased prey group defence. Another facet of group defence is that when prey populations congregate into herds, the weakest individuals are allowed to inhabit the centre of the herd, leaving healthier and stronger individuals on the herd's periphery. In any event, it's crucial that the more individualistic population only interacts with the more socialised one at the herd's perimeter. When attacks occur, especially in prey-predator or competitive models, it is typically the individuals at the herd's edge who suffer the consequences of the predator's deeds, which is reflected by the square root term\cite{ajraldi2011modeling}. In present study, we take into account the group defence of prey in the Holling type \RomanNumeralCaps{4} functional response. Freedman and Wolkowicz \cite{freedman1986predator} were the first to present a mathematical model of group defence. 

In~$2016$,Wang et al. \cite{wang2016modelling} analysed the prey-predator interaction, incorporating the cost of fear for the first time, and opening a wide scope in population ecology. He demonstrated that a high level of fear can have a stabilizing impact on the prey-predator system with Holling type \RomanNumeralCaps{2} by eliminating the periodic solution. The author of \cite{xiao2019stability} showed that mutual interference can stabilise the system by altering the extinction of predators after incorporating it into the Wang et al. model\cite{wang2016modelling}. According to study by Kar et al.\cite{kar2005stability} on a prey-predator model that incorporates a Holling-type-\RomanNumeralCaps{2} functional response and prey refuge, an increase in refuge may cause population outbreaks. In \cite{zhang2019impact}, the author extended the model of Kar et al.\cite{kar2005stability} by incorporating the fear factor into the model and found that fear can enhance system's stability by eliminating the occurrence of limit cycle oscillation or periodic solution. In \cite{xie2022influence}, the author analysed the prey-predator model, incorporating the fear effect and the prey refuge along with the Holling type \RomanNumeralCaps{3} functional response. It was discovered that the fear effect can stabilise the system by removing periodic solutions and reducing the final number of predator species at the coexist equilibrium, but it cannot cause predator extinction. Yuxin dong et al.\cite{dong2022influence} demonstrated that even when there is a high level of fear, the prey population remains stable rather than extinction. The high level of predator taxis sensitivity thwarts predation. Sasmal and Takeuchi\cite{sasmal2020dynamics} studied the Hopf bifurcation and multi-stability in a prey-predator system, incorporating group defence and the fear effect. According to Sasmal, the high predator-taxis sensitivity may cause instability in the system and the extinction of predator populations.

On the basis of the assumption that predators and preys move randomly in the habitat is modeled as diffusive prey-predator model. In reality, the spatial movement describes that how predator pursuing prey and in what way prey escapes from predator \cite{wu2018dynamics}. Biodiversity of the environment is governed by the spatio-temporal dynamics and Pattern formation presumes as an attractive proposition of the environment \cite{jana2020self}. In ecosystem, it the general tendency of prey is that it tries to stay away from the predators and therefore the escape velocity of prey may be considered proportional to the diffusive velocity of the predator population \cite{chakraborty2019complexity}. The significance of spatio-temporal models is
one of the leading issues in both ecology and mathematical ecology\cite{guin2021pattern}. The pattern formation in reaction–diffusion system has long been the subject of interest to the researchers in the domain of mathematical ecology because of its universal existence and importance \cite{guin2016existence}. Spatiotemporal dynamics of a prey-predator model was studied by Lui and Kang \cite{liu2022spatiotemporal}  to investigate the impact of fear effect and found that with increasing the cost of fear, the density of predator population decreases at the positive steady state.  The effect of fear factor on system dynamics of diffusive prey-predator model was investigated by Chen et al. \cite{chen2019nonexistence}  and authors suggested the proper range of parameters at which the spatiotemporal pattern can occur. Han et al. \cite{han2020cross} found that the fear factor has great influence on spatially inhomogeneous distribution of the two species under certain cross-diffusivity in a diffusive modified Leslie-Gower prey-predator model.\\

Wang et al.  \cite{wang2016modelling} and Zhang et al. \cite{xie2022influence}, in their studies, have considered Holling type \RomanNumeralCaps{2} and Holling type \RomanNumeralCaps{3} functional responses, respectively. They obtained that a high level of fear has a stabilising effect on the system by removing the periodic solutions. As a result, it could be interesting to investigate whether the fear effect still has a stabilizing effect on the system in the presence of a Holling type \RomanNumeralCaps{4} functional response (group defence). In this paper, we study a  spatial and its associated non spatial fear effect induced prey-predator interaction system in the presence of prey refuge and group defence. Our study objectives include:
\begin{itemize}
    \item What are the dynamic repercussions of fear effect in presence of group defence?
    \item In the presence of group defence, how the prey birth rate and intra-prey competition affect the dynamics of the model system ?
    \item To investigate different type of bifurcations (e.g., saddle-node, transcritical, Hopf) and stability results
    \item What is the impact of prey refuge and fear level parameters on the spatial prey-predator system?
\end{itemize}

The paper is organised in the following manner: The construction of our model is covered in Section \ref{Sec.2}, which explains how fear and group defence are related through predator-taxis sensitivity while incorporating fear effect, refuge, and non-monotone group defence in the model. Section \ref{Sec.3} provides a comprehensive equilibrium analysis and associated stability analysis. In this section, we also go through a few additional aspects of the suggested model, such as species persistence and extinction. Using numerical simulation and one- and two-parameter bifurcation diagrams, in Section \ref{Sec.4}, we discuss all possible local bifurcations, including transcritical, saddle-node, and Hopf bifurcations. Section \ref{Sec.5} discusses the effect of prey birth rate and fear level on prey and predator equilibrium densities. In section, \ref{Sec.6} spatial model system and its dynamics has been discusses.  In section \ref{Sec.7}, we numerically verify the conclusions of our analysis. The results and conclusions are discussed in  Section \ref{Sec.8}.
\section{Model Formulation:}\label{Sec.2}
Assume that we have a closed habitat with two species, namely prey and predator with population densities of $x$ and $y$, respectively. We consider intraspecies competition in both prey and predator as it is a very common phenomenon in real life and bound to happen. The following ordinary differential equation can be used to describe the dynamics of prey populations in the absence of predator species:
\begin{equation}
\dfrac{d x}{d t} = r x -d_1 x - d_2 x^2\label{1},
\end{equation}
where $r$ is the prey's birth rate, $d_1$ is prey's natural death rate, and $d_2$ is prey's death rate caused by intra-prey competition. From  Eq.\ref{1} It is straightforward to demonstrate that $x(t)$ tends to $0$ whenever $r<d_1$ for all positive initial conditions, i.e. the prey population  goes to extinct under this assumption. Therefore, we suppose that $r>d_1$.

In the presence of predator species, the population of prey species decreases due to predation by predators. When predators graze on prey species, the Holling type \RomanNumeralCaps{4} functional response  $\left(i.e.~~ p (x) =\frac{\beta x}{a+b x+x^2}\right)$ serves as the group defence mechanism or anti-predation response. Thus, the following system of differential equations  describes the dynamics of the prey and predator interaction:
\begin{equation}\label{2}
    \begin{aligned}
    \dfrac{d x}{d t} &= r x -d_1 x - d_2 x^2 -p\left(x\right) y,\\
     \dfrac{d y}{d t} &=  -d y -e y^2 +c p\left(x\right) y,
\end{aligned}
\end{equation}
where, $d$ is predator's natural death rate, $e$ is predator's death rate caused by intra-predator competition, $c$ is the conversion efficiency of prey biomass, $p(x)$ is the group defense function, $\beta$ is the predation rate, $a$ is the half saturation constant, and $b$ is the tolerance limit (defence level) of predator.

As previously stated, predators not only kill but also induce fear in the prey population, which can be more lethal than direct killing in the sense of a reduction in the prey numbers. Therefore, in order to model this reduction in prey numbers, we introduce a fear factor $f\left(k,y \right)$, which represents the cost of fear on prey's birth rate. The fear of prey also results in antipredator strategies such as vigilance, group defense, etc. But we also know that the time and energy of a system are limited and constant, which means if prey spends more time in group defence or other antipredator strategies, then it may lead to a lower production rate of prey due to reduced time for production. Because of this, it's crucial to interconnect these two aspects (fear factor and group defense) using a term called ``predator taxis sensitivity'' (which represents the sensitiveness of prey to predation) and it is termed as `$ \alpha$'. If ``predator-taxis sensitivity'' increases, then group defence among prey species will increase, which eventually decreases the predation rate. Also, as ``predator-taxis sensitivity'' increases, the value of the fear factor decreases as more sensitive prey will be more fearful (because their level of fear will increase) and that eventually results in a reduction in reproduction rate. In our model, we assume that both the tolerance limit of predators `b' and the fear level `k' increase linearly with respect to `$ \alpha$' $~(i.e.,~b\to b\alpha,k\to k\alpha)$. therefore, the fear factor mentioned above ~$\left (i.e.~f\left(k,y \right)\right)$~is now modified into ~$f\left(k,\alpha,y \right)$~and the group defence function ~$p\left(x\right)$~is modified in to ~$p \left(\alpha,x\right)$, where
\begin{align}
p \left(\alpha,x\right)=\frac{\beta x}{a+b \alpha x+x^2}\label{4}. 
\end{align}
Ecologically, as Zanette \cite{zanette2011perceived }and some other field experiments \cite{eggers2006predation, travers2010indirect} showed that fear causes a reduction in production rate, the fear factor and group defence function should satisfy the following conditions:\\
\begin{enumerate}[label=(\roman*)]
    \item $f\left( 0,\alpha,y \right)=1$: in the absence of fear of predation, the prey population grows at its birth rate;
    \item $f\left( k,\alpha,0 \right)=1$: in the absence of predators, there is no fear of predation in the prey population. Therefore, the fear factor does not affect the birth rate of prey individuals;
    \item $f\left( k,0,y \right)=1$: when prey species are not sensitive enough towards predation, the birth rate of prey individuals remains unaffected by the fear factor because they don't feel any kind of fear from predators;
    \item $\lim\limits_{k\to\infty}f\left( k,\alpha,y \right)=0$: when the level of fear is very high, prey species show high antipredator behaviour that leads the prey's birth rate towards zero;
    \item $\lim\limits_{y\to\infty}f\left( k,\alpha,y \right)=0$: a very large predator population size triggers a high level of fear in the prey population and leads to a zero prey birth rate;
    \item $\lim\limits_{\alpha\to\infty}f\left( k,\alpha,y \right)=0$: highly sensitive prey population becomes more fearful, and that leads to the fear factor approaching zero, which ultimately makes the prey birth rate zero;
    \item $\frac{\partial f\left( k,\alpha,y \right)}{\partial k} < 0$: the rate of prey production decreases as the level of fear or antipredator behaviour increases;
    \item $\frac{\partial f\left( k,\alpha,y \right)}{\partial y} < 0$: prey production rate decreases as predator population size increases;
    \item $\frac{\partial f\left( k,\alpha,y \right)}{\partial \alpha} < 0$: as the predator taxis-sensitivity increases, the prey birth rate decreases;
    \item $\lim\limits_{\alpha\to\infty}p\left( \alpha,x \right)=0$: when prey species are highly sensitive towards predation, their group defence becomes stronger and that leads to a zero predation rate;
    \item $\frac{\partial p(x,\alpha)}{\partial \alpha}< 0$: when the prey population becomes highly sensitive, they make their group defence a lot stronger, which leads to a reduction in predation rates.
\end{enumerate}
Therefore $f\left(k,\alpha,y \right)$ can  be taken as $\frac{1}{1+k\alpha y}~or~e^{-k\alpha y}$ or some other function. Some scholars and authors \cite{ sasmal2020dynamics}, in particular, have used the following form for the fear factor term:
\begin{align} \label{5}
f\left(k,\alpha,y \right)=\frac{1}{1+k \alpha y}.
\end{align}
Based on the preceding discussions, after incorporating the fear factor from Eq. \eqref{5} in Eq. \eqref{2}, our model becomes
\begin{equation}\label{6}
    \begin{aligned}
     \dfrac{d x}{d t} &= r x \left(f\left(k,\alpha,y \right)\right) -d_1 x - d_2 x^2 -p\left(\alpha,x\right) y,\\
     \dfrac{d y}{d t} &=  -d y -e y^2 +c p\left(\alpha,x\right) y.
\end{aligned}
\end{equation}
Now, in order to be more realistic, it is essential to take into account the system's prey refuge. In reality, not all the prey population is exposed to predators as they often have refuges where they can avoid predators. As a result, we assume that a portion $m x $ (where $m$ ($0\leq m < 1$) is the refuge parameter) of prey species is completely protected from predators (via refuges) and that only a portion $\left(1-m\right)x$ of prey species is available to predators.  By incorporating the refuge $m x$~and~$p \left(\alpha,x\right)$~from Eq. \eqref{4} into our system, we obtain:
\begin{equation}\label{7}
   \begin{aligned}
    \dfrac{d x}{d t} &=\frac{r x}{1+k\alpha y}   -d_1 x - d_2 x^2 -\frac{\left(-m+1\right) \beta x y}{a+  \left(-m+1\right) \alpha b x+\left(-m+1\right)^2x^2}~ , \\
     \dfrac{d y}{d t} &=  -d y -e y^2 +\frac{\left(-m+1\right) \beta  c x y}{a+  \left(-m+1\right) \alpha b x+\left(-m+1\right)^2x^2}.
\end{aligned} 
\end{equation}
Table \ref{tab 1} represents the system parameters and their biological meaning. Two sets of values for these parameters have also been provided, which we have used throughout this paper.
\begin{table}[H] 
\begin{center}
\begin{tabular}{|c|c|c|c|c|}
\hline
Parameters & Biological meaning & values (set 1) &  values (set 2)\\
\hline
$r$ &  prey's birth rate &1.1 & 1.1\\
\hline
$d_1$ & prey's natural death rate & 0.05 & 0.2\\
\hline
$d_2$ &  prey's death rate caused by intra-prey competition & 0.05 & 0.1  \\
\hline
$\beta$ & predation rate & 0.071 & 0.241\\
\hline
$m$ & prey refuge parameter & 0.3 & 0.3\\
\hline
$k$ & the level of fear &0.5 & 8\\
\hline
$\alpha$ & predator-taxis sensitive & 1 & 0.25 \\
\hline
$d$ & predator's natural death & 0.05 & 0.8 \\
\hline
$e$ & predator's death rate caused by intra-predator competition & 0.05 & 0.5\\
\hline
$c$ & conversion efficiency of prey biomass & 10 & 17\\
\hline
$a$ & the half saturation constant  & 0.1 & 0.9\\
\hline
$b$ & the tolerance limit of predators& 0.125 & 0.5 \\
\hline
\end{tabular}
\caption{Biological meaning of model parameters for system (\ref{7})} \label{tab 1}
\end{center}
\end{table}
\section{Mathematical Analysis}\label{Sec.3}
In this section we shall discuss the positivity and boundedness of the solutions of the model system (\ref{7}). We shall also discuss the persistence and extinction scenarios, local stability and bifurcation analysis.
\subsection{Positivity and Boundedeness }
 It is important to show the positivity of solutions in the system as it implies that the population survives, and once these solutions enter the positive quadrant, they will remain in it forever. So the solutions to the system will be biologically feasible if they lie in the positive quadrant. The boundedness of the system variables defining populations is vital to show. It can be seen as a natural limit to the growth of these system variables as a result of limited resources. Because of this, the model's boundedness shows that it is biologically appropriate.

\begin{theorem}\label{thm.1}
For the proposed model system (\ref{7}), $\mathbb{R}_+^2$ is positively invariant. Furthermore, any solution of the model system (\ref{7}) is bounded within this region $\mathbb{R}_+^2$ and satisfies the following condition:\\
 (\romannumeral 1)~$\limsup\limits_{t\to+\infty}x(t)\leq \left(\frac{r-d_1}{d_2}\right)$; provided $r>d_1$,~~~~~
 (\romannumeral 2)~$\limsup\limits_{t\to+\infty}y(t)\leq\left(\frac{c\beta-b d\alpha}{b\alpha e}\right)$; provided $c\beta>b d\alpha$.
\end{theorem}
\begin{proof}
    For the proof, refer to Appendix \ref{appA1}.
\end{proof}

\subsection{Extinction}
In biology, extinction refers to the complete annihilation of a species. Extinction happens when a species population is reduced due to environmental factors (habitat fission, natural disasters, over-exploitatage of species) or evolutionary changes in its members (genetic inbreeding, poor reproduction). In this section, we'll discuss about certain parametric criteria that, over a long span of time, guarantee the extinction of a species.

Let~~~~~~~ $\overline{x}$=$\limsup\limits_{t\to\infty}x(t)$,~~~~~~~ $\overline{y}$=$\limsup\limits_{t\to\infty}y(t)$,~~~~~~~ $\underline{x}$=$\liminf\limits_{t\to\infty}x(t)$~~~and~~~~~~~ $\underline{y}$=$\liminf\limits_{t\to\infty}y(t).$
 
  \begin{theorem}\label{thm.extinction}
  \begin{enumerate}[label=(\roman*)]
\item (\textbf{Extinction of prey in the absence of predators}) If $r<d_1, then\lim\limits_{t\to\infty}x\left(t\right)=0.$\label{thm.3}
\item (\textbf{Extinction of predator due to prey population}) If $a d d_2 >\left(-m+1\right) \beta c \left(r-d_1\right)$,~then $\lim\limits_{t\to \infty}y(t)=0.$\label{thm.6}
\end{enumerate}
  \end{theorem}
  \begin{proof}
   \begin{enumerate}[label=(\roman*)]
\item If possible, let $\lim\limits_{t\to\infty}x(t)=L>0$, then using the first equation of the model system (\ref{7}), we get
   \begin{flalign*}
  \dfrac{d x}{d t}& \leq \frac{r x}{1+k\alpha y} -d_1 x -d_2 x^2
   \leq r x - d_1 x -d_2 x^2
    \leq x \left(r- d_1\right),   \\
    x(t)& = x(0) exp\left( \int_{0}^{t}\left(r-d_1\right)\,d s\right).\\
   \end{flalign*}
    Now if $~r<d_1$~then~$\left(r-d_1\right)<0$, hence
    \begin{flalign*}
       x(t)& = x(0) exp\left( \int_{0}^{t}\left(r-d_1\right)\,d s\right)\to 0 ~ ~as ~~t \to \infty ,
    \end{flalign*}
  which is a contradiction to our assumption. Hence, our theorem is proved.
  
\textbf{Remark:}
According to this theorem, in the absence of predators, the prey population goes extinct since natural deaths are sufficient to eliminate the prey population from the system. In the considered model system, the growth of the predator population entirely depends on the prey population. Therefore, when the prey's natural death rate exceeds the birth rate of the prey, both the prey and predator populations are wiped out from the system. Henceforth, we consider that the prey birth rate is greater than their natural death rate $(r>d_1)$ unless otherwise stated.

\item From the second equation of the model system(\ref{7}), we have
\begin{align*}
\frac{d y}{d t} & =-d y-e y^2+\frac{(-m+1) \beta c y}{a}\left(x-\frac{(-m+1) \alpha b x^2+(-m+1)^2x^3}{a+(-m+1) \alpha b x+(-m+1)^2x^2}\right),\\
& \le y\left(-d-e y+\frac{(-m+1) \beta c x}{a}\right)
 \le y\left(-d+\frac{(-m+1) \beta c (\frac{r-d_1}{d_2})}{a}\right),\\
\int_{0}^{t} \ \frac{d y}{y}&\leq \int_{0}^{t}\left(-d +   \frac{\left(-m+1 \right) \beta c \left(\frac{r-d_1}{d_2}\right) }{a}\right)\,d s,\\
y\left(t\right)&\leq y\left(0\right)\exp \left(\int_{0}^{t}\left(-d +   \frac{\left(-m+1 \right) \beta c \left(\frac{r-d_1}{d_2}\right) }{a}\right)\,d s\right).\\
\end{align*}
So, if ~$a d d_2 >\left(-m+1\right) \beta c  \left(r-d_1\right)$, then $y\left(t\right) \to 0~~as~~t \to 0$.\\
From the second equation of the model system(\ref{7}), we have
\begin{align*}
  \dfrac{d y}{d t} &\leq  -d y  +\frac{\left(-m+1\right) \beta  c x y}{a+  \left(-m+1\right) \alpha b x} = y \left(\frac{-d\left(a+\left(m-1\right)\alpha b x\right) + \beta c x \left(1-m\right)}{a + (1-m) \alpha b x}\right),\\
 \dfrac{d y}{d t} &\leq y \left(\frac{-d a + (1-m) x (c \beta -d \alpha b)}{a + (1-m) \alpha b x} \right).
  \end{align*}
  So, if ~$c \beta < d \alpha b$, then $y\left(t\right) \to 0$ as $t \to 0$.\\
\end{enumerate}
\end{proof}
\subsection{Persistence of Species}
A species is ecologically persistent if it is never in danger of extinction or is never entirely eradicated from the system. The ability of a species to persist means that it can withstand exogenous perturbations. As a result, it is a significant system behaviour in the long run. Analytically, we can define the persistence of species through the following theorem:
\begin{theorem}\label{thm.7}
\begin{enumerate}[label=(\roman*)]
 \item The prey species is strongly persistent under the following condition:\label{thm.prey persistence}\\
$\frac{rbe}{be +k \left(c\beta-bd\alpha\right)}
> d_1 + \frac{\beta \left(-m+1\right)\left(c\beta-bd\alpha\right)}{abe \alpha}.$~~
\item The predator species is strongly persistent under the following conditions:\label{thm.predator persistence}\\
$r>d_1$ and $c\beta \left(1-m\right) \left(\frac{r-d_1}{d_2}\right)> d \left(a + b \alpha \left(1-m \right) \left(\frac{r-d_1}{d_2}\right) + \left(1-m\right)^2 \left(\frac{r-d_1}{d_2}\right)^2\right).$

\end{enumerate}
\end{theorem}

\begin{proof}
 \begin{enumerate}[label=(\roman*)]
\item Using upper bound of $y$ from Theorem \ref{thm.1}~in the first equation of the model system(\ref{7}) to obtain
  \begin{align*}
  \frac{d x}{d t}&\geq x\left(\left(\frac{r}{1+k\alpha \left(\frac{c\beta-b d\alpha}{b\alpha e}\right)}\right)-d_1 - d_2 x - \frac{\left(-m+1 \right) \beta  \left(\frac{c\beta-bd\alpha}{b\alpha e}\right)}{a +  \left(-m+1 \right) \alpha  b x + \left(-m+1\right)^2 x^2}\right),\\
  \end{align*}
  since $\left(a +  \left(-m+1 \right) \alpha  b x + \left(-m+1\right)^2 x^2\right)\geq a$, then from the above expression we get:
  \begin{align*}
   \frac{d x}{d t}&\geq x\left(\left(\frac{r}{1+k\alpha \left(\frac{c\beta-b d \alpha}{b\alpha e}\right)}\right)-d_1 - d_2 x - \frac{\left(-m+1 \right) \beta  \left(\frac{c\beta-b d \alpha}{b\alpha e}\right)}{a}\right).\\
   \end{align*}
   Now by comparison lemma (\cite{chen2005nonlinear}), we obtain
   \begin{align*}
   \liminf\limits_{t\to+\infty}x(t)\geq \frac{r b e}{bed_2 +k d_2 \left(c\beta-b d\alpha\right)}-\frac{d_1}{d_2}-\frac{\beta \left(-m+1\right)\left(c\beta-b d\alpha\right)}{abed_2 \alpha};\\
provided :~~\frac{r b e}{be +k \left(c\beta-b d\alpha\right)}
> d_1 + \frac{\beta \left(-m+1\right)\left(c\beta-b d\alpha\right)}{a b e \alpha} . 
\end{align*}
\item Using upper bound of $x$ from Theorem \ref{thm.1}~in the first equation of the model system(\ref{7}) to obtain\\
 \begin{align*}
 \frac{d y}{d t}& \geq y\left(-d  -e y +\frac{c\beta \left(1-m\right) \left(\frac{r-d_1}{d_2}\right) }{a + b \alpha \left(1-m \right) \left(\frac{r-d_1}{d_2}\right) + \left(1-m\right)^2 \left(\frac{r-d_1}{d_2}\right)^2} \right).\\
 \end{align*}
   Now by comparison lemma (\cite{chen2005nonlinear}), we obtain
  \begin{flalign*}
   \liminf\limits_{t\to+\infty}y(t)\geq \left(\frac{-d}{e}   +\frac{c\beta \left(1-m\right) \left(\frac{r-d_1}{d_2}\right)}{e\left(a + b \alpha \left(1-m \right) \left(\frac{r-d_1}{d_2}\right) + \left(1-m\right)^2 \left(\frac{r-d_1}{d_2}\right)^2\right)} \right),\\
   provided:~~ r>d_1,\,\, c\beta \left(1-m\right) \left(\frac{r-d_1}{d_2}\right)> d \left(a + b \alpha \left(1-m \right) \left(\frac{r-d_1}{d_2}\right) + \left(1-m\right)^2 \left(\frac{r-d_1}{d_2}\right)^2\right).\\
  ~~ \end{flalign*}
\end{enumerate}
\end{proof}

\subsection{Existence of Equilibrium Points}\label{sec.3.4}
This subsection investigates the parametric conditions for the presence of several different feasible equilibrium points. By analysing the eigenvalues of the related Jacobian matrices, we also investigate the local stability of these equilibrium points. We consider the model system (\ref{7}) in order to determine feasible equilibrium points:
\begin{align*}
    \frac{d x}{d t} &=\frac{r x}{1+k\alpha y}   -d_1 x - d_2 x^2 -\frac{ \left(-m+1\right) \beta x y}{a+ \left(-m+1\right) \alpha b x+\left(-m+1\right)^2x^2}\equiv~g\left(x,y\right), \\
  \frac{d y}{d t} &=  -d y -e y^2 + \frac{\left(-m+1\right) c \beta x y}{a+  \left(-m+1\right) \alpha b x+\left(-m+1\right)^2x^2}~~~~~~~~~~~~~~~\equiv~h\left(x,y\right).
\end{align*}
 Ecologically, equilibrium points with non-negative co-ordinates only, i.e., the points of intersection of $x$ and $y$ growth nullclines in  $\mathbb{R}_+^2 \cup\left(0,0\right)$, are feasible. $\left(\ i.e., \left(x,y\right)\in\mathbb{R}_+^2 \cup\left(0,0\right)~:~g\left(x,y\right)=0,h\left(x,y\right)=0\right)$.
Therefore, the model system (\ref{7}) does have the following feasible equilibrium point:-
\begin{enumerate}[label=(\roman*)]
\item The population extinct equilibrium point $E_0(0,0)$, which  always  exists.
\item The predator extinct equilibrium point $E_1(\frac{r-d_1}{d_2},0),$ which exists if $r>d_1$.
\item  An interior (coexisting) equilibrium point $E^*=(x^*,y^*)$ exists in addition to the trivial and boundary equilibrium points if $d<\frac{c\beta x^* (1-m)}{a+b\alpha x^*(1-m)+(1-m)^2(x^*)^2}$. For  $E^*=(x^*,y^*)$, the predator species' component is $y^*=\frac{1}{e}\left(-d+\frac{c\beta x^* (1-m)}{a+b\alpha x^*(1-m)+(1-m)^2(x^*)^2}\right)$ and the prey component $x^* $is a positive root of the following equation: 
\begin{align}\label{8}
p(x)=A_1x^7+A_2x^6+A_3x^5+A_4x^4+A_5x^3+A_6x^2+A_7x+A_8.
\end{align}
\end{enumerate}
Appendix \ref{appB1} contains the coefficients $A_1$, $A_2$, $A_3$, $A_4$, $A_5$, $A_6$, $A_7$, and $A_8$.
It is difficult to find the analytical conditions to ascertain the actual number of interior (coexisting) equilibrium points of the system (\ref{7}). However, through the relative position and shape of coexisting nullclines, the possible number of coexisting equilibrium points can be determined. To find a different number of equilibrium points, we vary~$d_2$~and assign values to other model parameters from set (1) of the Table (\ref{tab 1}). As shown in the Fig.\ref{fig:1}, predator nullcline (red curve) intersects prey nullclines once, twice, and thrice for~$d_2=0.05$~(black curve),~$d_2=0.055$~(green curve), and~$d_2=0.09$~(cyan curve), respectively. So, the model system (\ref{7}) can have a maximum of three coexisting equilibrium points for the chosen set of parametric values. Table \ref{tab 2} shows the interior equilibrium points and corresponding local stability for different values of~$d_2$. 

\subsection{Local Stability Analysis}
\begin{theorem}\label{thm.8}
\begin{enumerate}[label=(\roman*)]
\item The population extinct equilibrium point $E_0(0,0)$ is locally asymptotically stable if $(-d_1 +r)<0$. In particular, under this condition, $E_0(0,0)$ is globally asymptotically stable.\label{thm. trivial stability}
\item The predator extinct equilibrium $E_1(\frac{r-d_1}{d_2},0)$ is locally asymptotically stable  if $r>d_1$, and $\left[d(ad_2^2 + \right .\\
\left.(-m+1) \alpha b (r-d_1)d_2+(-m+1)^2 (r-d_1)^2)> (-m+1)  \beta c (r-d_1)  d_2\right]$.\label{thm. boundary stability}
\item The interior equilibrium point $E^*(x^*,y^*)$ is locally asymptotically stable if $Tr(J(E^*)<0$ and $Det(J(E^*)>0$, unstable otherwise.
\end{enumerate}
\end{theorem}
\begin{figure}[H]
\centering
\subfloat[]{\includegraphics[height=9cm,width=12cm]{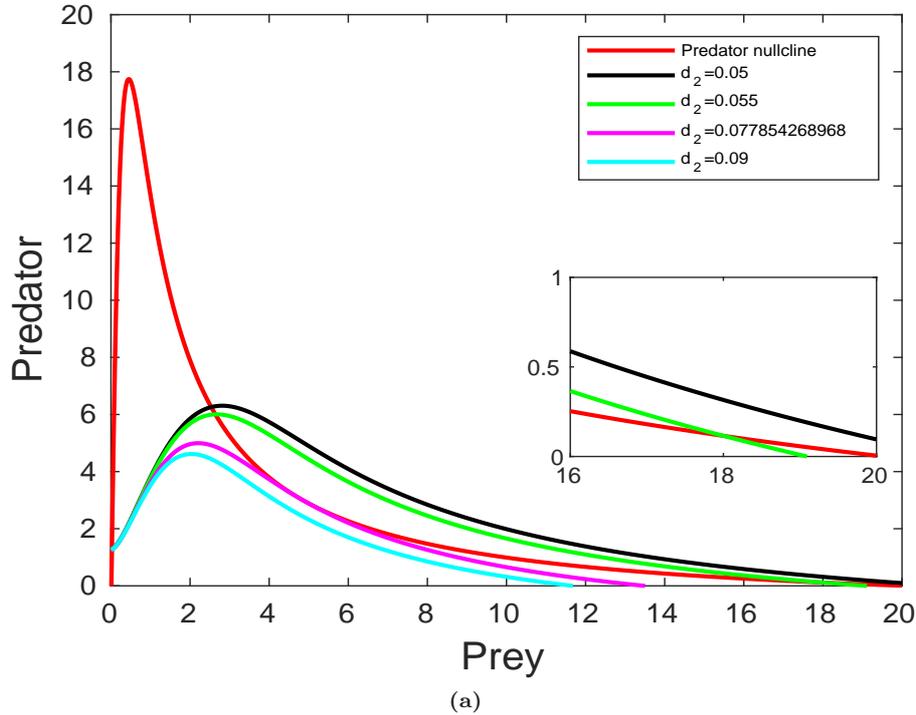}}
\caption{The red curve represents the predator nullcline for all values of $d_2$. The black curve, the green 
curve, the magenta curve and the cyan curve represent the prey nullcline for $d_2=0.05, d_2=0.055, d_2=0.077854268968, d_2=0.09$, respectively.}\label{fig:1}
\end{figure}

\begin{proof}
    \begin{enumerate}[label=(\roman*)] 
    \item For the proof, refer to Appendix \ref{appA3}.
    \item For the proof, refer to Appendix \ref{appA3}.
\item The Jacobian matrix of model system (\ref{7}) at $E^*(x^*,y^*)$ is given by 
\begin{equation}
J({E^*}) = {(J)_{{E^*}}} = \left( {\begin{array}{*{20}{c}}\label{9}
   {{g_x}} & {{g_y}}  \\
   {} & {}  \\
   {{h_x}} & {{h_y}}  \\
\end{array}} \right){|_{{E^*}}} = \left( {\begin{array}{*{20}{c}}
   {{j_{11}}} & {{j_{12}}}  \\
   {} & {}  \\
   {{j_{21}}} & {{j_{22}}}  \\
\end{array}} \right).
\end{equation}

    Therefore, the characteristic equation of $J(E^*)$ is
    \begin{equation*}
\lambda^2-Tr(J(E^*))\lambda+Det(J(E^*))=0,
  \end{equation*}
where, $Tr(J(E^*)) = j_{11} + j_{22} $ and $Det(J(E^*)) = j_{11}j_{22}-j_{12}j_{21}$.
Thus, $E^*(x^*,y^*)$ is locally asymptotically stable if~$Tr(J(E^*))<0$ and $Det(J(E^*))>0$.
 \end{enumerate}
\end{proof}
Finding the analytical conditions to determine the actual number of coexisting equilibrium points is difficult for the system \eqref{7}. Therefore, here in the Table \ref{tab 2}, we have numerically presented various numbers of equilibrium points and discussed their corresponding local stability for a chosen set of parametric values.
\begin{table}[H]
\scriptsize
\begin{center}  
\begin{tabular}{|c|c|c|c|c|c|c|} 
\hline
S.No. & $d_2$ & No. of &&&&\\&&E.P & E.P(s) &$\lambda_1$ & $\lambda_2$& Stability\\
\hline
\multirow{2}{*}{1} & \multirow{2}{*}{0.05} & \multirow{2}{*}{2} & $E_{2}^*=(2.5356379, 6.2586513)$ & -0.4883204
  & 0.2169094& S.P \\
\cline{4-7}& & & $E_{3}^*=(0.0234372, 1.2768451)$ & -0.024659~+~$\iota$~ 0.310341&-0.024659~-~$\iota$~0.310341 & S.S\\
\hline
\multirow{3}{*}{2} & \multirow{3}{*}{0.055} & \multirow{3}{*}{3} & $E_{1}^*=(18.0018467, 0.1151048)$& -0.9931931&  -0.0025924& S.N \\
\cline{4-7}& & & $E_{2}^*=(2.6405682, 6.0037036)$&-0.4820788& 0.1867656&S.P\\
\cline{4-7}& & & $E_{3}^*=(0.0234354, 1.2766738)$& -0.024715~+~$\iota$~ 0.310322&-0.024715~-~$\iota$~0.310322&S.S\\
\hline
\multirow{2}{*}{3} & \multirow{2}{*}{0.07785426896887} & \multirow{2}{*}{2} & $E_{12}^*=(4.78094, 3.05535)$ & -0.49974307&-0.00000226& S.N \\
\cline{4-7}& & & $E_{3}^*=(0.0234272, 1.2758912)$ &-0.024971+~$\iota$~0.3102333&-0.024971-~$\iota$~0.3102333 & S.S\\
\hline
4 & 0.09 & 1 & $E_{3}^*=(0.0234228, 1.2754756)$&  -0.025107+~$\iota$~0.3101861&-0.025107-~$\iota$~0.3101861& S.S\\
\hline
5 & $d_2=0.1,d_1=1.2$ & 0 & -&-&-&-\\
\hline
\end{tabular} 
\caption{ Other than $d_2$, all other parametric values are taken from the set (1) of the Table \ref{tab 1}. Interior equilibrium point(s) and corresponding local stability for different values of $d_2$. E.P, S.P, S.S, and S.N are abbreviations for equilibrium point, saddle point, stable spiral, and stable node, respectively.}\label{tab 2}
\end{center}
\end{table}
\section{Local bifurcation analysis}\label{Sec.4}
In this section, we shell discuss about how the different parameter values affect the stability of the different equilibrium points. In addition, we shell also investigate that how  the equilibrium points arise/ dissappear.
\subsection{Saddle-node bifurcation}
\begin{theorem}\label{thm.9}
The model system (\ref{7}) undergoes a saddle-node bifurcation around the interior equilibrium point $E^*(x^*,y^*)$ as the intra-prey competition parameter $d_2$ passes through the bifurcation value $d_2=(d_2)_{SN}$ if and only if
\begin{align*}
    \left[g_{xx}-\frac{g_x}{g_y}(g_{x y}+g_{y x})+\left(\frac{g_x}{g_y}\right)^2 g_{y y}-\left(\frac{g_x}{h_x}\right)\left(h_{xx}-\frac{g_x}{g_y}(h_{x y}+h_{y x})+\left(\frac{g_x}{g_y}\right)^2h_{y y}\right)\right] \Bigg|_{((d_2)_{SN}, E^*)} \neq 0.
\end{align*}
\end{theorem}
\begin{proof}
Let $d_2=(d_2)_{SN}$~and~$E^*(x^*, y^*)$~be the interior equilibrium point.
Define:
\begin{equation*}
    f(x,y)=\left(\begin{array}{c}
  g(x,y)   \\
  \\
h(x,y)
\end{array}\right)=\left(\begin{array}{c}
    \frac{r x}{1+k\alpha y}   -d_1 x - d_2 x^2 -\frac{\left(-m+1\right) \beta x y}{a+  \left(-m+1\right) \alpha b x+ x^2 \left(-m+1\right)^2 }  \\
    \\
 -d y -e y^2 - \frac{c \beta \left(-m+1\right)x y}{a+b\alpha  \left(-m+1\right)x+\left(-m+1\right)^2x^2}
\end{array}\right),
\end{equation*}
differentiating $f$ with respect to $d_2$, we obtain
\begin{align*}
    f_{d_2}(E^*)=\left(\begin{array}{c}
         g_{d_2 }\\
         \\
         h_{d_2 }  
    \end{array}\right)\Bigg|_{(x^*,y^*)}=\left(\begin{array}{c}
      -(x^*)^2 \\
      \\
        0
    \end{array}\right).
\end{align*}
Now we evaluate
\begin{align*}
    S=D f (E^*,(d_2)_{SN})=\left(\begin{array}{cc}
        g_x & g_y \\
        \\
        h_x & h_y
    \end{array}\right)\Bigg|_{(E^*,(d_2)_{SN}).}
\end{align*}
So, $V=\left(1 ~~\frac{-g_x}{g_y}\right)^T$~and~$W=\left(1 ~~\frac{-g_x}{h_x}\right)^T$~are the eigenvectors corresponding to eigenvalue $\lambda=0$ of $S$ and $S^T$, respectively. we also evaluate
\begin{align*}
    \left[D^2 f(E^*,(d_2)_{SN})(V,V)\right]=\left(\begin{array}{c}
      g_{xx}-\frac{g_x}{g_y}(g_{x y}+g_{y x})+\frac{g^2_x}{g^2_y} g_{y y}     \\
      \\
         h_{xx}-\frac{g_x}{g_y}(h_{x y}+h_{y x})+\frac{g^2_x}{g^2_y}h_{y y}
    \end{array}\right).
\end{align*}
The differentiations of functions \lq f' and \lq g' have been provided in Appendix \ref{appB2}.
Now using the Sotomayor theorem \cite{perko2013differential} for saddle-node bifurcation, we obtain
\begin{align*}
    W^T\left[f_{d_2}(E^*)\right]= -(x^*)^2 \neq 0,
\end{align*}
and
\begin{align*}
W^T\left[D^2 f(E^*,(d_2)_{SN})(V,V)\right]=&
\left(g_{xx}-\frac{g_x}{g_y}(g_{x y}+g_{y x})+\frac{g^2_x}{g^2_y} g_{y y}-\frac{g_x}{h_x}\left(h_{xx}-\frac{g_x}{g_y}(h_{x y}+h_{y x})+\right.\right.\\
    &\left.\left.\frac{g^2_x}{g^2_y}h_{y y}\right)\right)\Bigg|_{((d_2)_{SN}, E^*).}
\end{align*}
Therefore, system (\ref{7}) undergoes a saddle-node bifurcation at $E^*(x^*,y^*)$ if and only if
\begin{align*}
    \left[g_{xx}-\frac{g_x}{g_y}(g_{x y}+g_{y x})+\frac{g^2_x}{g^2_y} g_{y y}-\frac{g_x}{h_x}\left(h_{xx}-\frac{g_x}{g_y}(h_{x y}+h_{y x})+ \frac{g^2_x}{g^2_y}h_{y y}\right)\right]\Bigg|_{((d_2)_{SN}, E^*)}\neq 0.
\end{align*}
\end{proof}
Analytical demonstration of saddle-node bifurcation for interior equilibrium point is quite challenging. However, we provide a numerical example to verify the existence of saddle-node bifurcation for the system (\ref{7}). Here we verify the saddle-node bifurcation for $d_2=(d_2)_{SN}=0.0778542689688737$. The other parameter values are taken from set (1) of Table \ref{tab 1}. The interior equilibrium points $E^*_1$ and $E^*_2$ coincide at interior equilibrium point $E^*_{12}=(4.7809, 3.0553)$ and mutually annihilate via saddle-node bifurcation as shown in Fig. \ref{Fig.2}.
\begin{remark}
We also observe a saddle-node bifurcation for $r=(r)_{SN}=3.27311875453723710727160$ at interior equilibrium point $E^*_{12}=(1.247196, 2.438599)$ of the system (\ref{7}) where two interior equilibrium points $E^*_1$ and $E^*_2$ coincide and mutually annihilate as depicted in Fig. \ref{Fig.3}.
\end{remark}
\subsection{Transcritical bifurcation}
\begin{theorem}\label{thm.11}
The model system (\ref{7}) undergoes a transcritical bifurcation from the population extinct equilibrium point ~$E_0= (0,0)$~when ~$r=r_{TC}=d_1.$
\end{theorem}
\begin{proof}
The jacobian matrix at~$E_0= (0,0)$~and   ~$r=d_1$ is given by
$$J(E_0,r_{TC})=\left(
\begin{array}{cc}
 0 & 0 \\
 0 & -d \\
\end{array}
\right),$$\\
whose eigenvalues are $\lambda_1=0$ and $\lambda_2=-d$. Now we choose the eigenvectors V and W associated with eigen value ~$\lambda_1$~of ~$J(E_0,d_1)$~and ~$J(E_0,d_1)^T$~given respectively by 
$V=\left(\begin{array}{c}
  v_1   \\
v_2 
\end{array}\right)=\left(\begin{array}{c}
  1   \\
0
\end{array}\right)$~~~and~~~$W=\left(\begin{array}{c}
  w_1   \\
w_2 
\end{array}\right)=\left(\begin{array}{c}
  1   \\
0
\end{array}\right).$\\
Recall $f$ from the Theorem \ref{thm.9} and calculate: 
\begin{align*}
    \Delta_1&=W^T f_r \left(E_0,r=r_{TC}\right)=\left(\begin{array}{cc}
  1 & 0
\end{array}\right)\cdot\left(\begin{array}{c}
  \frac{\partial g}{\partial r}\\
    \\
\frac{\partial h}{\partial r}
\end{array}\right)_{\left({E_0};r_{TC}\right)}=\left(\begin{array}{cc}
 1 & 0
\end{array}\right)\cdot\left(\begin{array}{c}
  0  \\
0 
\end{array}\right)
=0,
\end{align*}
\begin{align*}
\Delta_2 =W^T\left[D f_r \left(E_0,r=r_{TC}\right) V \right]=&\left(\begin{array}{cc}
  1 & 0
\end{array}\right)\cdot\left(\begin{array}{cc}
  \frac{\partial^2 g}{\partial x \partial r} &\frac{\partial^2 g}{\partial y \partial r} \\
    \\
\frac{\partial^2 h}{\partial x \partial r} &\frac{\partial^2 h}{\partial y \partial r} 
\end{array}\right)_{\left(E_0;r_{TC}\right)}\cdot\left(\begin{array}{c}
  1   \\
0
\end{array}\right)  \\
=&\left(\begin{array}{cc}
  1 & 0
\end{array}\right)\cdot\left(\begin{array}{cc}
  1 & 0 \\
    \\
0 & 0 
\end{array}\right)\cdot\left(\begin{array}{c}
  1   \\
0 
\end{array}\right)=1 
\ne 0,
\end{align*}
\begin{align*}
\Delta_3 =W^T\left[D^2 f \left(E_0,r=r_{TC}\right) (V,V)\right]=&\left(\begin{array}{cc}
1 & 0
\end{array}\right)\left(\begin{array}{c}
            \frac{\partial^2g }{\partial x^2}v_1^2+\frac{2\partial^2g}{\partial x\partial y}v_1v_2 +\frac{\partial^2g }{\partial y^2}v_2^2 \\
            \frac{\partial^2h }{\partial x^2}v_1^2+\frac{2\partial^2h}{\partial x\partial y}v_1v_2 +\frac{\partial^2h }{\partial y^2}v_2^2
              \end{array}\right)\Bigg|_{(E_0,r_{TC})} \\
=&\left(\begin{array}{c}
      \frac{\partial^2g }{\partial x^2}v_1^2+\frac{2\partial^2g}{\partial x\partial y}v_1v_2 +\frac{\partial^2g }{\partial y^2}v_2^2  
     \end{array}\right)\Bigg|_{(E_0,r_{TC})}
    = \left(\begin{array}{c}
     \frac{\partial^2g }{\partial x^2}
    \end{array}\right)\Bigg|_{(E_0,r_{TC})}\\
=&-2d_2
\neq 0.
\end{align*}
Hence, all the transversality condition are satisfied, therefore the model system (\ref{7}) undergoes a transcritical bifurcation from population extinct equilibrium $E_0$ at $r=r_{TC}=d_1$.\\
\end{proof}
\begin{remark}
When we choose the parametric values from set (2) of Table \ref{tab 1} then a transcritical bifurcation for the system (\ref{7}) occurs for $r=(r)_{TC}=0.2$ at population extinct equilibrium point $E_0$ (refer to Fig. \ref{Fig.3}).
\end{remark}
\begin{theorem}\label{thm.13}
The model system (\ref{7}) undergoes a transcritical bifurcation from predator extinct equilibrium point
$E_1(\frac{r-d_1}{d_2},0)$ at $\beta=\beta_{TC}=\frac{d(ad_2^2+b\alpha(1-m)(r-d_1)d_2+(1-m)^2(r-d_1)^2)}{c(1-m)(r-d_1)d_2}$, if ($e c (r-d_1)^2) \neq S$, where $S=d((1-m)^2(r-d_1)^2-ad_2^2)(dd_2^2+c k r \alpha (r-d_1))$.
\end{theorem}
\begin{proof}
The jacobian matrix at  $E_1(\frac{r-d_1}{d_2},0)$ and $\beta=\beta_{TC}$,
$$J(E_1,\beta_{TC})=\left(
\begin{array}{cc}
 d_1-r & \frac{k r \alpha  \left(d_1-r\right)}{d_2}-\frac{d}{c} \\
 \\
 0 & 0 \\
\end{array}
\right),$$
whose eigenvalues are ~$\lambda_1=0$~and~$\lambda_2=-r+d_1$. Now we choose the eigenvectors V and W associated with eigen value ~$\lambda_1$ (zero eigen value)~of ~$J(E_1,\beta_{TC})$~and ~$J(E_1,\beta_{TC})^T$~given respectively by 
$V=\left(\begin{array}{c}
  v_1   \\
  \\
v_2 
\end{array}\right)=\left(\begin{array}{c}
  \frac{-c k r^2 \alpha + c k r \alpha d_1 - d d_2 }{c (r-d_1)d_2} \\
  \\
1
\end{array}\right)$~and~$W=\left(\begin{array}{c}
  w_1   \\
  \\
w_2 
\end{array}\right)=\left(\begin{array}{c}
  0  \\
  \\
1
\end{array}\right)$.\\
Recall $f$ from Theorem \ref{thm.9} and calculate:
\begin{equation*}
\begin{aligned}
    f_\beta({E_1,\beta_{TC}})&=\left(\begin{array}{c}
\frac{(1-m) x y}{a+ b \alpha (1-m) x +(1-m)^2 x^2}  \\
\\
\frac{c(1-m) x y}{a+ b \alpha (1-m) x +(1-m)^2 x^2} 
\end{array}\right)\Bigg|_{(E_1,\beta_{TC})},~
f_\beta({E_1,\beta_{TC}})=\left(\begin{array}{c}
  0  \\
0
\end{array}\right),\\
W^T[f_\beta({E_1,\beta_{TC}})]&=\left(\begin{array}{cc}
  0 & 1 
\end{array}\right)\cdot\left(\begin{array}{c}
  0 \\
  0
\end{array}\right)=0.
\end{aligned}
\end{equation*}
Now, we evaluate~$W^T[D f_\beta(E_1,\beta_{TC}) V]$ and find 
\begin{align*}
    W^T[D f_\beta(E_1,\beta_{TC}) V]
&=\frac{c d_2 (1-m)(r-d_1)}{(a d_2^2+(1-m)^2(r-d_1)^2+ d_2 b\alpha(1-m)(r-d_1))}\neq 0.
\end{align*}
Now, also we have
\begin{align*}
    W^T[D^2f(E_1,\beta_{TC})(V,V)]=
              &\left(\begin{array}{c}\frac{\partial^2h }{\partial x^2}v_1^2+\frac{2\partial^2h}{\partial x\partial y}v_1v_2 +\frac{\partial^2h}{\partial y^2}v_2^2\\
              \end{array}\right)\Bigg|_{(E_1,\beta_{TC})},
\end{align*}
where
\begin{equation*}
\begin{aligned}
\frac{\partial^2h }{\partial x^2}\Bigg|_{(E_1,\beta_{TC})}=&0,\,\,\,\,
\frac{\partial^2h }{\partial y^2}\Bigg|_{(E_1,\beta_{TC})}=-2e,\\ \frac{\partial^2h }{\partial x \partial y}\Bigg|_{(E_1,\beta_{TC})}=&\frac{-c(1-m)(\beta_{TC})d_2^2((m-1)^2(r-d_1)^2 -ad_2^2)}{((m-1)^2(r-d_1)^2+b\alpha d_2(1-m)(r-d_1)+ad_2^2)}\\
 =&\frac{-dd_2((m-1)^2(r-d_1)^2-ad_2^2)}{(r-d_1)}.
\end{aligned}
\end{equation*}
So, using the values of~$ \frac{\partial^2h }{\partial x^2}\Bigg|_{(E_1,\beta_{TC})}, \frac{\partial^2h }{\partial y^2}\Bigg|_{(E_1,\beta_{TC})}, \frac{\partial^2}{\partial x \partial y}\Bigg|_{(E_1,\beta_{TC})}$, $v_1$, and $v_2$, we get
\begin{align*}
    W^T[D^2f(E_1,\beta_{TC})(V,V)]=2\left(\frac{-e c (r-d_1)^2+d((1-m)^2(r-d_1)^2-ad_2^2)(dd_2^2+c k r \alpha (r-d_1))}{c(r-d_1)^2}\right)\neq 0.
\end{align*}
So, using the Sotomayor theorem \cite{perko2013differential}, the theorem (\ref{thm.13}) is proved. We provide a numerical example to verify the existence of transcritical bifurcation for the system (\ref{7}). We check the transcritical bifurcation for $d_2=(d_2)_{TC}=0.052246635114$, which corresponds to the transcritical bifurcation threshold $\beta=\beta_{TC}$. All other parameters are taken from set (1) of Table \ref{tab 1}. The predator extinct equilibrium point, $E_1$, exchanges its stability with interior equilibrium points, $E^*_1$, through transcritical bifurcation as depicted in Fig. \ref{Fig.2}.
\end{proof}
\begin{remark}
 Two transcritical bifurcations for the system (\ref{7}) occur at $r=(r)_{{TC}_1}= 0.226735001383921$ and $r=(r)_{{TC}_2}=0.88701499861607$ (refer to Fig. \ref{Fig.3}), which correspond to the transcritical bifurcation threshold $\beta=\beta_{TC}$. All other parameters from set (2) of Table (\ref{tab 1}) .
\end{remark}
\subsection{Hopf bifurcation}

\begin{theorem}\label{thm.15}
At the bifurcation threshold $k=k_H$, the model system \eqref{7} endures a hopf bifurcation at the coexisting equilibrium point $E^*$ if the following conditions are met:
(\romannumeral 1) $Tr\left(J_{E^*}\right)|_{k=k_H}=0$,~
(\romannumeral 2)  $Det\left(J_{E^*}\right)|_{k=k_H}>0$,~ 
(\romannumeral 3)  $\dfrac{d}{d k}\Big[Re\left(\lambda_i(k)\right)\Big]_{k=k_H}\ne 0$ for $i = 1, 2$; where $\lambda_i$ is an eigenvalue of the Jacobian matrix corresponding to $E^*$.
\end{theorem}
\begin{proof}
    For the proof, refer to Appendix \ref{appA4}.
\end{proof}
 It is difficult to establish the analytic expression for the Hopf-bifurcation threshold because of the algebraic complexity of our model system. However, we numerically test the Hopf-bifurcation for a set of parameter values. If we take all of the parameter values from  set (1) of Table (\ref{tab 1}) and set $\alpha=0.25$ and $c=3$, the interior equilibrium point of the system (\ref{7}) $E_2^*= (0.101007273887931,1.809506754599451)$ becomes stable through Hopf-bifurcation at threshold value $k=k_H=0.509642226$ (see Fig. \ref{Fig.8}).
\subsection*{Nature of Hopf-bifurcation }
By determining the first Lyapunov number at the equilibrium point $E_H^*$, we can investigate the direction of the Hopf-bifurcation or the stability of limit cycle. To begin, we use the transformations $x=\hat{X}-x_H^*$ and $y=\hat{Y}-y_H^*$ to move $E_H^*(x_H^*,y_H^*)$ to the origin. Then the model system (\ref{7}) in the neighbourhood of the origin can be written as:
\begin{equation}\label{13}
\begin{aligned}
\frac{d\hat{X}}{d t}&=\alpha_{01}\hat{Y}+ \alpha_{10}\hat{X}+\alpha_{02}\hat{Y}^2+\alpha_{20}\hat{X}^2+\alpha_{11}\hat{X}\hat{Y}+\alpha_{30}\hat{X}^3+\alpha_{03}\hat{Y}^3+\alpha_{21}\hat{X}^2\hat{Y}+\alpha_{12}\hat{X}\hat{Y}^2+F_1(\hat{X},\hat{Y})\\
\frac{d\hat{Y}}{d t}&=\beta_{01}\hat{Y}+\beta_{10}\hat{X}+\beta_{02}\hat{Y}^2+\beta_{20}\hat{X}^2+\beta_{11}\hat{X}\hat{Y}+\beta_{30}\hat{X}^3+\beta_{03}\hat{Y}^3+\beta_{21}\hat{X}^2\hat{Y}+\beta_{12}\hat{X}\hat{Y}^2+F_2(\hat{X},\hat{Y}),
\end{aligned}
\end{equation}
where, the $\alpha_{ij}=\frac{1}{i!j!}\frac{\partial^{i+j}F}{\partial x^i\partial y^j}|_{(E_H^*,\theta_H)}$  , $\beta_{ij}=\frac{1}{i!j!}\frac{\partial^{i+j}G}{\partial x^i\partial y^j}|_{(E_H^*,\theta_H)}$  (i, j = 0, 1, 2 and 3) and  $F_i(\hat{x}, \hat{y})$ ($i=1,2$)  are power series in powers of $\hat{X}^i\hat{Y}^j$ satisfying $i + j\geq 4$.
The expressions of $\alpha_{ij}$ and $\beta_{ij}$ have been provided in Appendix \ref{appB3}.

We know that expression for Lyapunov coefficient ($\sigma$) \cite{perko2013differential} for a general planar system (\ref{13}) is given by:
\begin{equation*}
 \begin{aligned}
\sigma=&-\frac{3\pi}{2\alpha_{01}\Delta^\frac{3}{2}}\{[\beta_{10} \alpha_{10}(\beta_{02}\alpha_{11}+\alpha_{11}^2+\beta_{11}\alpha_{02})+\alpha_{01} \alpha_{10}(\beta_{11}\alpha_{20}+\beta_{11}^2+\beta_{02}\alpha_{11})+\beta_{10}^2(2\alpha_{02}\beta_{02}+\alpha_{11}\alpha_{02})\\
&-2\beta_{10}\alpha_{10}(-\alpha_{02}\alpha_{20}+\beta_{02}^2)-2\alpha_{01}\alpha_{10}(-\beta_{02}\beta_{20}+\alpha_{20}^2)-\alpha_{01}^2(\beta_{20}\beta_{11}+2\beta_{20}\alpha_{20})+(-2\alpha_{10}^2+\alpha_{01}\beta_{10})\\
&(-\alpha_{20}\alpha_{11}+\beta_{02}\beta_{11})]-(\beta_{10}\alpha_{01}+\alpha_{10}^2)[3(-\alpha_{30}\alpha_{01}+\beta_{03}\beta_{10})+2\alpha_{10}(\beta_{12}+\alpha_{21})+(-\beta_{21}\alpha_{01}+\alpha_{12}\beta_{10})]\},
\end{aligned}   
\end{equation*}
where, $\Delta=\alpha_{10}\beta_{01}-\alpha_{01}\beta_{10}.$
Using the model parameter values provided above in this subsection, we calculate~$\alpha_{i j},\beta_{i j}$, and get~$\sigma=120.438,$~which implies that the Hopf bifurcation is of subcritical type.

\section{Effect of fear level, and Birth Rate of Prey}\label{Sec.5}
\textbf{Effect of fear level on prey and predator equilibrium densities:}
Here, we are interested in observing the changes of the prey population with respect to the fear level. Since, the prey population at coexisting equilibrium is given by the positive solution of the Eq. (\ref{8}). Therefore, we differentiate the Eq. (\ref{8}) with respect to  parameter $k$ (level of fear) and we obtain
$$\dfrac{d x}{d k}\bigg|_{x=x^*}=\frac{-\left(x^7\dfrac{d A_1}{d k}+x^6\dfrac{d A_2}{d k}+x^5\dfrac{d A_3}{d k}+x^4\dfrac{d A_4}{d k}+x^3\dfrac{d A_5}{d k}+x^2\dfrac{d A_6}{d k}+x\dfrac{d A_7}{d k}+\dfrac{d A_8}{d k}\right)}{(7A_1x^6+6A_2x^5+5A_3x^4+4A_4x^3+3A_5x^2+2A_6x+A_7)}\bigg|_{x=x^*},$$
where $A_1, A_2, A_3, A_4, A_5, A_6, A_7,$ and $A_8$  have same meaning as mentioned in Eq. (\ref{8}) in subsection (\ref{sec.3.4}). Now, we have
\begin{align*}
\dfrac{d x}{d k}\bigg|_{x=x^*}=\frac{\gamma}{\delta}\bigg|_{x=x^*},
\end{align*}
where
\begin{equation*}
\begin{aligned}
\gamma=&-\alpha  (d a +x (-1+m)  (\beta  c-b \alpha   d+ x (-1+m) d)) \left( (-1+m) \beta (d a +x (-1+m)  (-\alpha  b d+\beta  c+x  (-1+m) d )) \right.\\
&\left.+e \left(d_1+  x d_2\right) (a +  x (-1+m) (-\alpha  b + x (-1+m) ))^2\right),\\
\delta=&(-1+m) \left(- c  (-1+m) \beta ^2 (-2 k \alpha  d  +e) (a+ x (-2 \alpha  b + 3  x (-1+m)) (-1+m) )+2  d (-1+m) \beta  (-\alpha  b + 2 x \right.\\
&\left.(-1+m) ) (-e + k \alpha  d ) (a+ x (-1+m)  (-\alpha  b + x (-1+m) ))+3 r e^2  (-\alpha  b + 2 x (-1+m) ) (a+(-1+m) x\right.\\ 
&\left.(-\alpha  b + x (-1+m) ))^2 +2 \alpha  c^2  (-1+m)^2 \beta ^3 k x\right)-e (a+ x (-1+m) (-\alpha  b + (-1+m) x)) \left(d_2 \left(a^2 (-k \alpha  d  + e) \right.\right.\\
&\left.\left.+ (-1+m) x a ((8 x (-1+m)-5 \alpha  b) (-k \alpha  d  + e)-2 k \alpha  \beta  c )+ x^2 (-1+m)^2  (2 \alpha  \beta  c k (2 \alpha  b-3 (-1+m) x)+\right.\right.\\
&\left.\left.(\alpha  b-m x+x) (-7  x (-1+m) + 4 \alpha  b) (-k \alpha  d  + e))\right)+ (-1+m) d_1 (3 a (-\alpha  b + 2 x (-1+m) )(-\alpha  d k + e)-a \alpha \right.\\
&\left. \beta  c k+ x (-1+m)  (\alpha  \beta  c k (-5  x (-1+m) + 3 \alpha  b)-3 (-\alpha  b + x (-1+m) ) (-\alpha  b + 2 x (-1+m) ) (-e +k \alpha  d )))\right).
\end{aligned}
\end{equation*}
Further, let
\begin{align}
 (d a + x(-1+m)  (-\alpha  b d+\beta  c+d (-1+m) x))<&0,~
 \delta< 0\label{14}.
\end{align}
The Eq. \eqref{14} ensures that~$\gamma>0$. If Eq.\eqref{14} holds, then~$\dfrac{d x}{d k}\bigg|_{x=x^*}<0$. It indicates that as fear in a prey population against predation increases, prey equilibrium density decreases monotonically; hence, an increasing level of fear gradually decreases the prey population density. Because the increasing level of fear compels prey to change their foraging tactics and spend more time on vigilance, which ultimately hampers their time for reproduction. As a result, prey birth rates decrease as fear levels rise, which eventually decreases the level of prey population density. Additionally, when $x^*$ is feasible, the predator population at the coexisting equilibrium point is given by $y^*=\frac{1}{e}\left(-d+\frac{c\beta x^* (1-m)}{a+b\alpha x^*(1-m)+(1-m)^2(x^*)^2}\right),$ whenever it is feasible.
Now, we differentiate $y$ with respect to $k$ as follows to see how $k$ affects the predator population:
$$\dfrac{d y^*}{d k}=\frac{c\beta (-m+1)}{e}\left[\frac{a+ x^*(-m+1) b\alpha +(-m+1)^2(x^*)^2-\left(b(-m+1) \alpha +2 x^* (-m+1)^2\right) x^*}{\left(a+x^*(-m+1) b\alpha + (x^*)^2 (-m+1)^2\right)^2}\right]\dfrac{d x^*}{d k},$$
$$\dfrac{d y^*}{d k}=\frac{\beta c (-m+1)}{e}\left[\frac{\left(a-(x^*)^2 (-m+1)^2\right)}{\left(a+ x^*(-m+1) b\alpha +(x^*)^2 (-m+1)^2\right)^2}\right]\dfrac{d x^*}{d k}.$$
Further, let $ \left(a- (x^*)^2 (-m+1)^2\right)>0$. It could be observed that the change in the equilibrium density of the predator population depends on changes in the prey equilibrium density, i.e., the prey equilibrium density decreases, then the predator equilibrium density also decreases with respect to level of fear $k$. Hence, the increased level of fear reduces the prey population directly and the predator population indirectly. In a biological sense, increased fear not only forces the prey populations to forage less but also to tighten their group defences, which consequently hampers the food intake rate of predators. As a result, predator numbers decrease as the fear level increases.\\
\textbf{Effect of prey birth rate on prey equilibrium density:}
In a similar fashion, as we have discussed in the above subsection, we can observe the impact the prey's birth rate $(r)$ on the prey equilibrium density. The change in prey population with respect to $r$ is given by
$$\dfrac{d x}{d r}\bigg|_{x=x^*}=\frac{N }{\delta}\bigg|_{x=x^*},$$
where,   $ N=-e^2 (a+ x  (-1+m) (-\alpha  b + (-1+m) x))^3<0.$
Therefore, if Eq. \eqref{14} holds then we have $\dfrac{d x}{d r}\bigg|_{x=x^*}>0$. It indicates that as prey birth rate increases, prey equilibrium density increases monotonically. It may happens for two reasons: (1) as prey birth rates rise, so do prey numbers; and (2) as prey numbers rise, so do prey group defenses, which reduce the chances of prey being predated and indirectly contribute to prey numbers.
\section{Spatial Model System}\label{Sec.6}
Reaction-diffusion equations are widely used as models for spatial effects in ecology. They support three important types of ecological phenomena: the existence of a minimal patch size necessary to sustain a population and the formation of spatial patterns in the distributions of populations in homogeneous environments \cite{cosner2008reaction}. The spatial model between the interaction of prey $x(u, v, t)$ and predator $y(u, v, t)$ is described by the following set of partial differential equations:
\begin{equation}\label{18}
   \begin{aligned}
\frac{{\partial x}}{{\partial t}} =& \frac{{rx}}{{(1 + k\alpha y)}} - {d_1}x - {d_2}{x^2} - \frac{{\left( { - m + 1} \right)\beta xy}}{{a + ( - m + 1)\alpha bx + {{(- m+1)}^2}{x^2}}} + {D _1}{\nabla ^2}x,\\
\frac{{\partial y}}{{\partial t}} =&  - dy - e{y^2} + \frac{{\left( { - m + 1} \right)\beta cxy}}{{a + ( - m + 1)\alpha bx + {{(- m+1)}^2}{x^2}}} + {D _2}{\nabla ^2}y,
 \end{aligned}
\end{equation}
  
with initial conditions
\begin{equation*}
x(u,\,v,\,0) > 0,\,\,y(u,\,v,\,0) > 0\,\,\,\,\, {\rm for}\,\,\, (u,\,v) \in \Omega,
\end{equation*}
and boundary conditions
\begin{equation*}
\frac{{\partial x}}{{\partial \vartheta }} = \frac{{\partial y}}{{\partial \vartheta }} = 0,\,\,\,\,\,\,(u,\,v) \in \partial \Omega \,\,\,\, {\rm  for\,\,\, all}\,\,\,\, t,
\end{equation*}
where $D_1$ and $D_2$ are diffusion coefficients of prey and predator population respectively, the laplacian operator ${\nabla ^2} \equiv \frac{{{\partial ^2}}}{{\partial {u^2}}} + \frac{{{\partial ^2}}}{{\partial {v^2}}} $  and $\vartheta $ is outward normal to the $\partial \Omega $
.
\subsection{Dynamics of spatial Model system}
To understand the spatial dynamics of model system \eqref{18}, we consider the linearized form of the system about $E^{*}(x^{*}, y^{*})$ as follows:
\begin{equation}\label{22}
   \begin{aligned}
\frac{{\partial \bar x}}{{\partial t}} = &{j_{11}}\bar x + {j_{12}}\bar y+ {D _1}{\nabla ^2}\bar x,\\
\frac{{\partial \bar y}}{{\partial t}} = &{j_{21}}\bar x + {j_{22}}\bar y+ {D _2}{\nabla ^2} \bar y,
\end{aligned}
\end{equation}
where we introduce small perturbations $\bar x = x - {x^*}$ and $\bar y = y - {y^*}.$ Let us suppose that the solution of model system \eqref{18} is of the form:
\begin{equation*}
\left( \begin{array}{l}
 {\bar x} \\ 
 {\bar y} \\ 
 \end{array} \right) = \left( \begin{array}{l}
 {l_1} \\ 
 {l_2} \\ 
 \end{array} \right)\exp ({\lambda _\Lambda}t)\cos ({\Lambda_u}u)\cos ({\Lambda_v}v),
\end{equation*}
where $l_{1}$, $l_{2}$ are sufficiently small constants, $\Lambda_{u}$ and $\Lambda_{v}$ are the component of wave number along $u$ and $v$ direction respectively and $\Lambda_{k}$ is wave length.\\
The jacobian matrix of linearized model system \eqref{22} is
\begin{equation*}
\bar J = {\left( {\begin{array}{*{20}{c}}
   {{j_{11}} - {D_1}{\Lambda^2}} & {{j_{12}}}  \\
   {{j_{21}}} & {{j_{22}} - {D_2}{\Lambda^2}}  \\
\end{array}} \right)},
\end{equation*} 
where $\Lambda$ is wave number given by $\Lambda^{2}=\Lambda^{2}_{u}+\Lambda^{2}_{v}.$ 
The characteristic equation of $ \bar J $ is given by
\begin{equation*}
\lambda _\Lambda^2 + {\rho _1}{\lambda _\Lambda} + {\rho _2} = 0,
\end{equation*} 
where,
\begin{eqnarray}
{\rho _1} &=&  - ({j_{11}} + {j_{22}}) + \left( {{D_1} + {D_2}} \right){\Lambda^2},\nonumber \\
{\rho _2} &=& {D_1}{D_2}{\Lambda^4} - \left( {{j_{11}}{D_2} + {j_{22}}{D_1}} \right){\Lambda^2} + {j_{11}}{j_{22}} - {j_{12}}{j_{21}}.\nonumber
\end{eqnarray}
The equilibrium point ${E^ * }\left( {{x^ * },\,\,{y^ * }} \right)
$ is locally asymptotically stable in the presence of diffusion if $\rho_{1}>0$ and $\rho_{2}>0$.
\subsection{Turing Instability}
Turing instability occurs when steady state is stable in the absence of diffusion $(D_{1}=D_{2}=0)$ and unstable in the presence of diffusion ${D_i} \ne 0\,\,\,\,\rm {for} \,\,\, {\it i} = 1,\,\,2.$ The condition for occurrence of Turing instability is given by
\begin{equation*}
H({\Lambda^2}) = {D_1}{D_2}{\Lambda^4} - ({j_{11}}{D_2} + {j_{22}}{D_1})\Lambda^2+ {j_{11}}{j_{22}} - {j_{12}}{j_{21}} < 0.
\end{equation*}
The function $H(\Lambda^{2})$ is a quadratic equation in $\Lambda^{2}$. It represents a parabola opening upward and $H(\Lambda^{2})$ has its minimum $({H_{\min }})$ at some value $\Lambda_c^2$ of $\Lambda^{2}$ at vertex of the parabola where $ \Lambda_c^2 = \frac{{{j_{11}}{D_2} + {j_{22}}{D_1}}}{{2{D_1}{D_2}}}. 
$ Then the condition that $H(\Lambda_c^2)<0$ gives,
 \begin{equation}
 {\left( {{j_{11}}{D_2} + {j_{22}}{D_1}} \right)^2} > 4{D_1}{D_2}\left( {{j_{11}}{j_{22}} - {j_{12}}{j_{21}}} \right).\label{28}
 \end{equation} 
 Therefore, $H(\Lambda^2_c)$ will be negative, when condition of Eq. \eqref{28} satisfies. 
 We summarize the above discussion in the form of following theorem.
\begin{theorem}
At the homogeneous steady state $E^{*}(x^{*}, y^{*})$, the model system undergoes Turing instability if 
(\romannumeral 1) ${D_1}{j_{22}} + {D_2}{j_{11}} > 0,$~
(\romannumeral 2) ${\left( {{D_1}{j_{22}} + {D_2}{j_{11}}} \right)^2} > 4{D_1}{D_2}\left( {{j_{11}}{j_{22}} - {j_{12}}{j_{21}}} \right),$~
(\romannumeral 3) ${j_{11}} + {j_{22}} < 0,\,\,\,\,\,\,\,\,\,{j_{11}}{j_{22}} - {j_{12}}{j_{21}}>0.$
\end{theorem} 
The diffusion induced instability in the model
system \eqref{18} has been discussed by taking the following
parameter values:
\begin{eqnarray}
r &= &1.1,\,\, d_{1}=0.05,\,\,d_{2}=0.055,\,\,\beta=0.071,\,\,m=0.3,\,\, k=0.5,\,\,\alpha=1,\,\,d=0.05,\,\,e=0.05,\,\,c=5,\nonumber \\
a&=& 0.1,\,\,b=0.125, D_{1}=0.005,\,\,\,D_{2}=2.\label{29}
\end{eqnarray}
At the set of parametric values given in Eq. \eqref{29}, we obtain ${E^ * }\left( {{x^ * },\,\,{y^ * }} \right)=(0.0488345, 1.30182)$. Also
\begin{equation}\nonumber
j_{11}+j_{22}=-0.0293104<0\,\,\,\,\, {\rm and} \,\,\,\,\, j_{11}j_{22}-j_{12}j_{21}=0.0922047>0
\end{equation} 
ensure the locally asymptotical stability of the system without diffusion. For diffusive model \eqref{18}, $H(\Lambda^{2})<0$  for $1.7001<\Lambda^{2}<5.42347.$ For the above set of parameter values, the plot of Turing instability is given in Fig \ref{Fig.insta}. 
\begin{figure}[H]
\subfloat[]{\includegraphics[height=5cm]{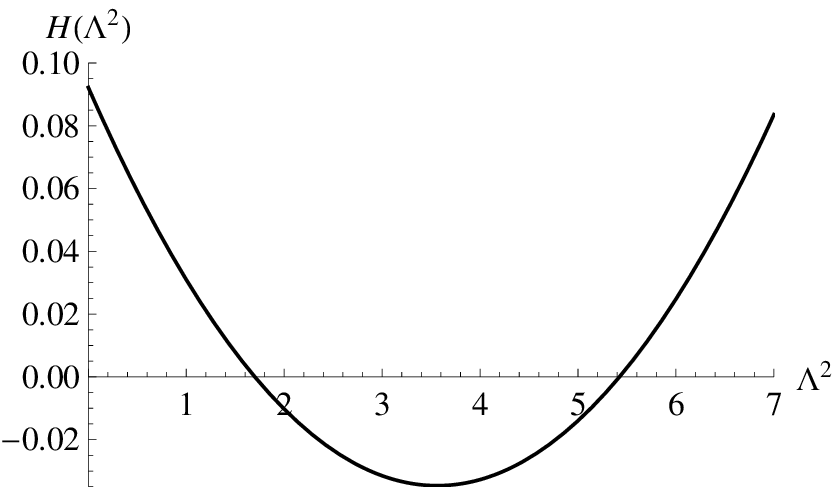}\label{Fig.Insta(a)}}
\qquad 
\subfloat[]{\includegraphics[height=5cm]{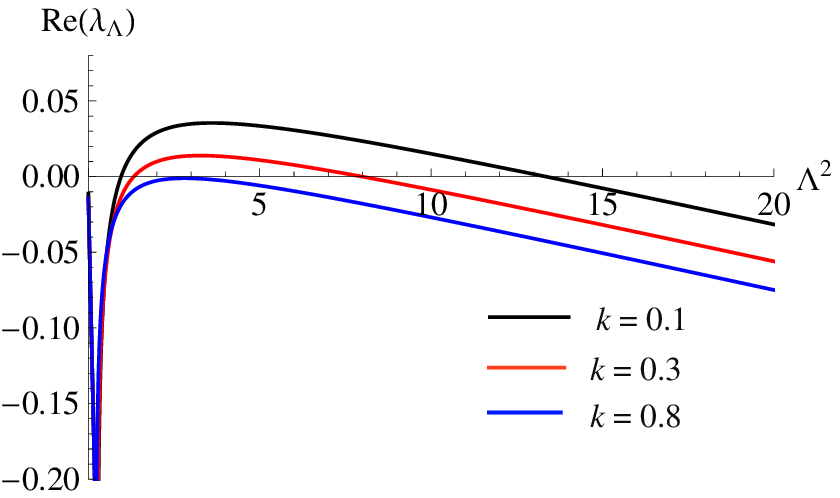}\label{Fig.Insta(b)}}
\caption{(a) Behavior of $H(\Lambda^2)$ for the occurrence of Turing instability and (b) for different values of level of fear factor $k$, the real part of the dispersion relation of model system \eqref{18}. For the values of other parameters, refer to Eq. \eqref{29}.}\label{Fig.insta}
\end{figure}

\section{Numerical simulation}\label{Sec.7}
Now, we perform the numerical simulation to validate our obtained analytical  results of spatial and non spatial prey-predator systems.
\subsection{Numerical Simulation of Non-Spatial Model}
\textbf{Influence of Intra-Prey Competition Parameter:}
Fig. \ref{Fig.2(a)} and Fig. \ref{Fig.2(b)} show that how  the existence and stability of various feasible equilibrium points are influenced by the intra-prey competition parameter $d_2$.  The interior equilibrium point $E^*_3$ exist always and is stable. It can be visualized in Fig. \ref{Fig.2(b)} as the corresponding $Max(Re(\lambda_i))$  is negative throughout.  The interior equilibrium point $E^*_1$ emerges at $(d_2)_{TC}$, whereas the predator extinct equilibrium point $E_1$ becomes unstable due to stability exchange with $E^*_1$ via transcritical bifurcation. The interior equilibrium point $E^*_1$ exist and remain stable in the region $R_2$ (presented in Fig. \ref{Fig.2(a)} and \ref{Fig.2(b)}) because the corresponding $Max(Re(\lambda_i))$ in Fig.\ref{Fig.2(b)} is negative throughout. The interior equilibrium point $E^*_2$ exist in the region $R_1\cup R_2 $ and is unstable as the corresponding $Max(Re(\lambda_i))$  is positive throughout regions (please refer the Fig. \ref{Fig.2(b)}). As $d_2$ increases, interior equilibrium points $E^*_1$ and $E^*_2$ in region $R_2$ move closer to each other, and finally collide  at $d_2=(d_2)_{SN}$ via saddle-node bifurcation. The collided equilibrium point $E^*_{12}=(4.7809, 3.0553)$ will be half stable. The regions $R_1$ and $R_2$ are regions of bistabilty as there are two local basins of attraction. Local attractors in region $R_1$ are the interior equilibrium point $E^*_3$ and the predator extinct equilibrium point $E_1$. Therefore, depending on the initial population size, the predator population may extinct or the system would get stabilised at a low prey density. The interior equilibrium points $E^*_1$ and $E^*_3$ are local attractors in region $R_2$, so the population always survives but only changes the level of population. Because there is only one basin of attraction, which is around the interior equilibrium point of $E^*_3$, the population survives at a low prey density in region $R_3$.\\ 
Hence, from Fig. \ref{Fig.2}, we conclude that there is always a possibility of the coexistence of prey and predator species. The scenario of predator extinction at low values of $d_2$ arises because of a relatively higher level of prey population than of predator population. It leads to robust group defence by prey against predation, which eventually forces predators to starve to death. As a result, predators are no longer present in the system. For a large enough increment in values of $d_2$, two possibilities occur for stable coexistence of populations depending on the initial populations: (a) both the population coexist with higher levels of prey population and (b) both the population coexists with lower levels of prey population. Whichever of these two possibilities is the case, both species somehow find a way to live with each other and do not affect the biodiversity of the ecosystem. In the case of the first possibility, the increasing value of $d_2$ has two effects: (1) it decreases prey numbers directly (2) as well as indirectly by causing possible laxity in the group defence of prey, which allows predators to hunt more prey. Therefore, the gradual increment in $d_2$ removes the possibility (a), leaving the system with only the possibility (b). The possibility (b) occurs for all values of $d_2$ in which both the populations live together with relatively lower densities without forcing each other toward extinction.
\begin{figure}[H]
\subfloat[]{\includegraphics[height=7cm,width=8.2cm]{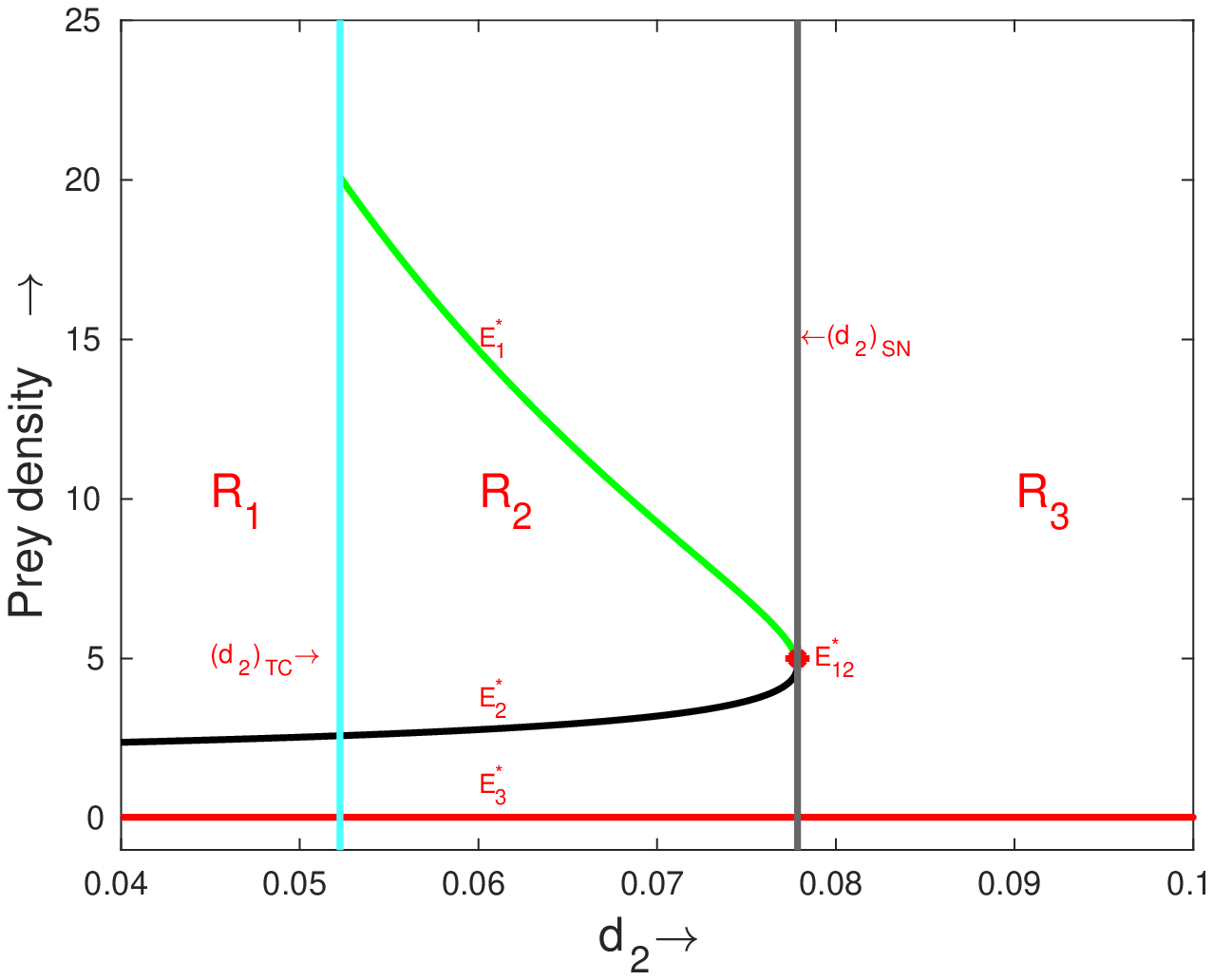}\label{Fig.2(a)}}
\qquad
\subfloat[]{\includegraphics[height=7cm,width=8.1cm]{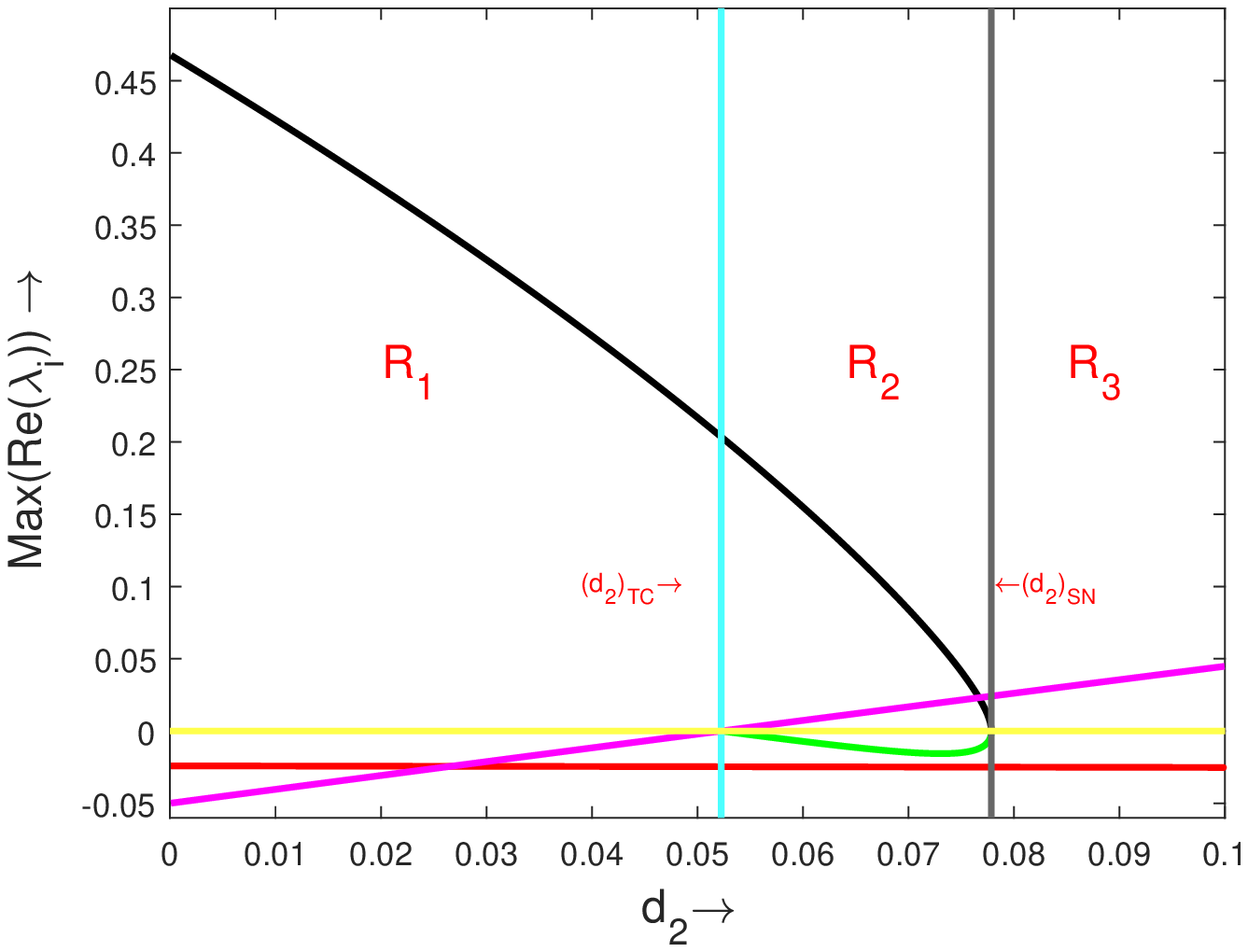}\label{Fig.2(b)}}
\caption{All parameter values, apart from $d_2$, are chosen from set (1) in Table \ref{tab 1}. (a) One parameter bifurcation diagram of  model system \eqref{7}, corresponding to $d_2$. The green, black, and red curves represent the prey density at equilibrium points $E^*_1, E^*_2$,  and $E^*_3,$ respectively. The cyan and violet vertical lines represent the transcritical and saddle-node bifurcation thresholds  $d_2=(d_2)_{TC}=0.0522466351144$ and $d_2=(d_2)_{SN}=0.077854268968873742$, respectively.  (b) The green, black, red, and magenta curves represent the maximum of the real part of the eigenvalues corresponding to the interior equilibrium points $E^*_1, E^*_2, E^*_3$, and  $E_1$, respectively}.\label{Fig.2}
\end{figure}

\textbf{Influence of Birth Rate of Prey Parameter:} Here, we focus on how the birth rate of prey influences the prey-predator dynamics. Here, Fig. \ref{Fig.3(a)} and Fig. \ref{Fig.3(b)} depict the dependence of existence and stability of various feasible equilibrium points on the prey birth rate parameter $r$. Only the population extinct equilibrium point $E_0$ is feasible and a global attractor when $r<r_{TC}$. Therefore, both populations are eliminated from the system when $r< r_{TC}$.  As $r$ achieves $r_{TC}$, the population extinct equilibrium point $E_0$ exchanges its stability to become unstable permanently via transcritical bifurcation, and the predator extinct equilibrium point $E_1$ comes into role. In the region $R_1$, the $E_1$ remains  stable. When $r$ reaches $r_{{TC}_1}$, the predator extinct equilibrium point $E_1$ becomes unstable by exchanging stability with the interior equilibrium point $E^*_2$ via transcritical bifurcation. The predator extinct equilibrium point $E_1$ remains unstable in the region $R_2.$ At  $r=r_{{TC}_2}$, it exchanges stability with the interior equilibrium point $E^*_1$ via transcritical bifurcation and returns to stability. Region $R_2\cup R_3$ is the region of existence for interior equilibrium point $E^*_2$ and  it remains stable because the corresponding $Max(Re(\lambda_i))$  is remain negative throughout this region (presented in Fig. \ref{Fig.3(b)}). On the other hand, interior equilibrium point $E^*_1$ will unstable whenever it exists.  Interior equilibrium point $E^*_1$ exist only in the region $R_3$ (refer to Fig. \ref{Fig.3(a)}), where it is always unstable because the corresponding $Max(Re(\lambda_i))$ in Fig. \ref{Fig.3(b)} is positive throughout. On other hand, it can be observed that interior equilibrium  $E^*_1$ remains unstable whenever exists. As $r$  moves through region $R_3$, the equilibrium points $E^*_1$ and  $E^*_2$ move towards each other, eventually collide at $E^*_{12}=(1.247196, 2.438599)$ when $r=r_{SN}$. In region $R_4$, only the predator extinct equilibrium point remains stable and no interior equilibrium points exist. Since, $R_3$ is the bistability region for $E_1$ and $E^*_2,$ therefore, it is the region of local basins for $E_1$ and $E^*_2$. So the predator population may extinct or the system could stabilise around co-existence equilibrium point depending on the initial population size. In the region $R_2$, the both the population survive because the interior equilibrium point $E^*_2$ is the only basin of attraction for system. $R_1$ and $R_4$ are the regions in which predator population would go to extinction, as the predator extinct equilibrium point $E_1$ is the only local attractor for the system. 

As a result, predator species do not exist in the system for both extremely high and extremely low birth rates of prey. When birth rates of prey are extremely low, predators do not get enough food to survive and vanish from the system. In the case of extremely high prey birth rates, prey enable themselves to tighten the group defence as much as they can so that predators are not able to hunt them down. It triggers the predators' starvation and eventually their extinction from the system. When both prey and predators exist in the system and prey birth rates increases gradually, then a case of bistability arises between the predator extinct equilibrium and a coexisting equilibrium. If the initial predator species' size is not large enough, then high levels of prey population will cause the extinction of predators by using their group defence. In the case of predators with a large enough initial size, both species live together until prey birth rates exceed a threshold, because after this, prey numbers will be high enough to make predators starve by using their group defence and cause predator extinction.
\begin{figure}[H]
\subfloat[]{\includegraphics[height=7cm,width=8.2cm]{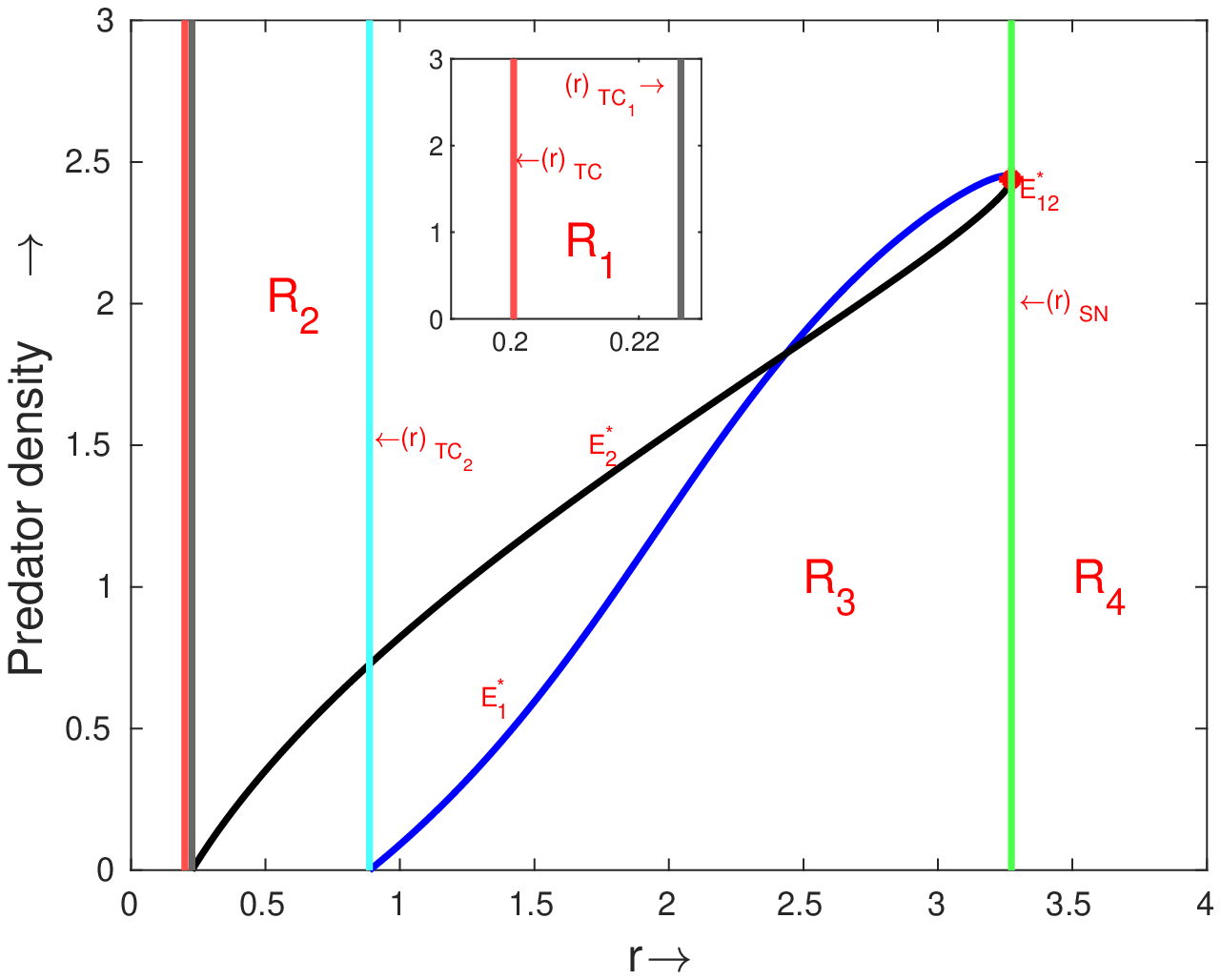}\label{Fig.3(a)}}
\qquad
\subfloat[]{\includegraphics[height=7cm,width=8.1cm]{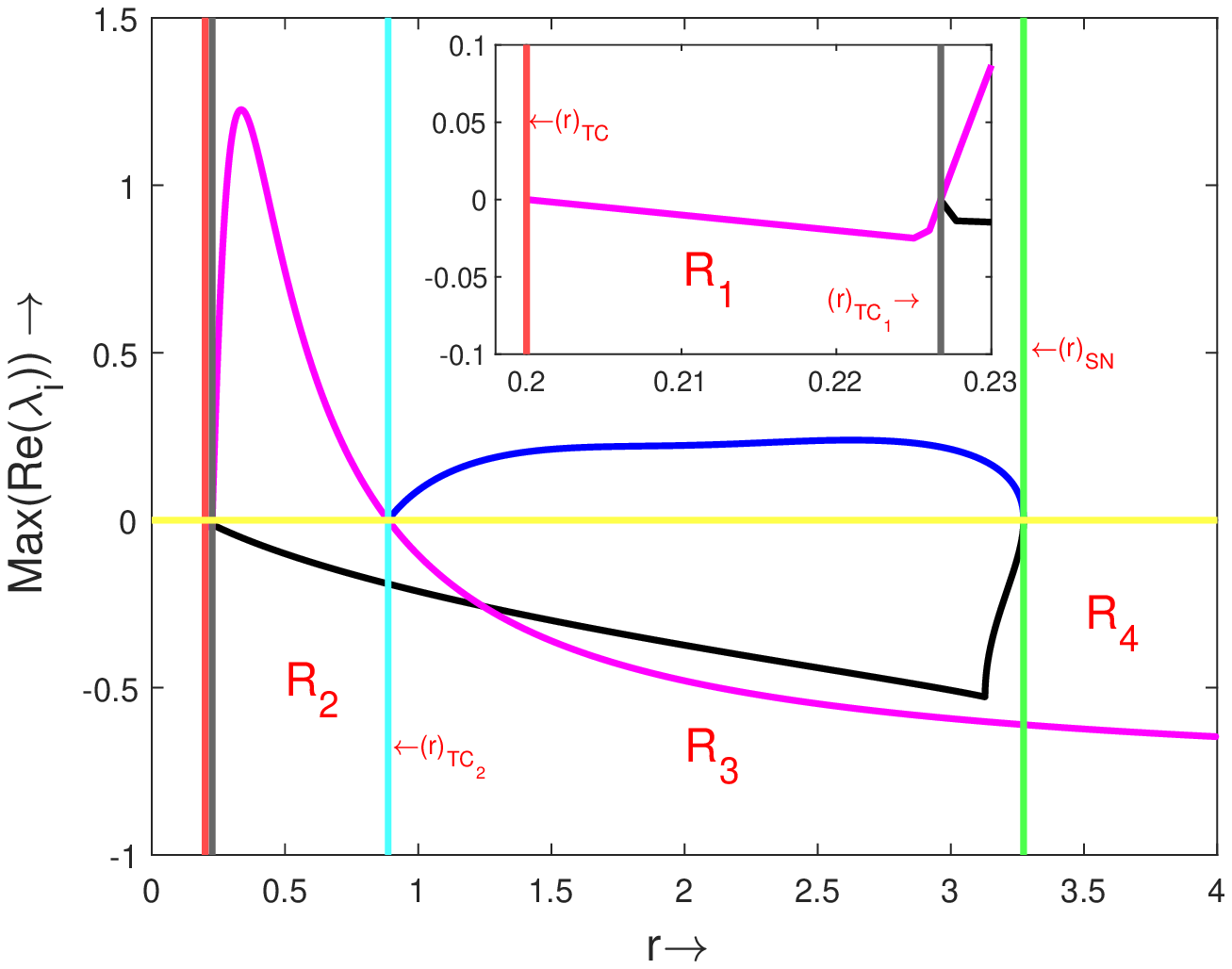}\label{Fig.3(b)}}
\caption{All parameter values, apart from $r$, are chosen from set (2) in Table \ref{tab 1}. (a) Diagram of a single parameter bifurcation for model system (\ref{7}), taking $r$ as a bifurcation parameter. The blue, and black curves represent the existence of interior equilibrium points $E^*_1, E^*_2$ respectively. The red, the violet, the cyan, and the green vertical lines represent the transcritical bifurcation at $r=(r)_{TC}=0.2$, transcritical bifurcation at $r=(r)_{{TC}_1}=0.226735001383921$, transcritical bifurcation at $r=(r)_{{TC}_2}=0.887014998616079$, and saddle-node bifurcation at  $r=(r)_{SN}=3.273118754537237107271607$,  respectively.   (b) The blue, black, and magenta curves represent the maximum of real part of eigen values corresponding to interior equilibrium points $E^*_1, E^*_2$ and predator free equilibrium point $E_1$, respectively.}\label{Fig.3}
\end{figure}

\textbf{Stability analysis of equilibria:} Only the population extinct equilibrium point is feasible and globally stable if $r<d_1$, as stated in Theorem \ref{thm.8}\ref{thm. trivial stability}. Therefore whenever this condition holds, both populations go to extinction as time goes on, as depicted in the Fig. \ref{Fig.4(a)}. If we choose a suitable set of parametric values which satisfy the local stability condition of the predator extinct equilibrium point $E_1$ as stated in Theorem \ref{thm.8}\ref{thm. boundary stability}, then the predator population will become extinct and the prey population will approach  $\left(\frac{r-d_1}{d_2}\right)$, as shown in Fig. \ref{Fig.4(b)}.

\begin{figure}[H]
\subfloat[The local stability of $E_0(0,0)$ with $d_1=1.2$ ]{\includegraphics[height=7cm,width=8.1cm]{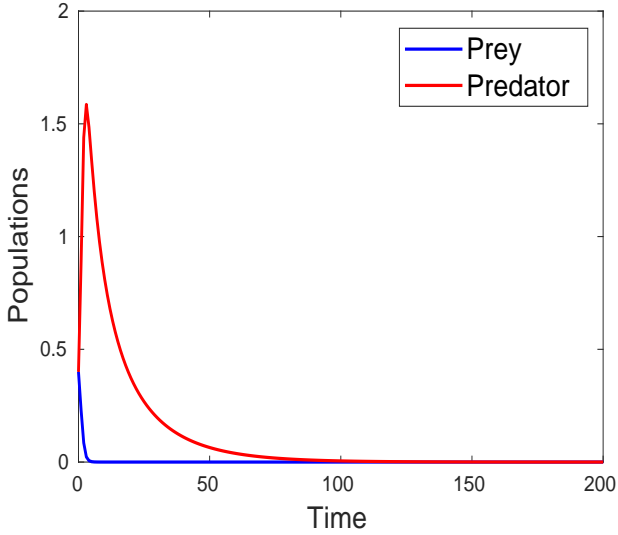}\label{Fig.4(a)}}
\qquad 
\subfloat[The local stability of the predator extinct equilibrium $E_1$  with $d_1=1, d_2=2, c=1$ ]{\includegraphics[height=7cm,width=8.1cm]{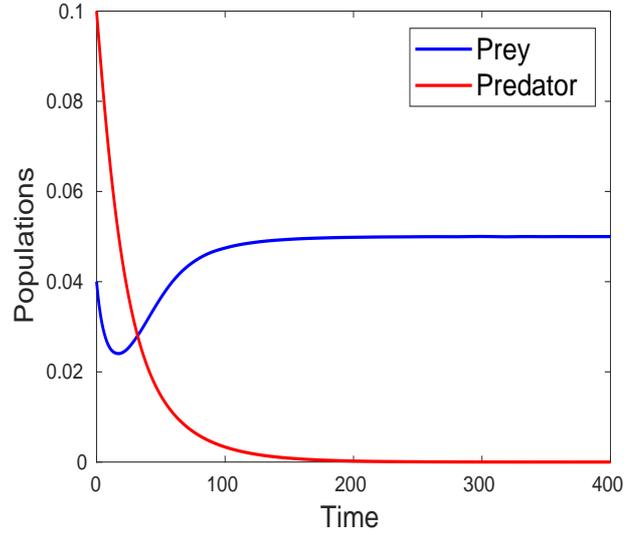}\label{Fig.4(b)}}
\caption{The figure shows the local stability of equilibria when the  numerical values of parameters are taken same as set (1) in Table \ref{tab 1} apart from stated above. }\label{Fig.4}
\end{figure}
 Further, we vary $d_2$ by taking all other parametric values from the set (1) of the Table \ref{tab 1}, the interior equilibrium point (coexisting equilibrium) $E^*_3$ always exists and is locally asymptotically stable. If we consider $d_2=-.055$ and other parametric values same as set (1) of the Table \ref{tab 1} then interior equilibrium point $E^*_{23}=(0.0234354,1.2758912)$ exists. In Fig. \ref{Fig.5(a)}, we see that both populations initially exhibit damped oscillations, but over time, they move closer to a positive level. Therefore, Fig. \ref{Fig.5(a)} and  \ref{Fig.5(b)} represent the local stability of interior equilibrium point. As a result, we can achieve the conclusion that both populations may survive throughout time and the ecosystem may enter a stable state when both populations coexist under specific environmental conditions.   

\begin{figure}[H]
\subfloat[Long term behavior of the system \eqref{7}. ]{\includegraphics[height=7cm,width=8.2cm]{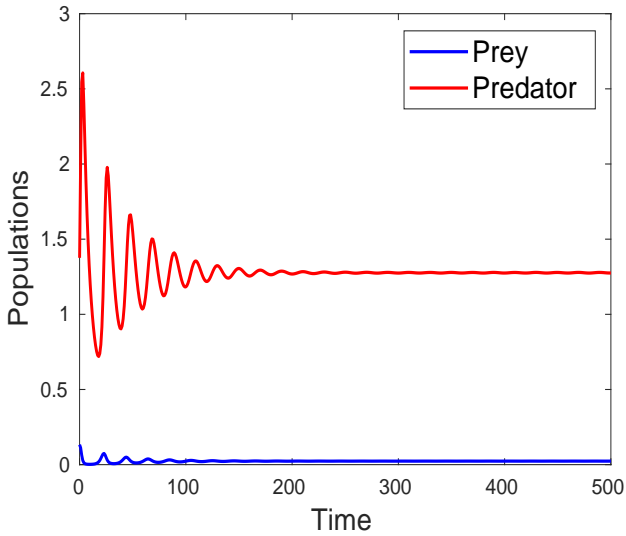}\label{Fig.5(a)}}
\qquad 
\subfloat[Phase portrait corresponding to time series plot Fig. \ref{Fig.5(a)}.  ]{\includegraphics[height=7cm,width=8.1cm]{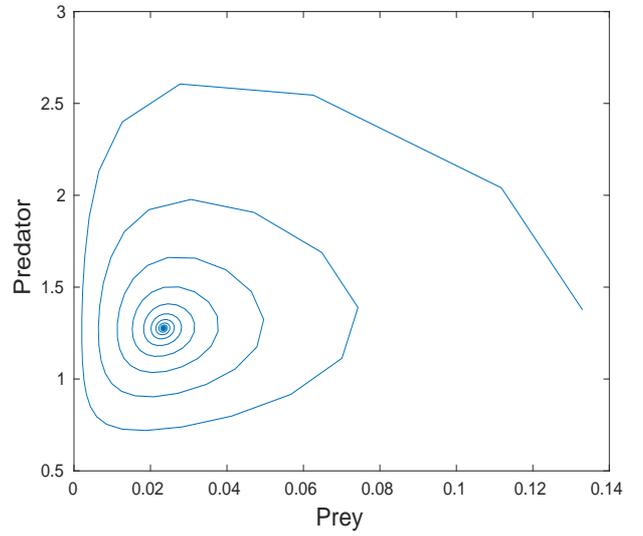}\label{Fig.5(b)}}
\caption{The figure shows the local stability of coexisting equilibrium $E^*_3$. Here, all parametric values, apart from $d_2=0.055$, are chosen from set (1) of  Table \ref{tab 1}.}\label{Fig.5}
\end{figure}
Now, we look at how the fear level and the prey birth rate affect the equilibrium density of prey and predator. Whenever the conditions \eqref{14}  hold, the equilibrium density of prey decreases as the level of fear increases, which is depicted in Fig. \ref{Fig.6(a)}. As shown in Fig. \ref{Fig.6(b)}, whenever condition ($ \left(a- (x^*)^2 (-m+1)^2\right)>0$) holds and the prey equilibrium density decreases as the level of fear increases, the predator equilibrium density decreases as well. It happens because prey modifies its foraging behaviour and spends more time being vigilant as its level of fear rises, which leads to decreased food intake and reduced reproduction. The population density of prey at equilibrium drops as a result. This altered prey foraging behaviour reduces the amount of food available to predators, hinders their capacity to consume it, and lowers the predator population density at equilibrium. Whenever $\delta < 0 $ (refer to Eq. \eqref{14}), the equilibrium density of prey and predator increases as the birth rate of prey increases, which is depicted in Fig. \ref{Fig.7}. Increasing the prey birth rate does not only increase prey numbers but also the food availability for predators. Hence, both prey and predator populations increase at equilibrium.
\begin{figure}[H]
\subfloat[Effect of fear on prey density at coexisting equilibrium.]{\includegraphics[height=6cm,width=8.2cm]{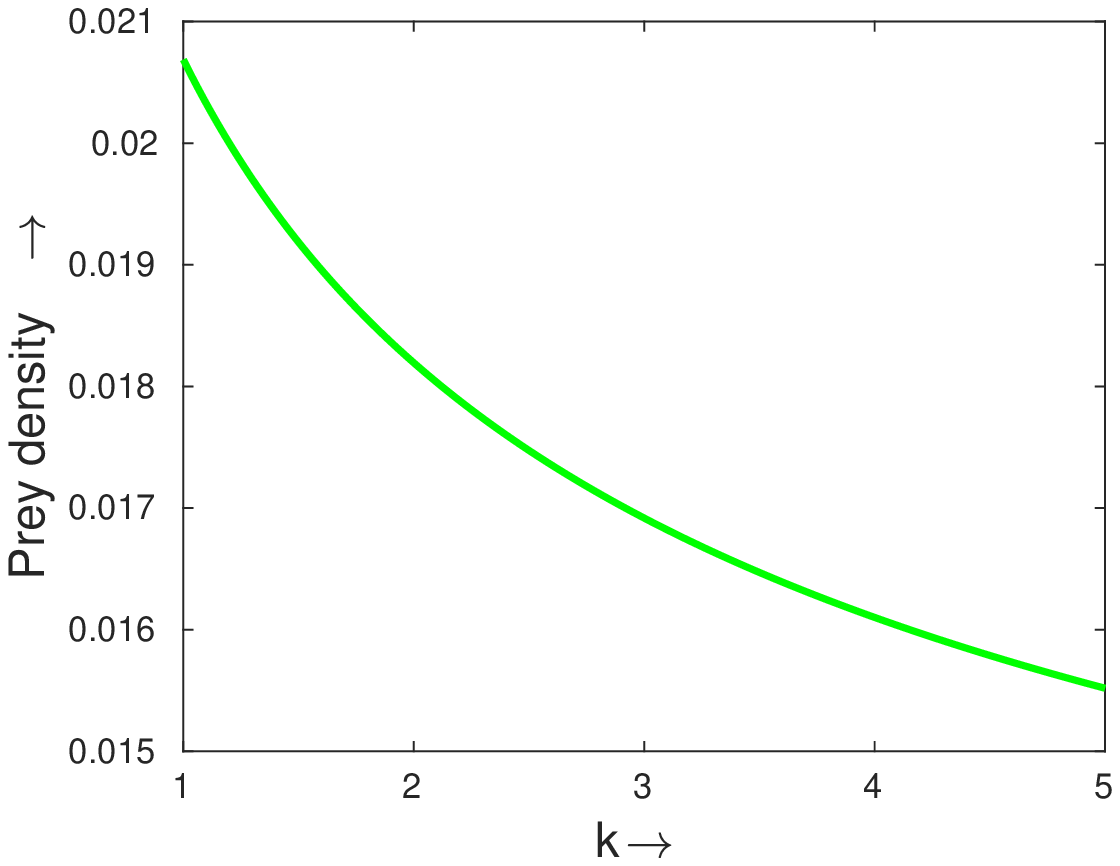}\label{Fig.6(a)}}
\qquad
\subfloat[Effect of fear on predator density at coexisting equilibrium.]{\includegraphics[height=6cm,width=8.1cm]{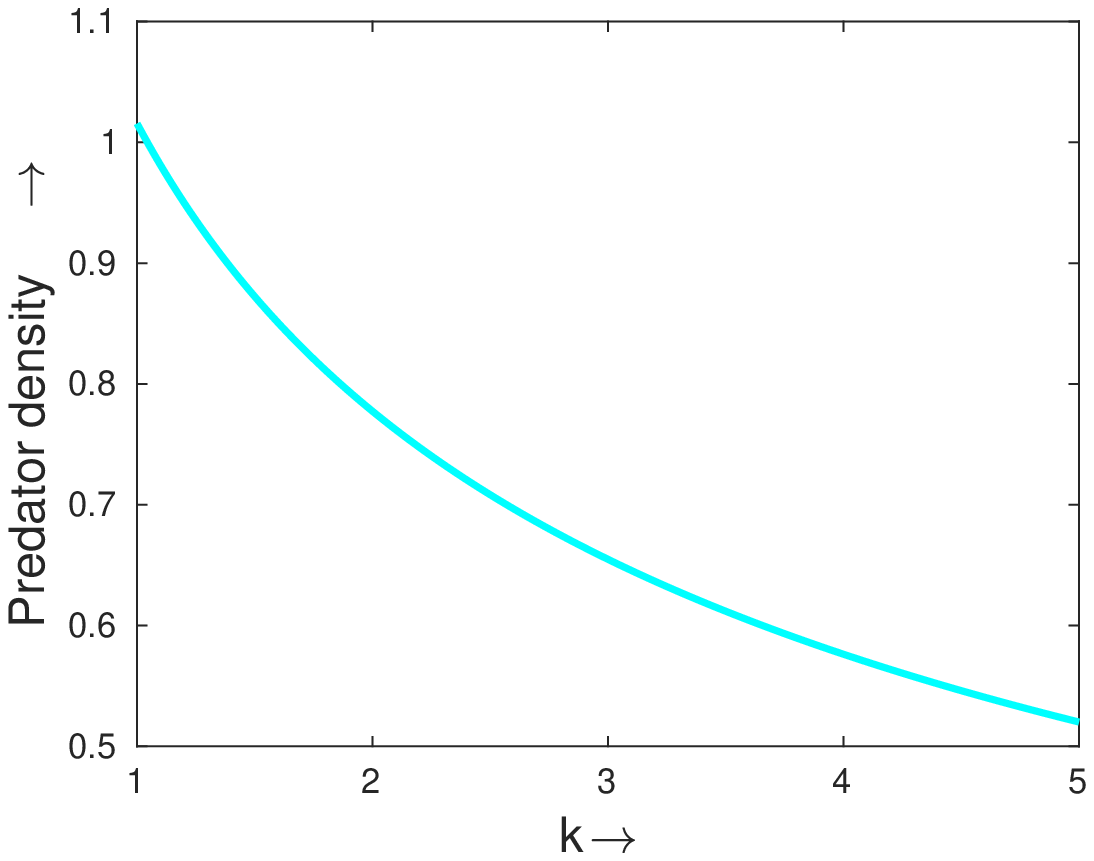}\label{Fig.6(b)}}
\caption{The figure depicts the effect of fear in prey on equilibrium state of species. Here, all parametric values are taken from set (1) in Table \ref{tab 1}.}\label{Fig.6}
\end{figure}
\begin{figure}[H]
\subfloat[Effect of prey birth rate on prey density at coexisting equilibrium.]{\includegraphics[height=6cm,width=8.2cm]{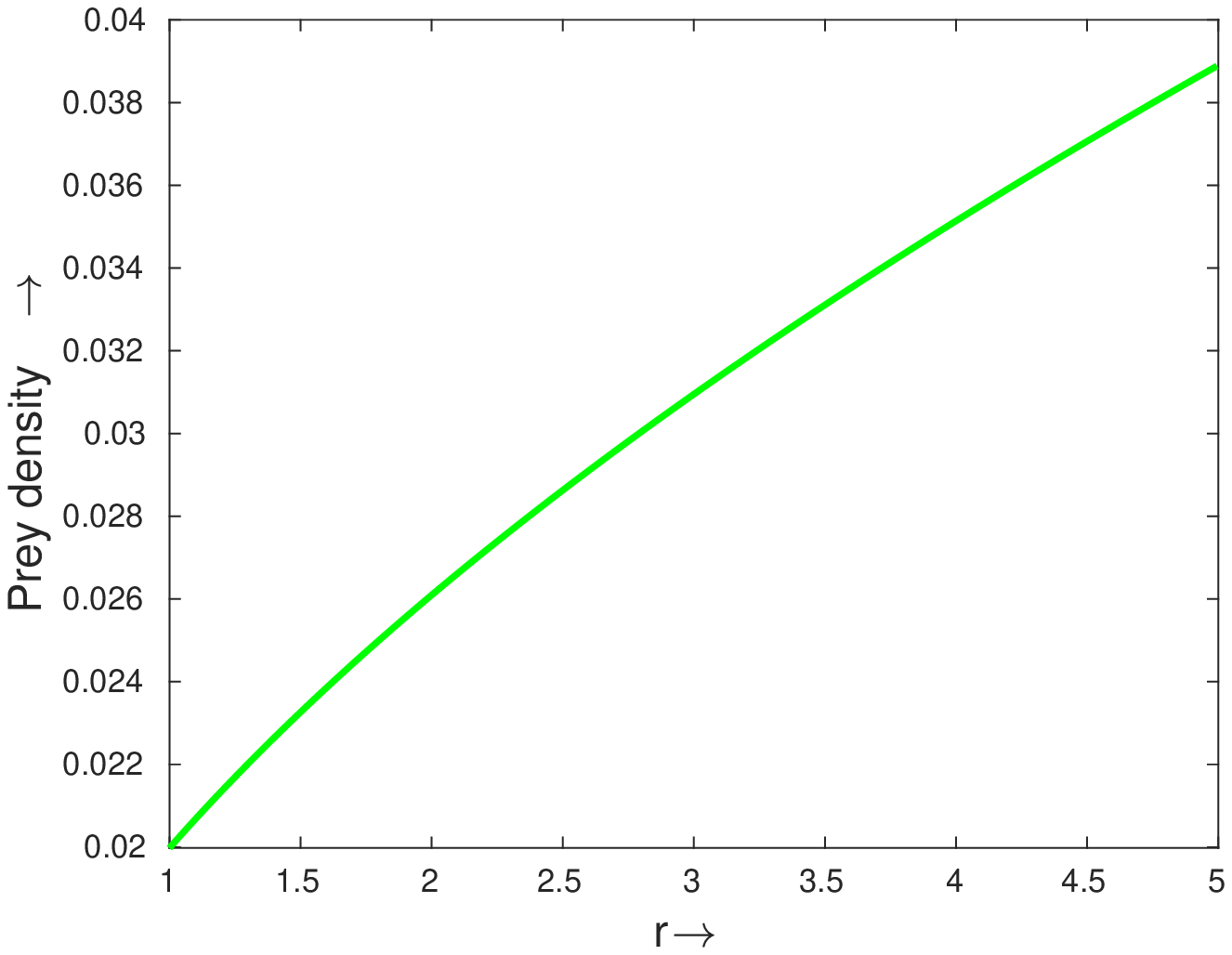}}
\qquad
\subfloat[Effect of prey birth rate on predator density at coexisting equilibrium.]{\includegraphics[height=6cm,width=8.1cm]{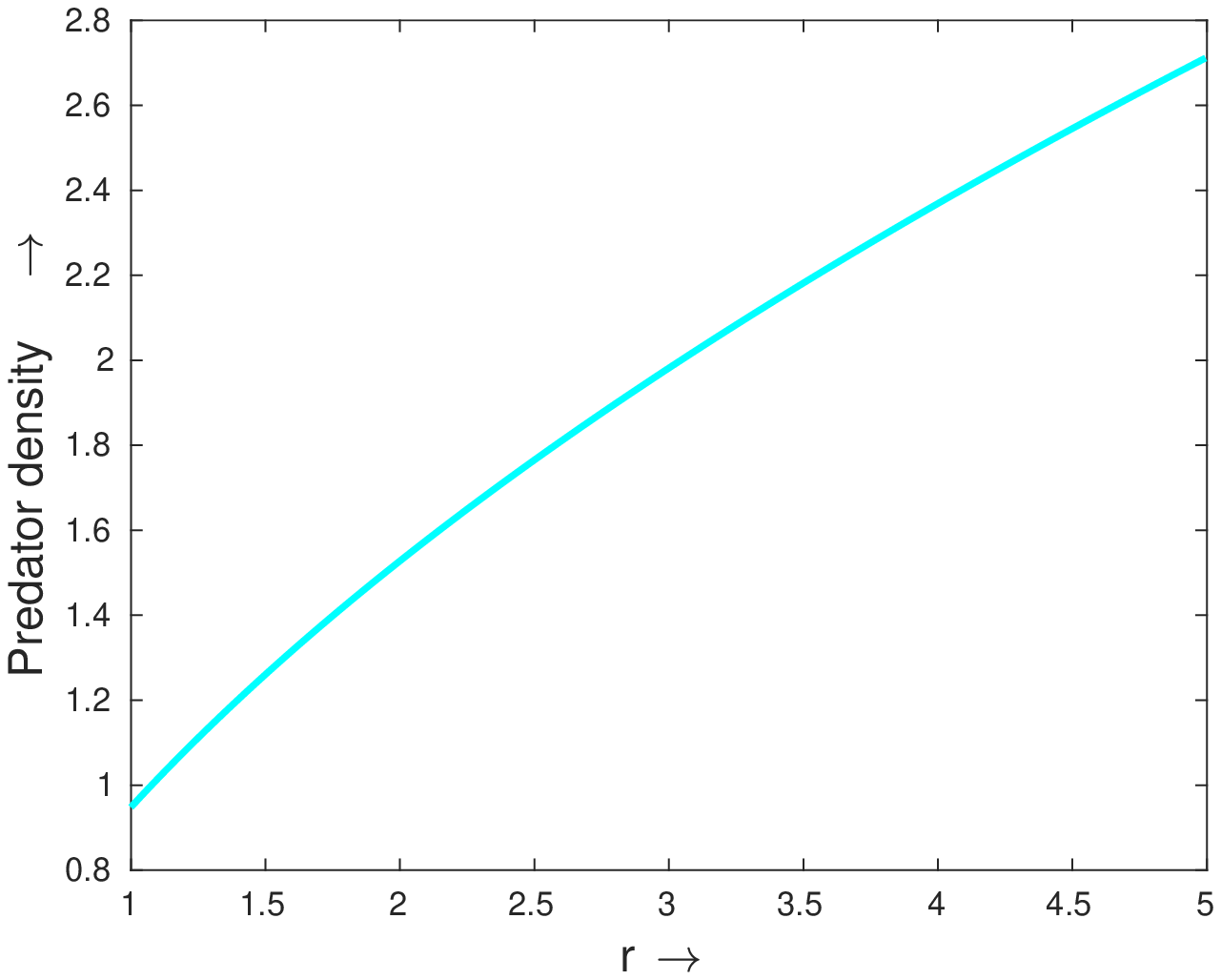}}
\caption{The figure depicts the effect of prey birth rate on equilibrium state of species. Here, all parametric values are taken from set (1) in Table \ref{tab 1}.}\label{Fig.7}
\end{figure}
\begin{figure}[H]
\subfloat[Region plot for extinction criteria of predator species ]{\includegraphics[height=7cm,width=8.2cm]{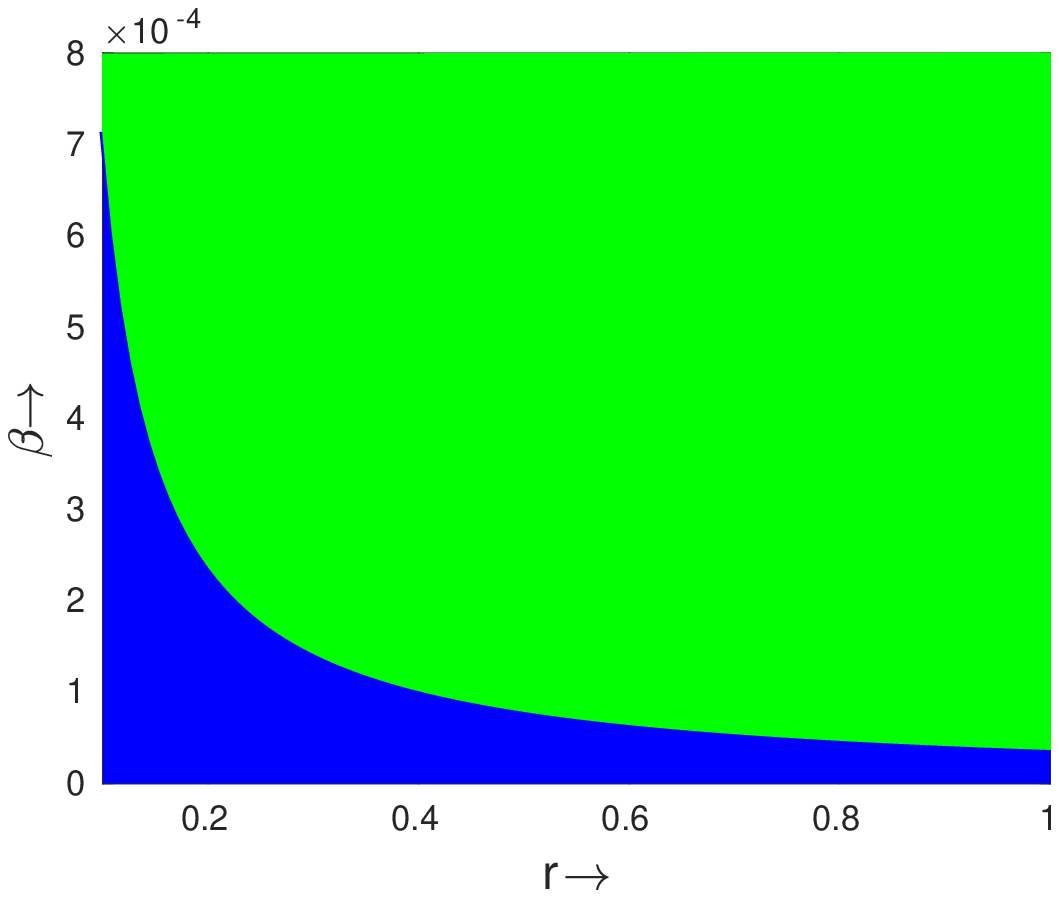}\label{Fig.9(a)}}
\qquad
\subfloat[Long term dynamics of the model system \eqref{7}.]{\includegraphics[height=7cm,width=8.1cm]{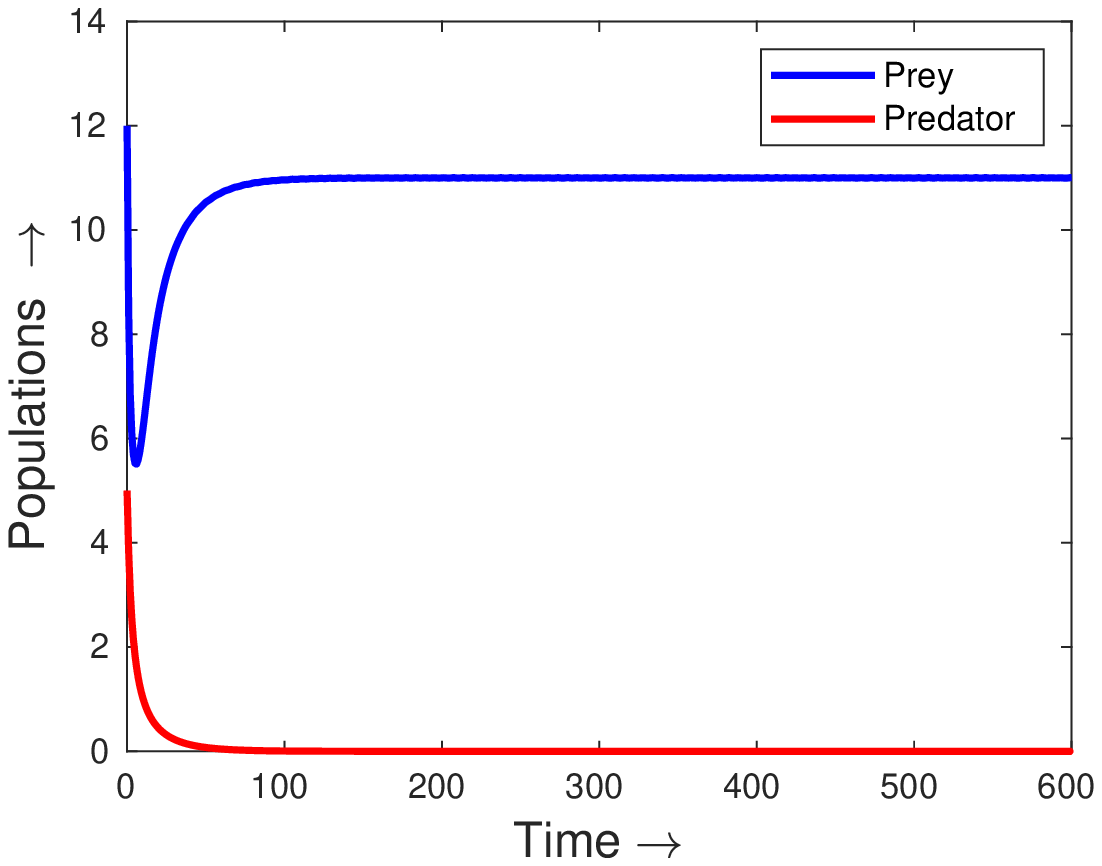}\label{Fig.9(b)}}
\caption{The figure ensures the extinction of predator species from prey-predator system whenever values $ r$, and $\beta$  have been chosen from  blue shaded region of Fig. \ref{Fig.9(a)}, i.e. $(r=0.6,, \beta= 0.00005) $.  Here,  all parameter values, apart from $r$ and $\beta$, are chosen from set (1) of Table \ref{tab 1}.}\label{Fig.9}
\end{figure}
The extinction of any species may arise the imbalance in the ecosystem. Therefore, its study is an important aspect of prey-predator interaction. In Fig. \ref{Fig.9(a)}, the blue-shaded region is the region of extinction for predator species in $(r-\beta)$ space. This blue-shaded region corresponds to the parametric condition for the extinction of predators stated in Theorem \ref{thm.extinction}\ref{thm.6}. The green-shaded region of $(r-\beta)$ space in the Fig. \ref{Fig.9(a)} is the region of uncertainty for predator species extinction, because predator species may or may not go extinct in this region. In Fig. \ref{Fig.9(b)}, we observe that the solutions curves of the model system \eqref{7} converge to predator extinct equilibrium. For this time series plot, we have chosen the numerical values of $ r$, and $\beta$ from the blue shaded region which ensure the extinction of predator species. Therefore, from Fig. \ref{Fig.9}, we conclude that the predator species goes to extinction for the small values of $r$ and $\beta.$

In Fig. \ref{Fig.8}, we observe  the occurrence of periodic solution of the system  \eqref{7} with a gradual increment in the level of fear (related parameter is $k$). From Fig.  \ref{Fig.8(a)}, we notice that at a low level of fear, predator species extinct from the system. As the level of fear increases, periodic solutions occur, but the system remains unstable. But the higher level of fear eliminates the periodic solutions and stabilises the system. Further, we have that the interior equilibrium point $E^*_2$ is unstable for k=0.509554, so the solution spirals out, as shown in the Fig. \ref{Fig.8(d)}. This is also demonstrated by corresponding time series plot as shown in Fig. \ref{Fig.8(b)} and \ref{Fig.8(c)} (refer to blue curves). The gradual increment in the fear enforces the occurrence of periodic solutions for the system \eqref{7}. When fear reaches the Hopf bifurcation threshold $k=k_H=0.509642226$, an unstable limit cycle occurs around $E^*_2$ via hope bifurcation, as shown in Fig. \ref{Fig.8(e)}. The presence of periodic oscillations in the time series corresponding to phase diagram Fig. \ref{Fig.8(e)} has also been shown in Fig. \ref{Fig.8(b)} and \ref{Fig.8(c)} by magenta curves. With another gradual increase in the level of fear, the unstable limit cycle surrounding the stable interior equilibrium point $E^*$ disappears and the solution spirals into $E^*_2$ for the appropriate initial conditions  as shown in Fig. \ref{Fig.8(f)} and red curves of Fig. \ref{Fig.8(b)} and \ref{Fig.8(c)}. 
\begin{figure}[H]
\subfloat[]{\begin{minipage}[c][1\width]{0.3\textwidth}\centering{\includegraphics[height=5cm,width=5cm]{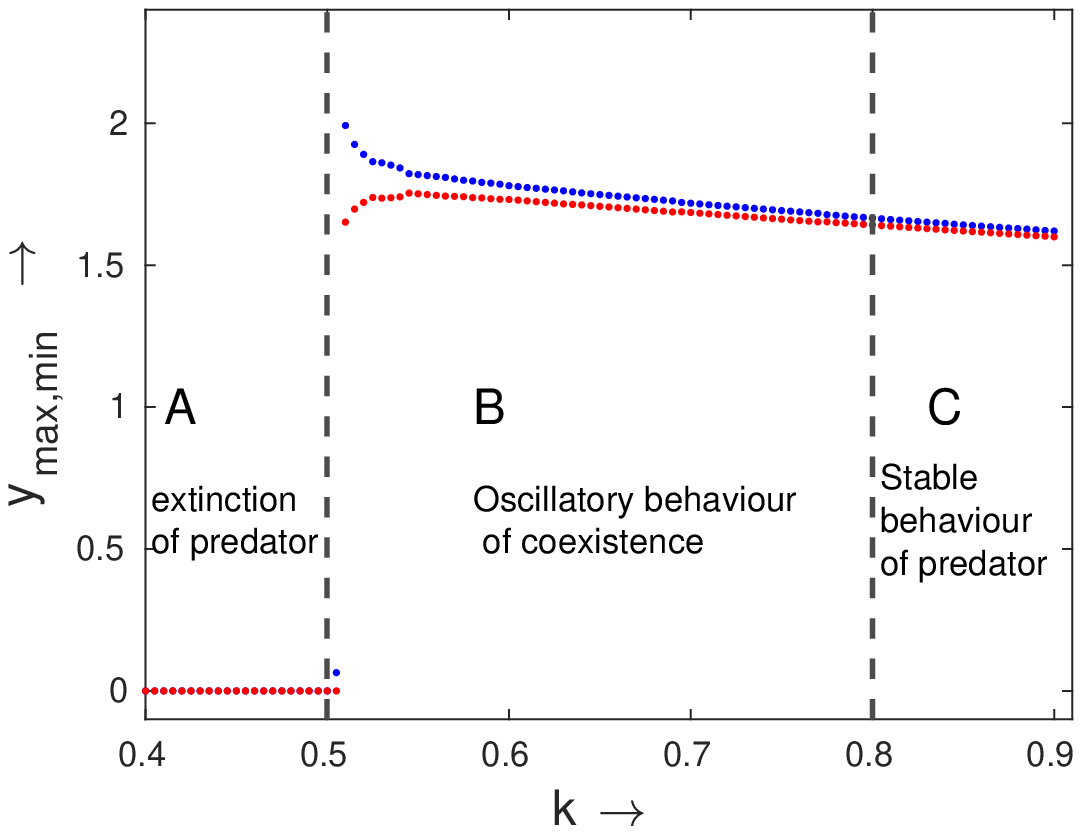}}\label{Fig.8(a)}
\end{minipage}}
\hfill
\subfloat[]{\begin{minipage}[c][1\width]{0.3\textwidth}\centering{\includegraphics[height=5cm,width=5cm]{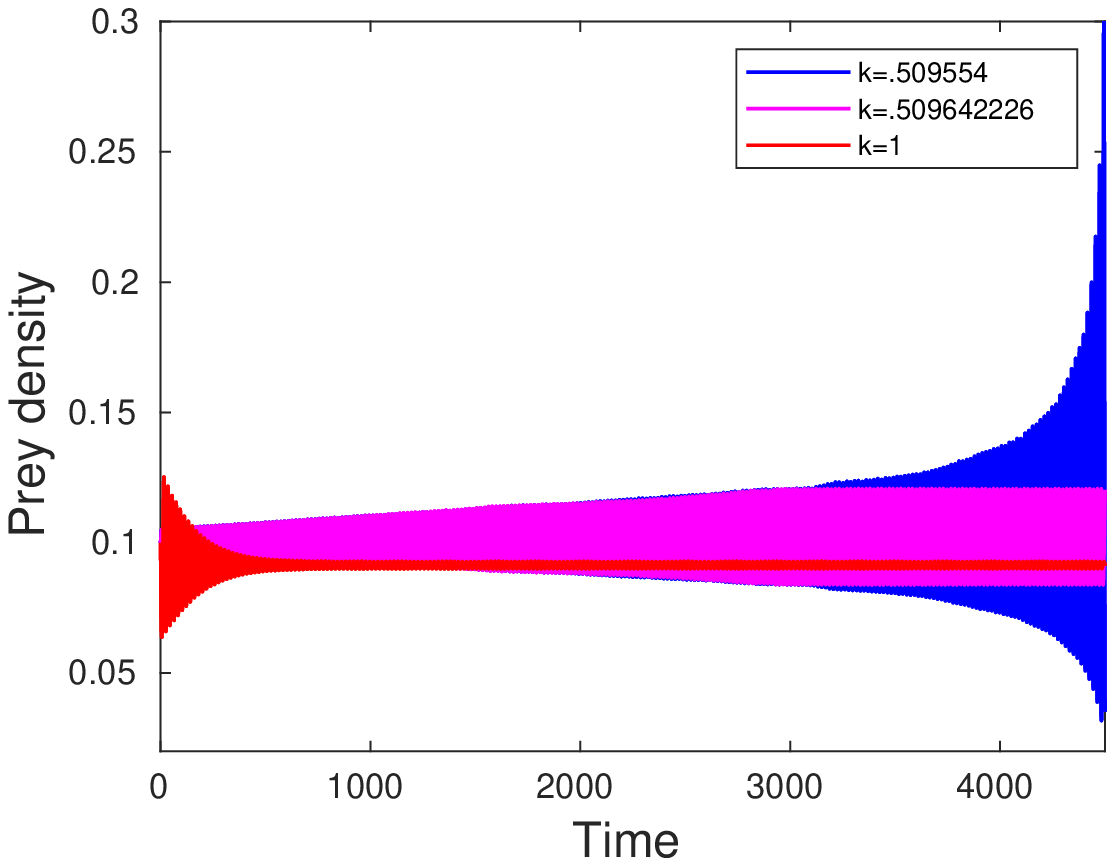}\label{Fig.8(b)}}
\end{minipage}}
\hfill
\subfloat[]{\begin{minipage}[c][1\width]{0.3\textwidth}\centering{\includegraphics[height=5cm,width=4.9cm]{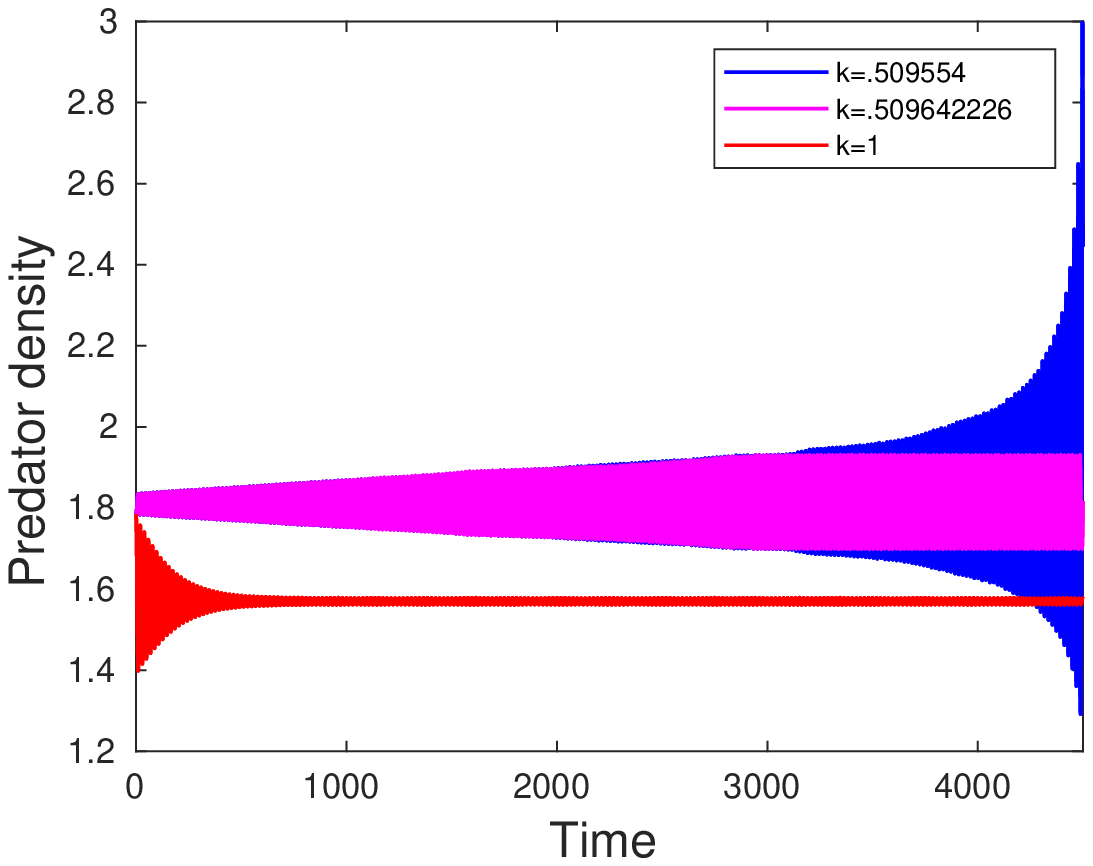}\label{Fig.8(c)}}
\end{minipage}}\hfill
\subfloat[]{\begin{minipage}[c][1\width]{0.3\textwidth}\centering{\includegraphics[height=5cm,width=5cm]{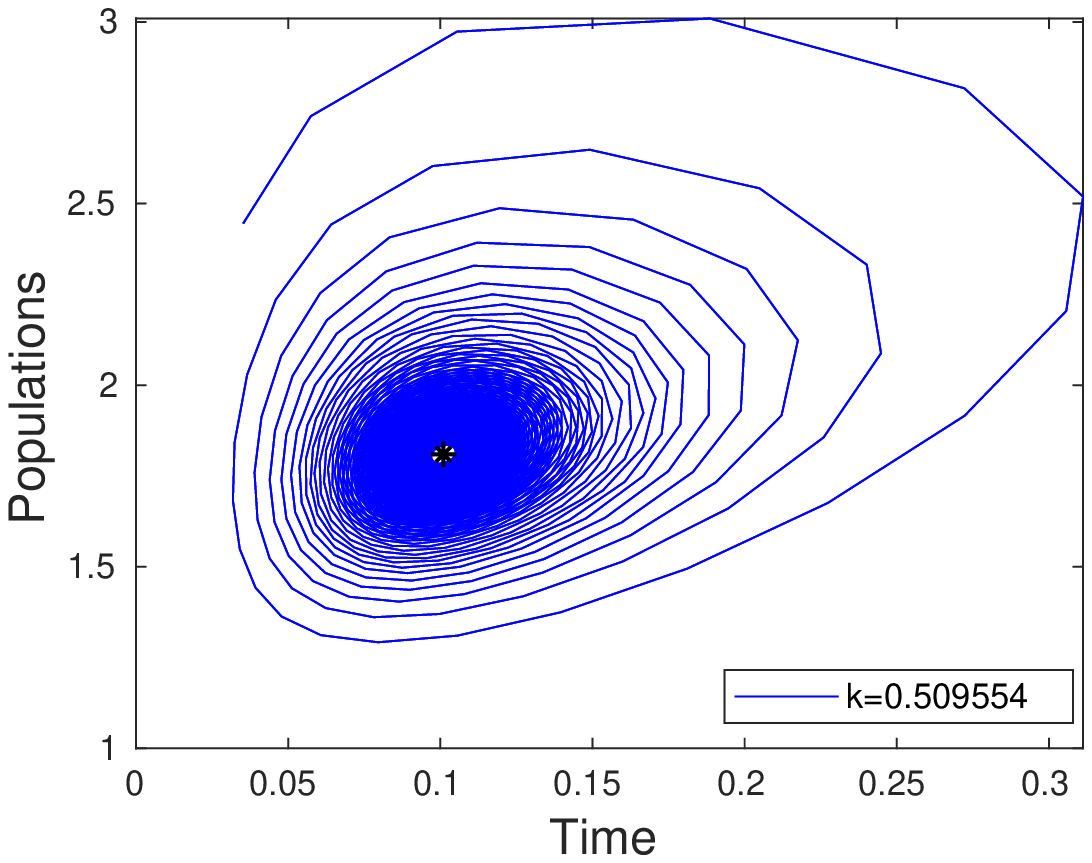}\label{Fig.8(d)}}
\end{minipage}}
\hfill
\subfloat[]{\begin{minipage}[c][1\width]{0.3\textwidth}\centering{\includegraphics[height=5cm,width=5cm]{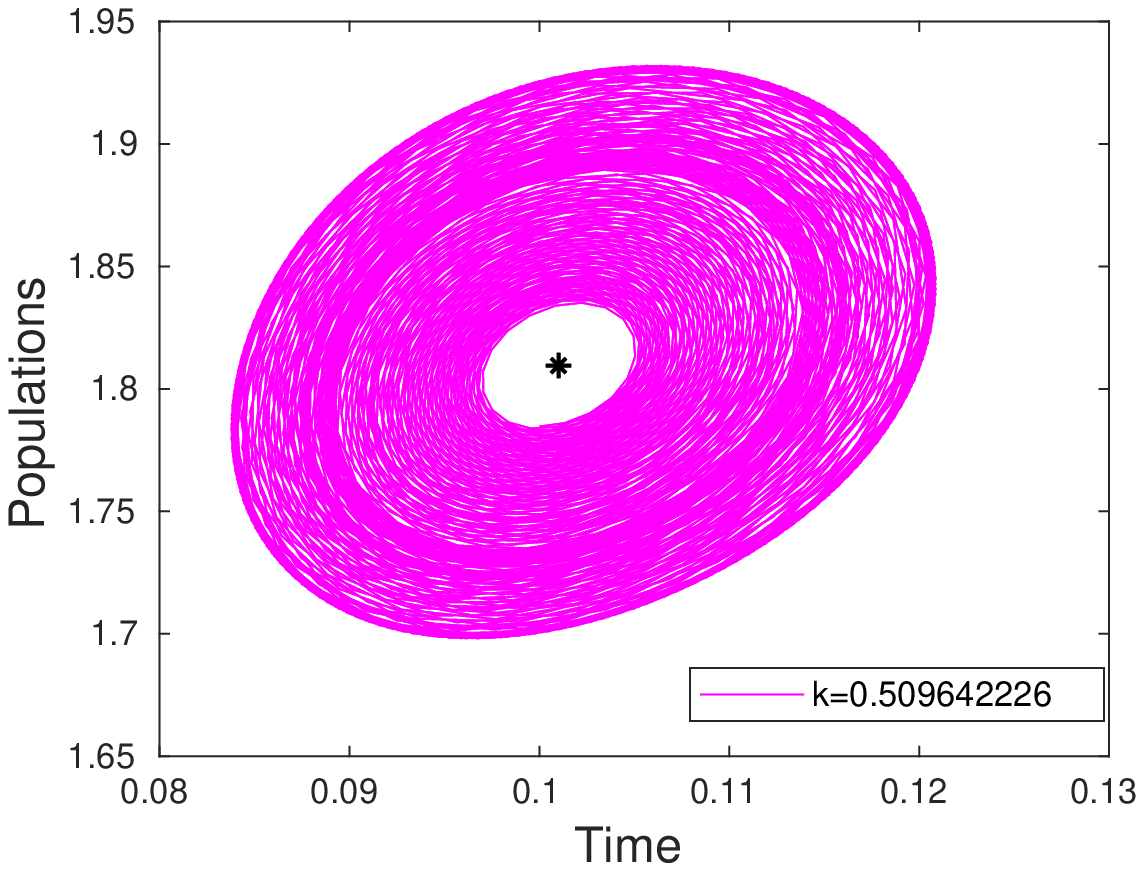}\label{Fig.8(e)}}
\end{minipage}}
\hfill
\subfloat[]{\begin{minipage}[c][1\width]{0.3\textwidth}\centering{\includegraphics[height=5cm,width=5cm]{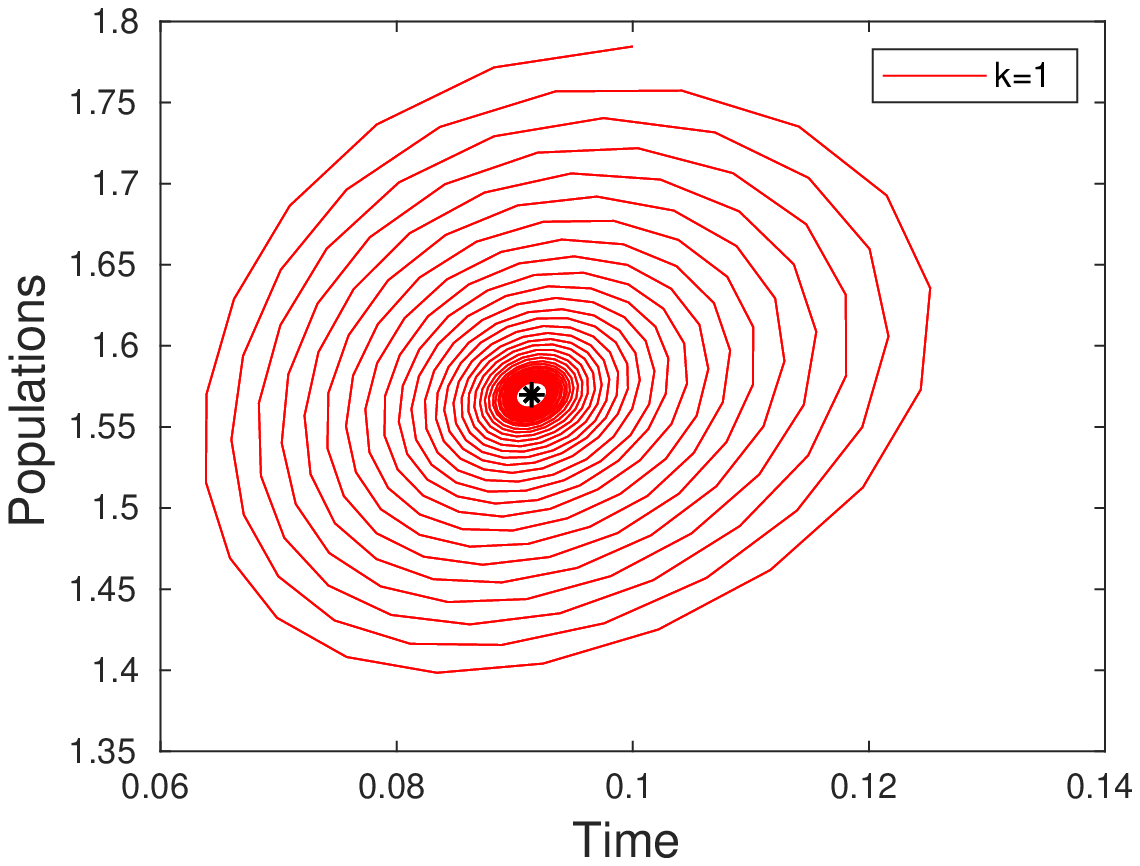}\label{Fig.8(f)}}
\end{minipage}}
\caption{All parameter values, apart from $\alpha$ = 0.25, c = 3 and k, are chosen from set (1) in Table \ref{tab 1}. Here, the initial condition is $(x_0, y_0)=(0.0999994016, 1.784623932176383)$. (a) The blue and red colour curves represent the time series of $y_{max}$ and $y_{min}$, respectively. The blue, magenta, and red curves in Figs.(\ref{Fig.8(b)})-(\ref{Fig.8(c)}) represent the time series corresponding to $E^*_2$ before the Hopf bifurcation threshold, at the Hopf bifurcation threshold, and after the Hopf bifurcation threshold, respectively. (d) It represents the instability of $E^*_2$ before the occurrence of the limit cycle through Hopf bifurcation. (e) It represents the existence of a limit cycle. (f) It represents the stability of $E^*_2$ after the Hopf bifurcation.}\label{Fig.8}
\end{figure}
The stability region for the predator extinct equilibrium point $E_1$ is represented by the blue shaded region of Fig. \ref{Fig.10(a)} in $(\alpha, \beta)$ space. If $(\alpha, \beta)$ is taken from the blue shaded region, it meets the locally asymptotically stable condition of the predator extinct equilibrium point $E_1$ stated in Theorem \ref{thm.8}\ref{thm. boundary stability}, and the predator species will eventually become extinct in the system. The cyan-shaded region in Fig. \ref{Fig.10(a)} is the region of instability for the predator extinct equilibrium point $E_1$.  A time series of a numerical example for the stability of a predator free equilibrium point is provided  in Fig. \ref{Fig.10(b)} by choosing the values of $(\alpha,\beta)$ from the blue shaded region. In Fig. \ref{Fig.10(b)}, we notice that the solution curves approach to predator extinct equilibrium $E_1$ as time passes. Therefore, it ensure the local stability of $E_1$ equilibrium.

\begin{figure}[H]
\subfloat[Region plot for extinction criteria of predator species in $(\alpha-\beta)$ space.]{\includegraphics[height=6cm,width=8.2cm]{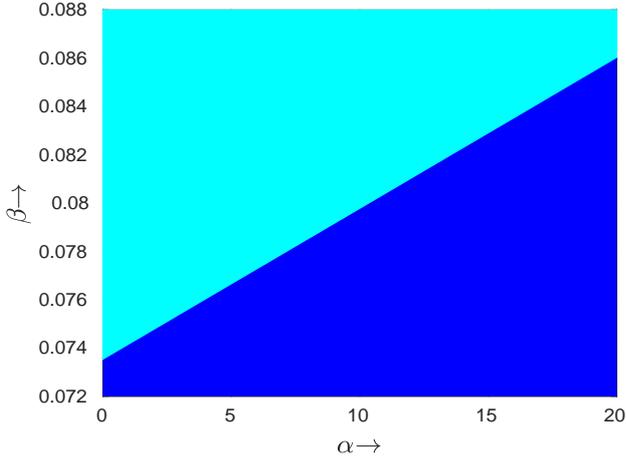}\label{Fig.10(a)}}
\qquad
\subfloat[Long term dynamics of the model system \eqref{7}.]{\includegraphics[height=6cm,width=8.1cm]{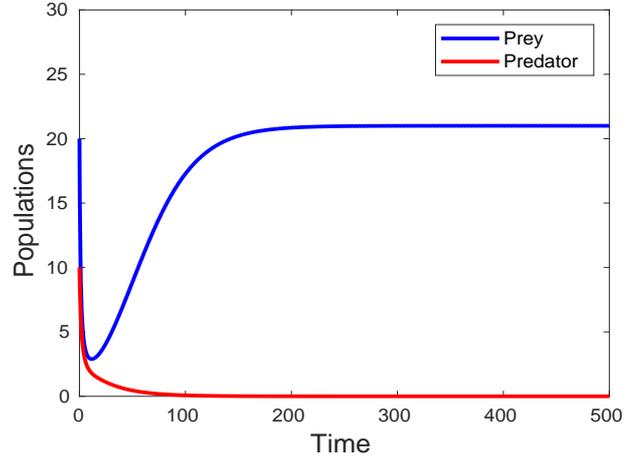}\label{Fig.10(b)}}
\caption{The figure ensures the local stability of predator extinct equilibrium $E_1$ of system \eqref{7} whenever values $ \alpha$, and $\beta$  have been chosen from blue shaded region of Fig. \ref{Fig.10(a)}, i.e. $(\alpha=5,, \beta= 0.025) $.  Here,  all parameter values, apart from $\alpha$ and $\beta$, are chosen from set (1) of Table \ref{tab 1}.}\label{Fig.10}
\end{figure}
 As shown in Fig. \ref{Fig.11(a)}, the blue-shaded region of $(\alpha-\beta)$ space is the region of persistence for prey species. If $(\alpha, \beta)$ is taken from the blue shaded region, it meets the persistence condition of the prey species stated in Theorem \ref{thm.7}\ref{thm.prey persistence} and ensures the survival of prey species in the system. The red-shaded region in Fig. \ref{Fig.11(a)} represents the region of uncertainty for prey species persistence, because prey species may persist or become extinct in this region. A time series of a numerical example for the persistence of prey species is provided in Fig. \ref{Fig.11(b)} by selecting values of $(\alpha,\beta)$ from the blue shaded region.
\begin{figure}[H]
\subfloat[]{\includegraphics[height=6cm,width=8.2cm]{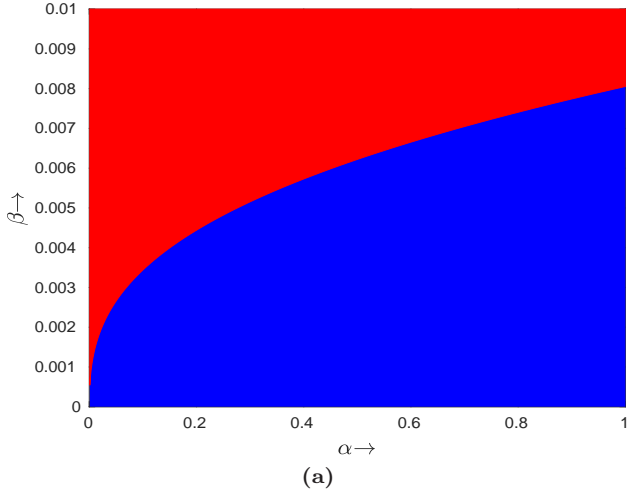}\label{Fig.11(a)}}
\qquad
\subfloat[]{\includegraphics[height=6cm,width=8.1cm]{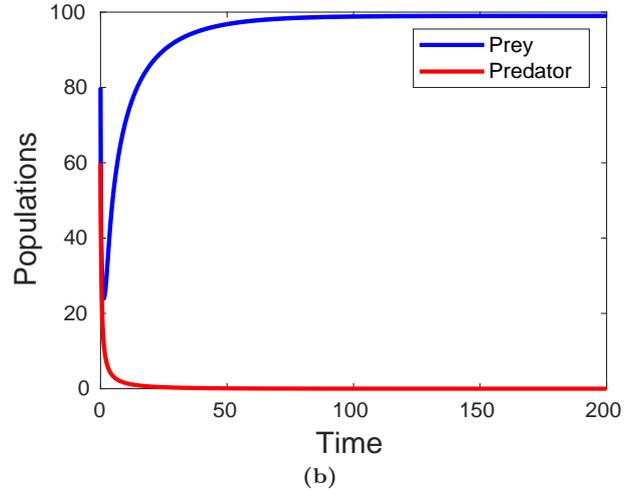}\label{Fig.11(b)}}
\caption{All parameter values, apart from $r=5$, $\alpha$ and $\beta$, are chosen from set (1) in Table \ref{tab 1}. (a) The blue-shaded region is the region of persistence for prey species. (b) The time series for a numerical example with $(\alpha, \beta)$=(0.5, 0.002) and initial condition (x(0), y(0))=(80, 60).}\label{Fig.11}
\end{figure}
\subsection{Numerical Simulation of Spatial Model}
In this subsection, we perform the numerical simulation of spatial model. In Fig. \ref{Fig.13}, we observe the effect of time evolution on the spatial density distribution of predator population. For time iteration $2000$, the density distribution of predator population shows the mixture of  spot and stripes like pattern. As we increase the value of time iteration to $5000$, the stripes and spot like pattern moves toward the irregular spots. The density distribution of predator population shows the mixture of stripes and spot like pattern at time evolution $7000$. Finally, at time iteration $10000$, the spots comes closure to each other and form the cluster of spots. Thus time evolution plays an important role on system dynamics of spatial model. Also we can observe that the upper and lower bound of density of predator population decreases with decrease in time evolution (refer to Colorbar Fig. \ref{Fig.13}). 
 Fig. \ref{Fig.14} depicts the impact of diffusion coefficient of predator population $D_{2}$ on the dynamics of spatial model system. Initially at $D_{2}=4.0$,  the density is distributed as mixture of stripe and spot pattern. We observe that stripes and spot are very near to each other. The spots in the pattern die out and form the irregular stripes at $D_{2}=6.0$. The population density of predator distributed in the form of spots at $D_{2}=8.0$.  
From Fig. \ref {Fig.15}, one can observe that the predator population is spatially  distributed in the form of mixture of spots and stripes at prey fear level $k=0.1$ and $k=0.3$. When we increase the value of  fear level as $k=0.6$ and $k=0.7$, the spots come closure to each other and form irregular stripe like pattern. Also we have found that the  upper and lower density of population decreases with increase in prey fear level parameter $k$ (refer colorbar of Fig. \ref {Fig.15} ). We can also observe that the density of predator population decreases with increase in fear level. Thus we can say that fear factor have much impact on spatial distribution of predator population and coexistence of species. Also from Fig. \ref{Fig.Insta(b)}, one can observe that the lower value of fear effect is responsible for stability of model system as higher value of fear level leads the occurrence of Turing instability. Thus fear level acts as control parameter for system dynamics of spatial model system. 
Fig. \ref{Fig.16} describes the impact of prey refuge on the spatial distribution of predator population. Initially we have obtained the mixture of irregular spot and stripe like pattern and higher density of  predator population is distributed in almost whole domain of consideration. However, when we increase the value of prey refuge there is no much change in the pattern and higher density of predator are not distributed in whole domain. Also the the upper and lower bound of density of predator increases with increase in prey refuge (refer Fig. \ref{Fig.16(a)}-\ref{Fig.16(c)}).
\begin{figure}[H]
\subfloat[]{\includegraphics[height=7cm,width=8.2cm]{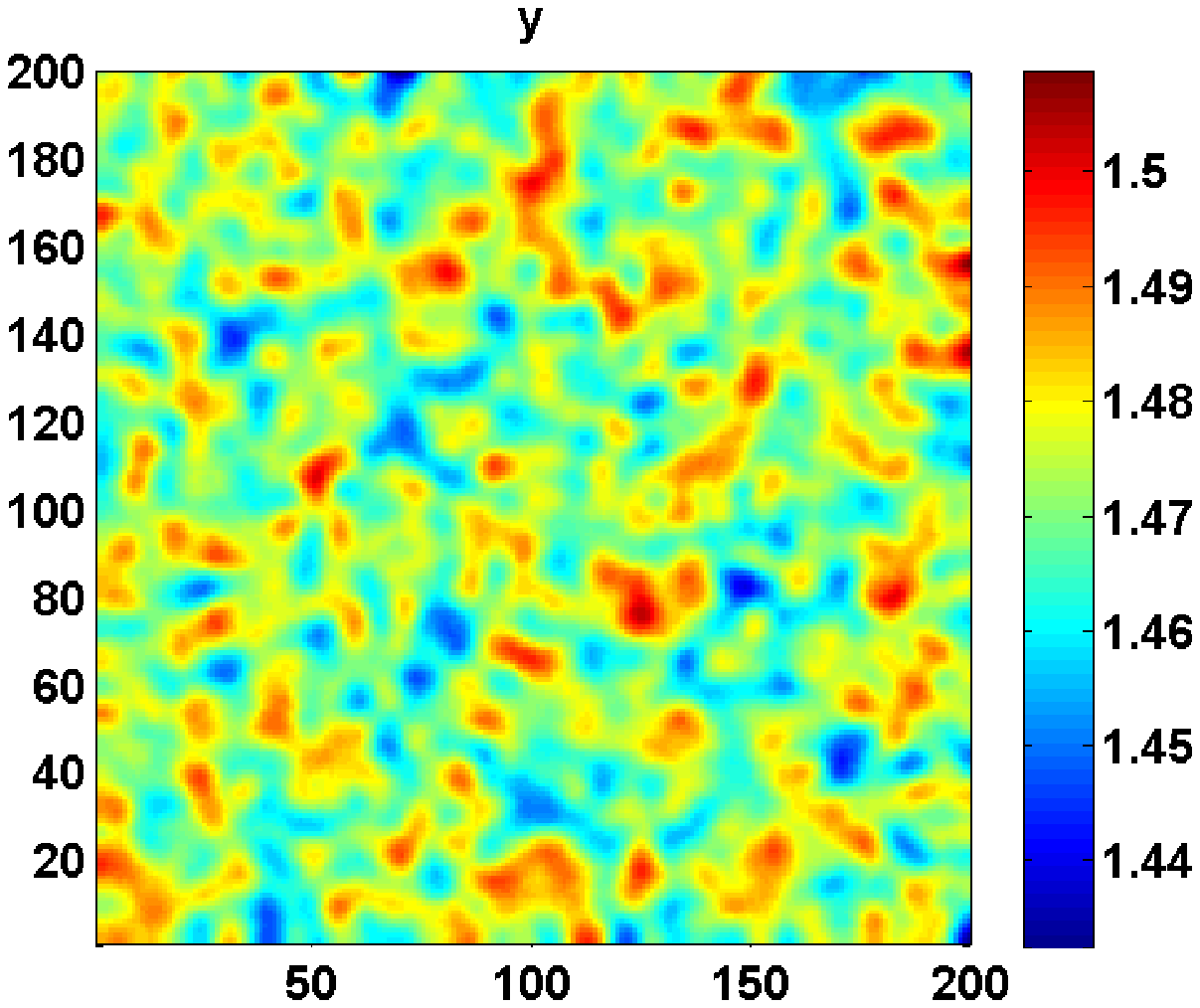}\label{Fig.13(a)}}
\qquad
\subfloat[]{\includegraphics[height=7cm,width=8.1cm]{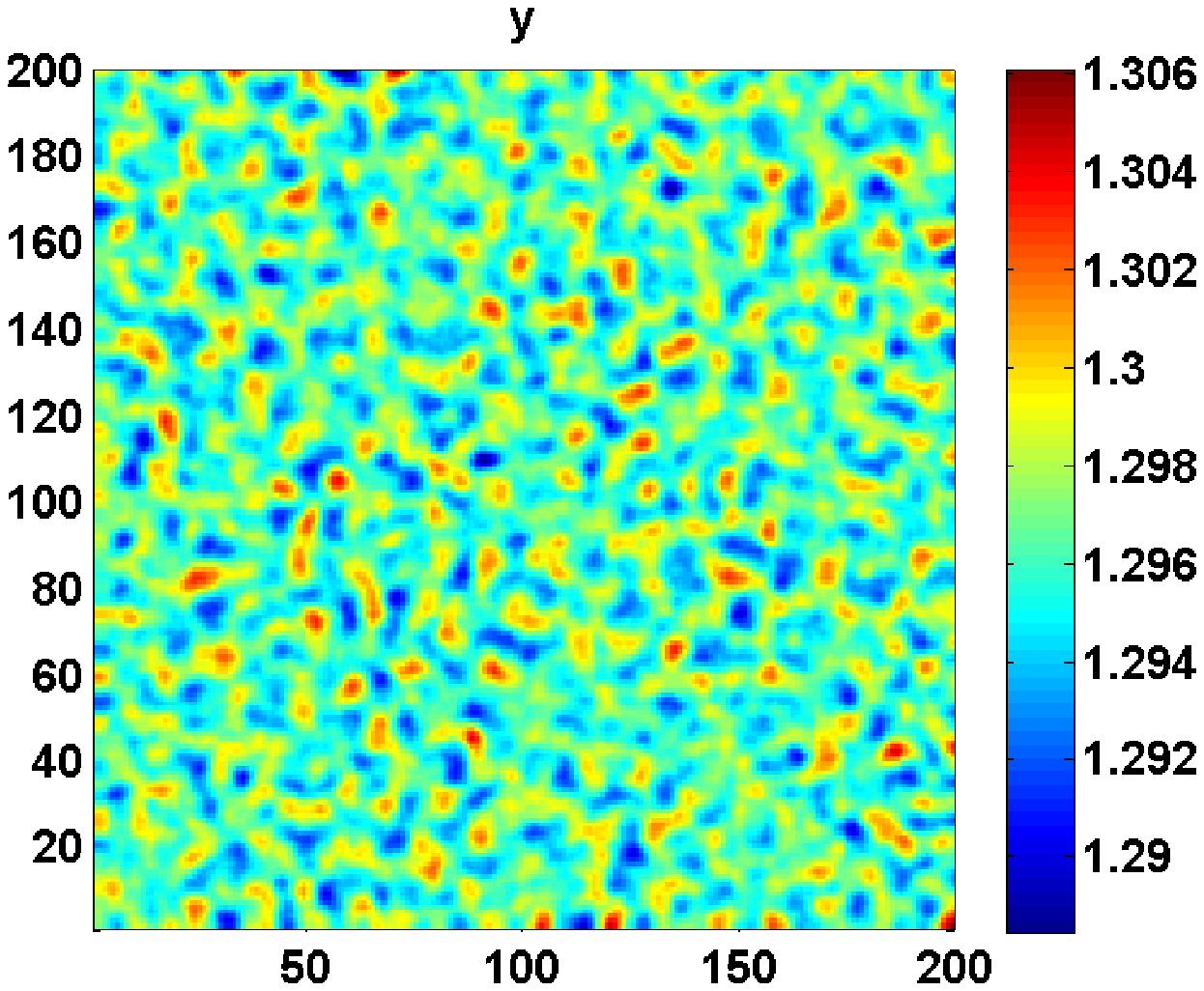}\label{Fig.13(b)}}
\qquad
\subfloat[]{\includegraphics[height=7cm,width=8.1cm]{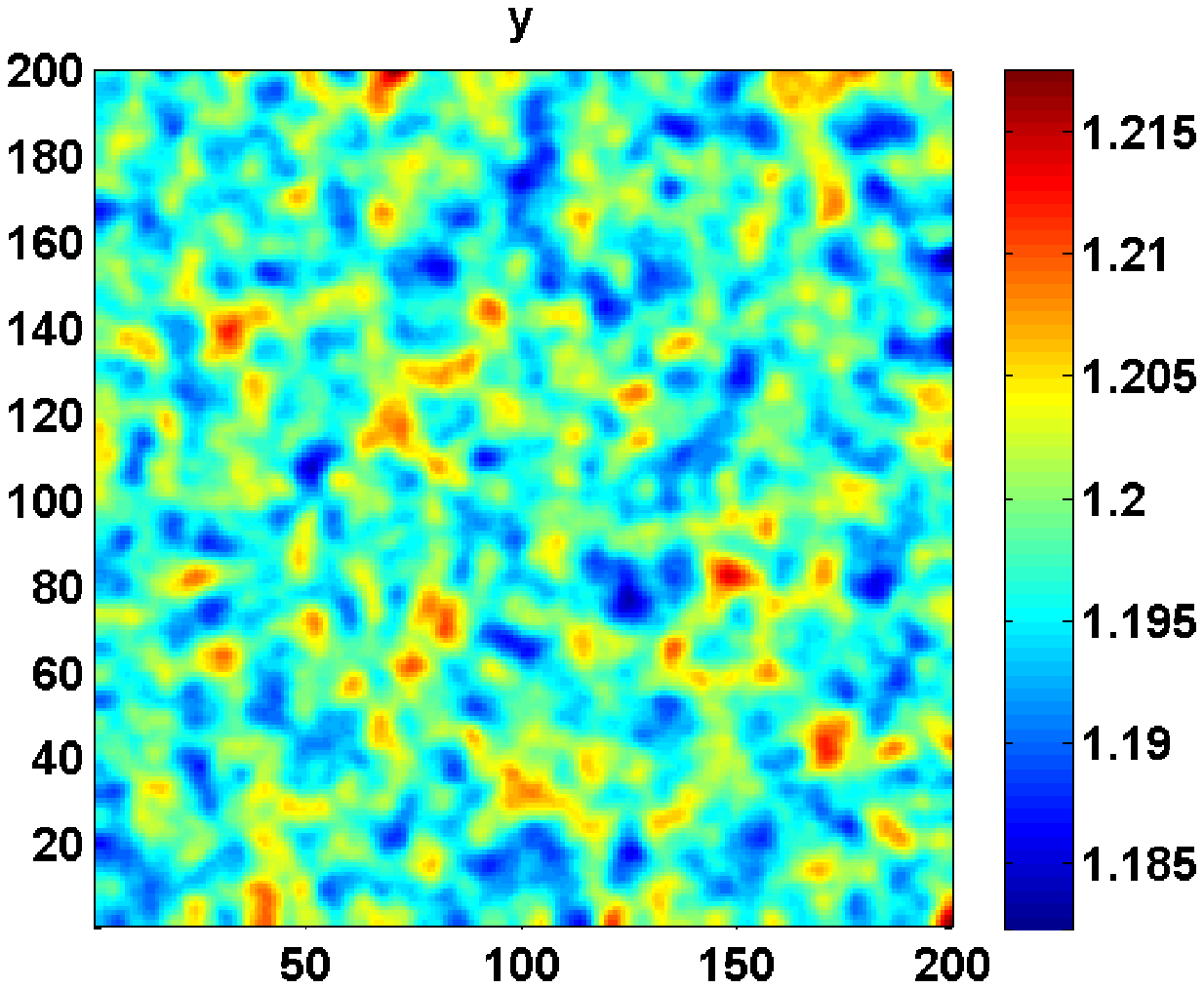}\label{Fig.13(c)}}
\qquad
\subfloat[]{\includegraphics[height=7cm,width=8.1cm]{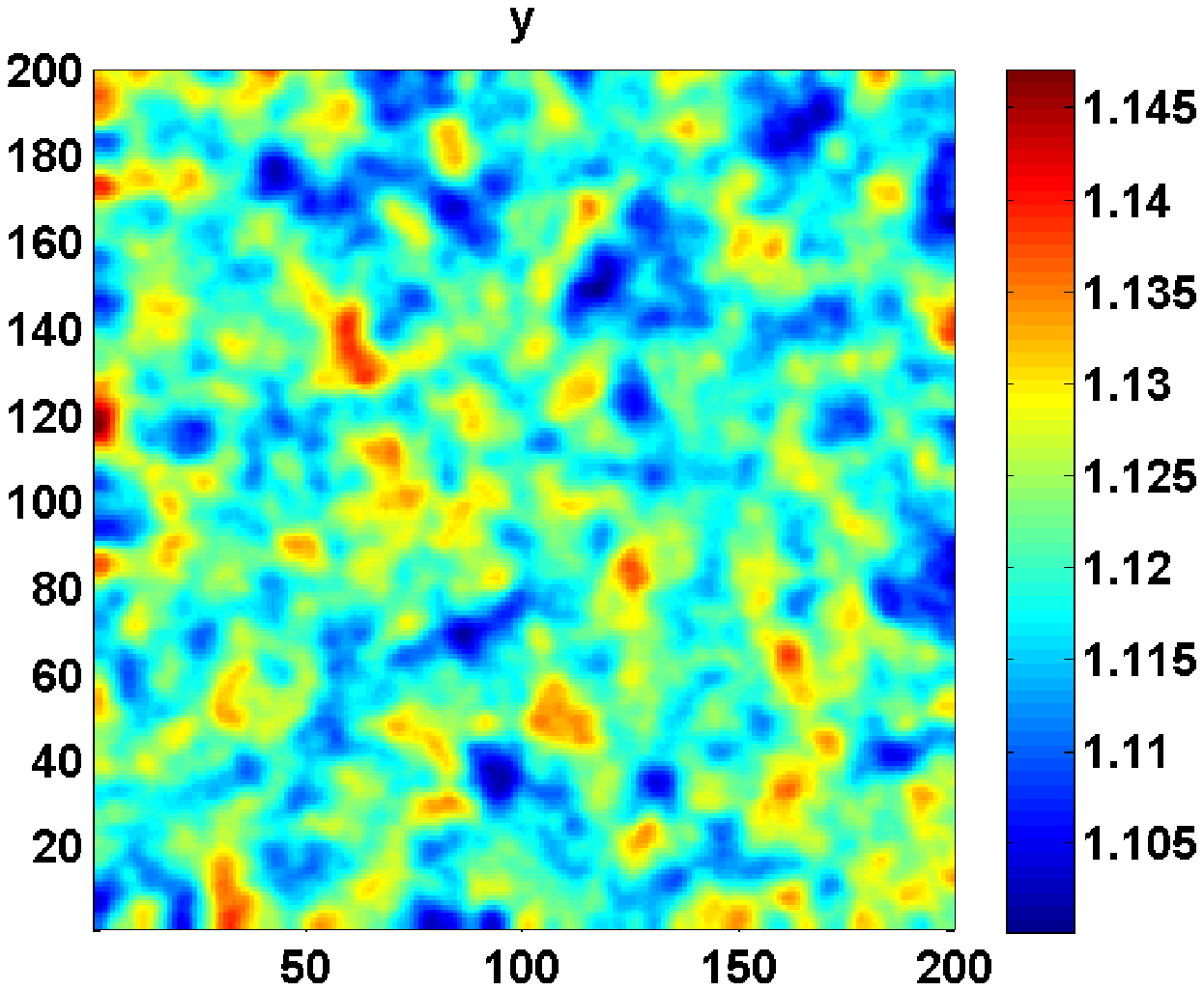}\label{Fig.13(d)}}
\caption{Spatial distribution of predator $y$ in $u-v$ plane at different values of iteration (a) $2000$, (b) $5000$, (c) $7000$ and (d) $10000$
for model system \eqref{18}. For the values of other parameters, refer to Eq. \eqref{29}.} \label{Fig.13}
\end{figure}
\begin{figure}[H]
\subfloat[]{\includegraphics[height=7cm,width=8.2cm]{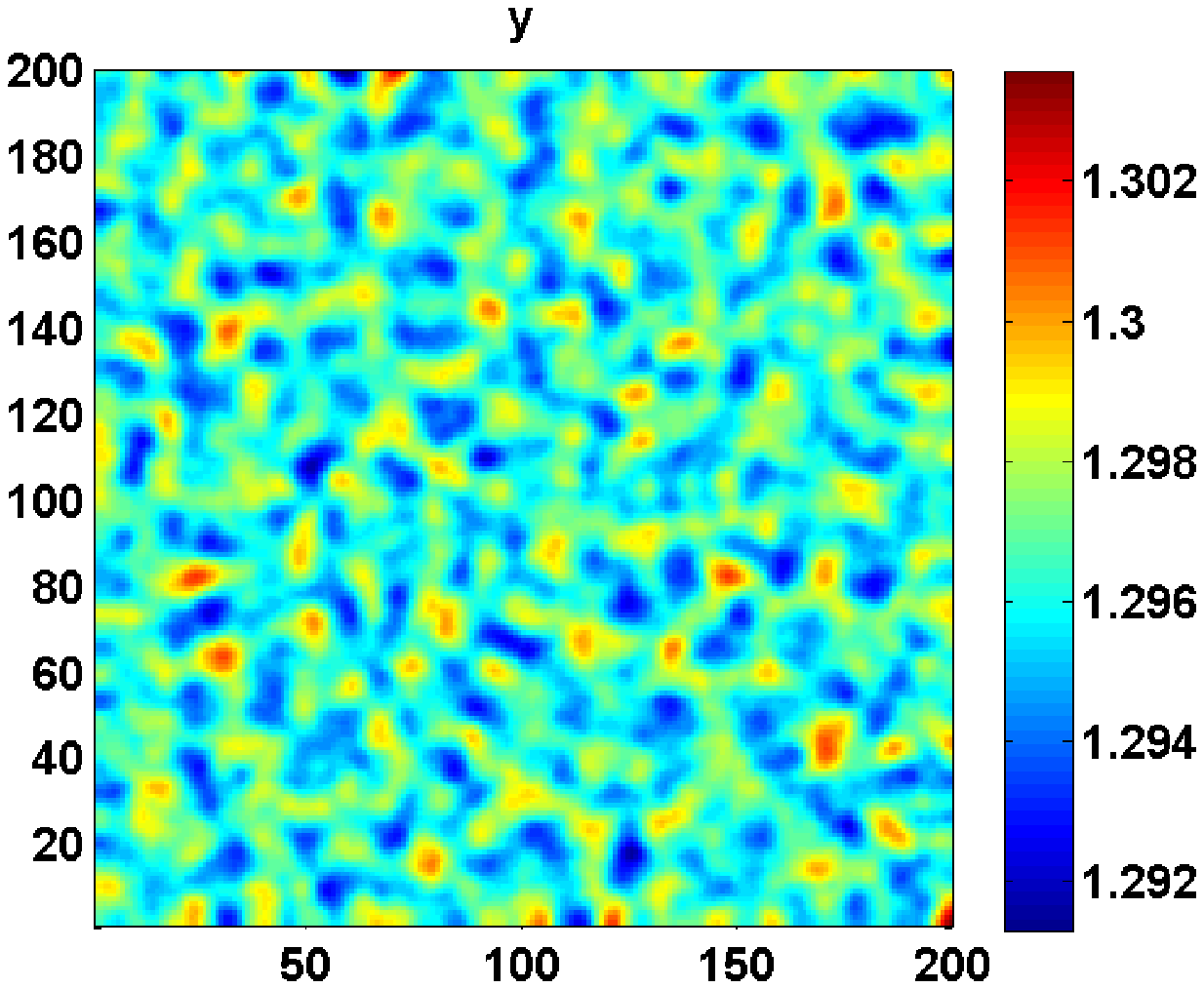}\label{Fig.14(a)}}
\qquad
\subfloat[]{\includegraphics[height=7cm,width=8.1cm]{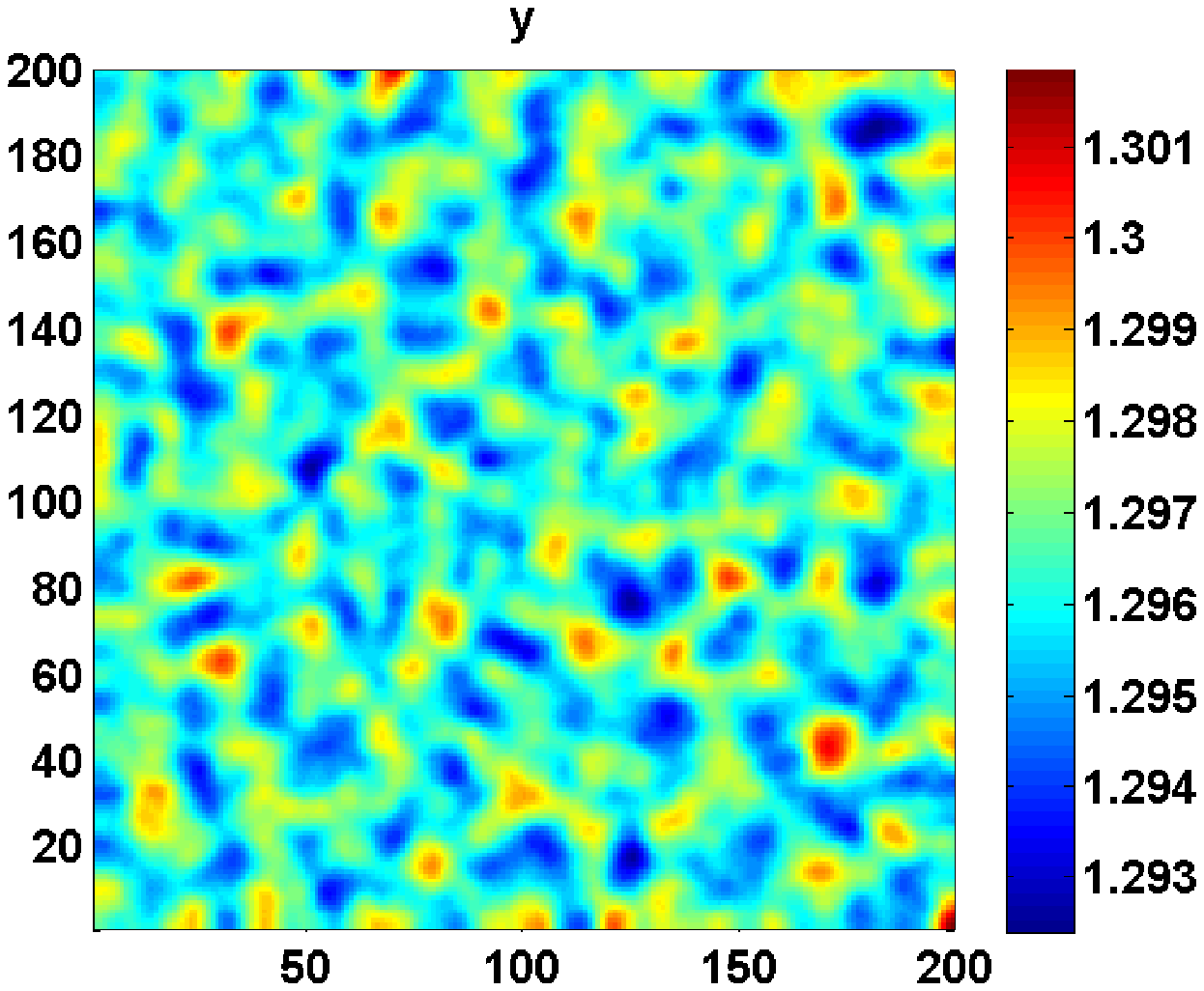}\label{Fig.14(b)}}
\begin{center}
\qquad
\subfloat[]{\includegraphics[height=7cm,width=8.1cm]{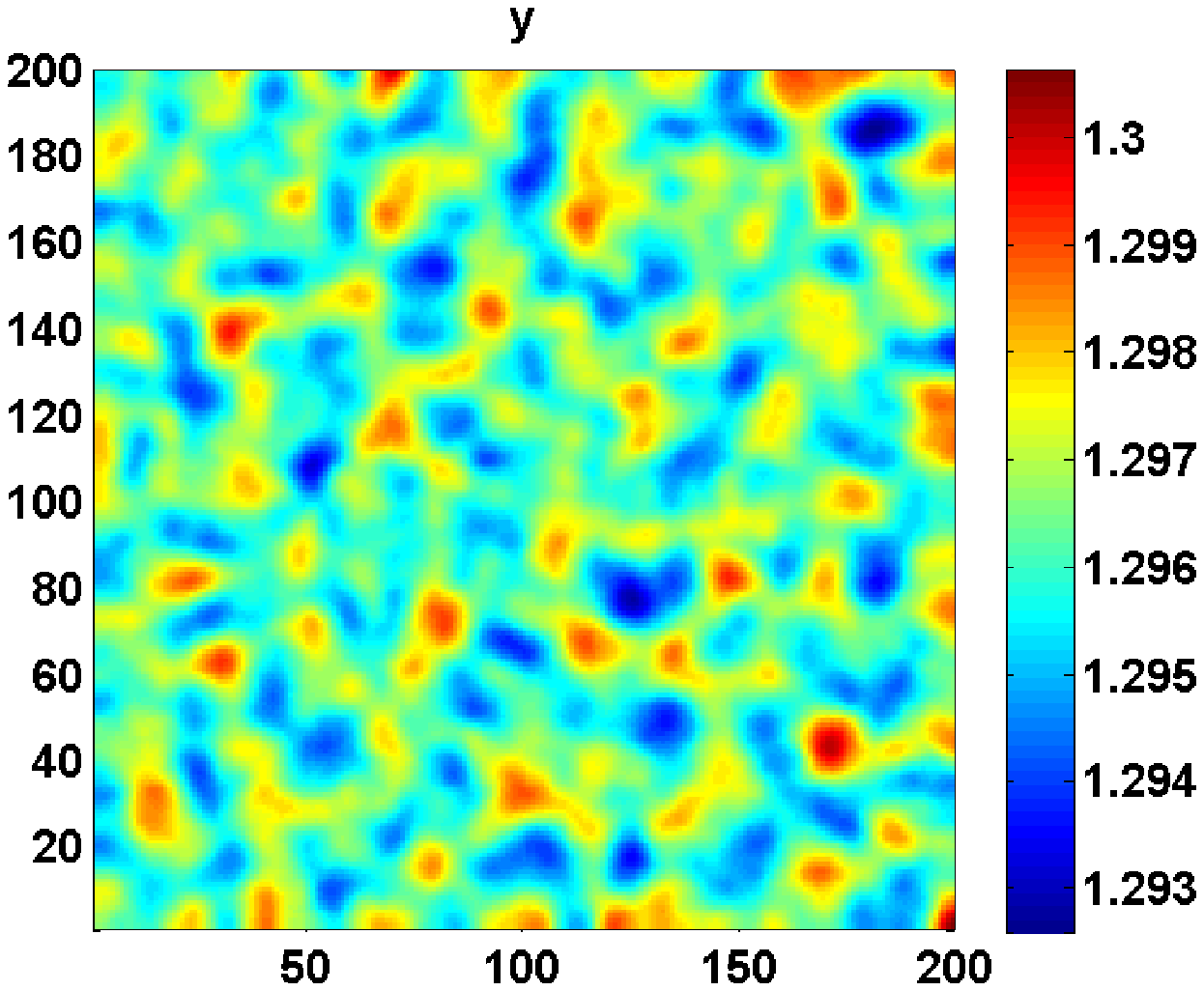}\label{Fig.14(c)}}
\caption{Spatial distribution of predator $y$ in $u-v$ plane at time iteration $5000$ and different values of diffusion coefficient of predator $D_{2}$ (a) $D_{2}=4.0$, (b) $D_{2}=6.0$ and (c) $D_{2}=8.0$
for model system \eqref{18}. For the values of other parameters, refer to Eq. \eqref{29}.}\label{Fig.14}
\end{center}
\end{figure}
\begin{figure}[H]
\subfloat[]{\includegraphics[height=7cm,width=8.2cm]{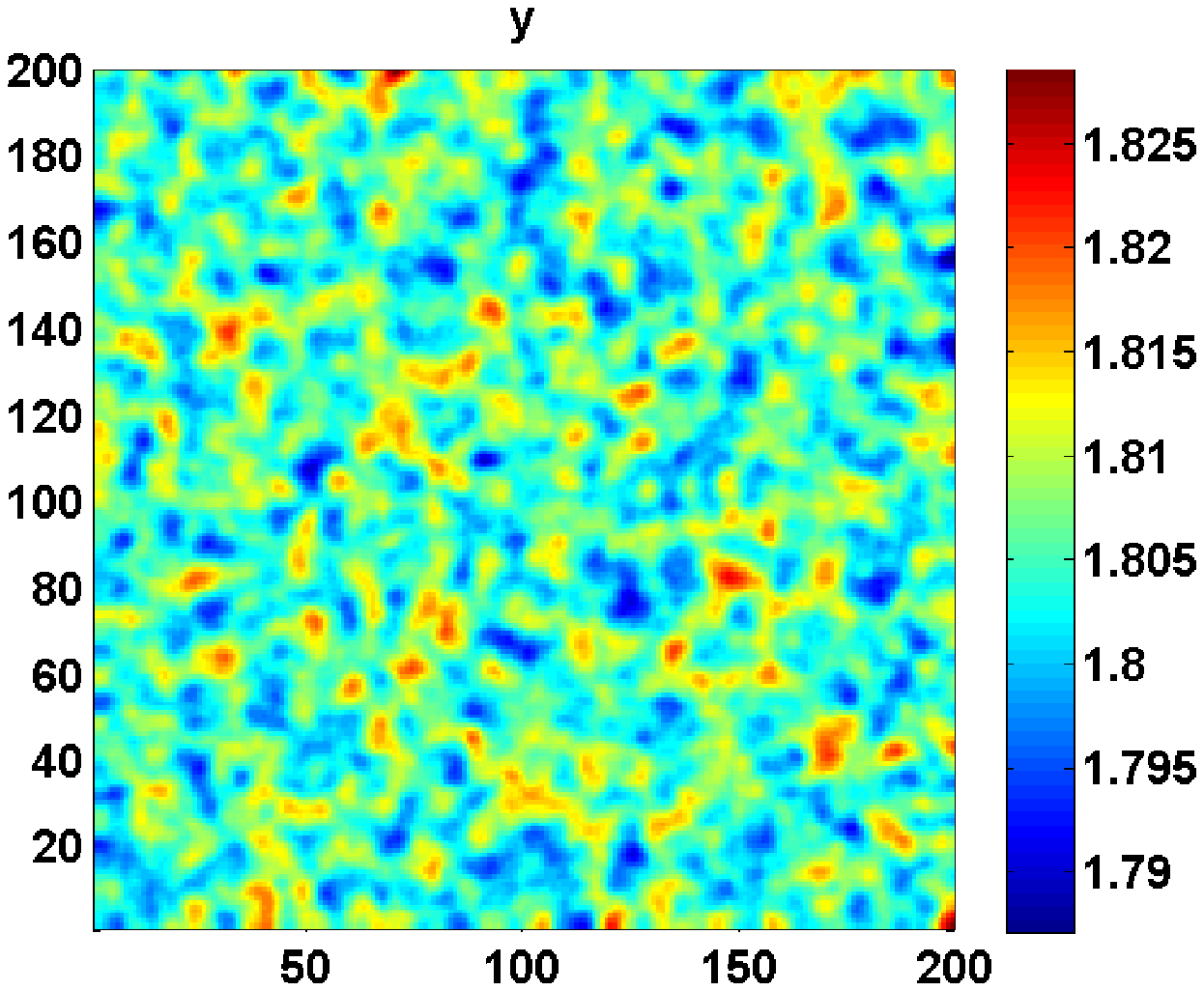}\label{Fig.15(a)}}
\qquad
\subfloat[]{\includegraphics[height=7cm,width=8.1cm]{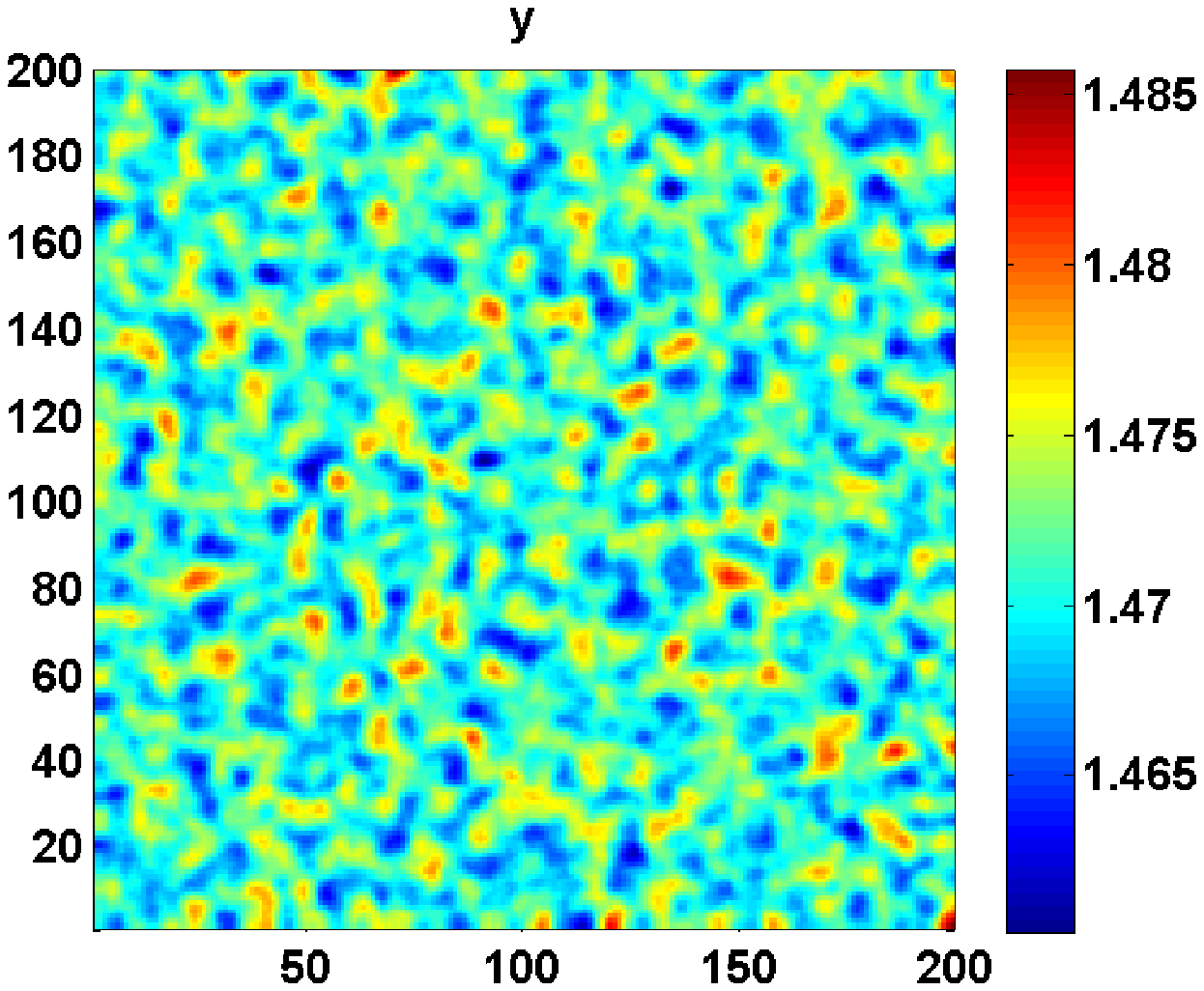}\label{Fig.15(b)}}
\qquad
\subfloat[]{\includegraphics[height=7cm,width=8.1cm]{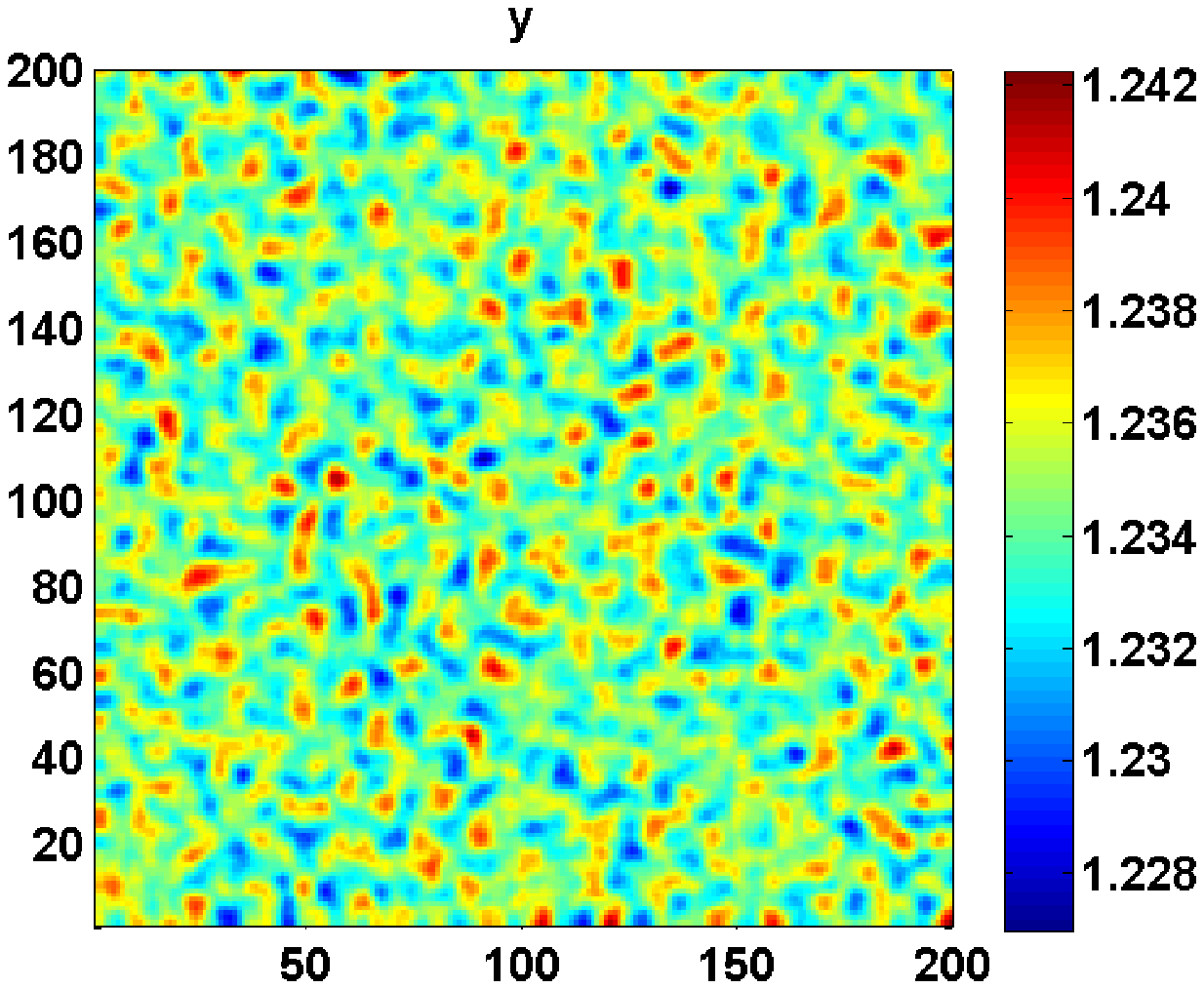}\label{Fig.15(c)}}
\qquad
\subfloat[]{\includegraphics[height=7cm,width=8.1cm]{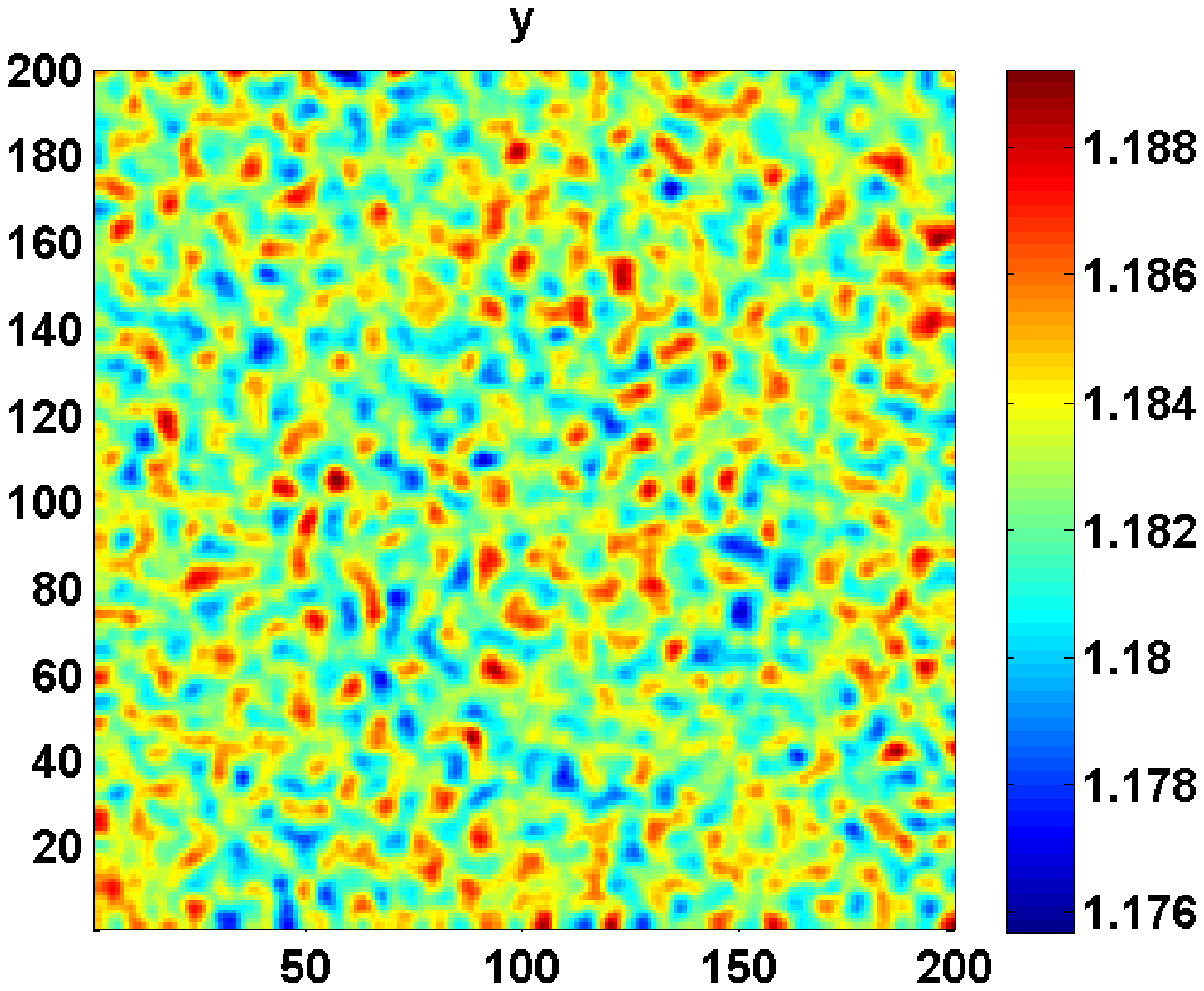} \label{Fig.15(d)}}
\caption{Spatial distribution of predator $y$ in $u-v$ plane at time iteration $5000$ and different values of fear level parameter (a) $k=0.1$, (b) $k=0.3$, (c) $k=0.6$ and (d) $k=0.7$ for model system \eqref{18}. For the values of other parameters, refer to Eq. \eqref{29}.}\label{Fig.15}
\end{figure}
\begin{figure}[H]
\subfloat[]{\includegraphics[height=7cm,width=8.2cm]{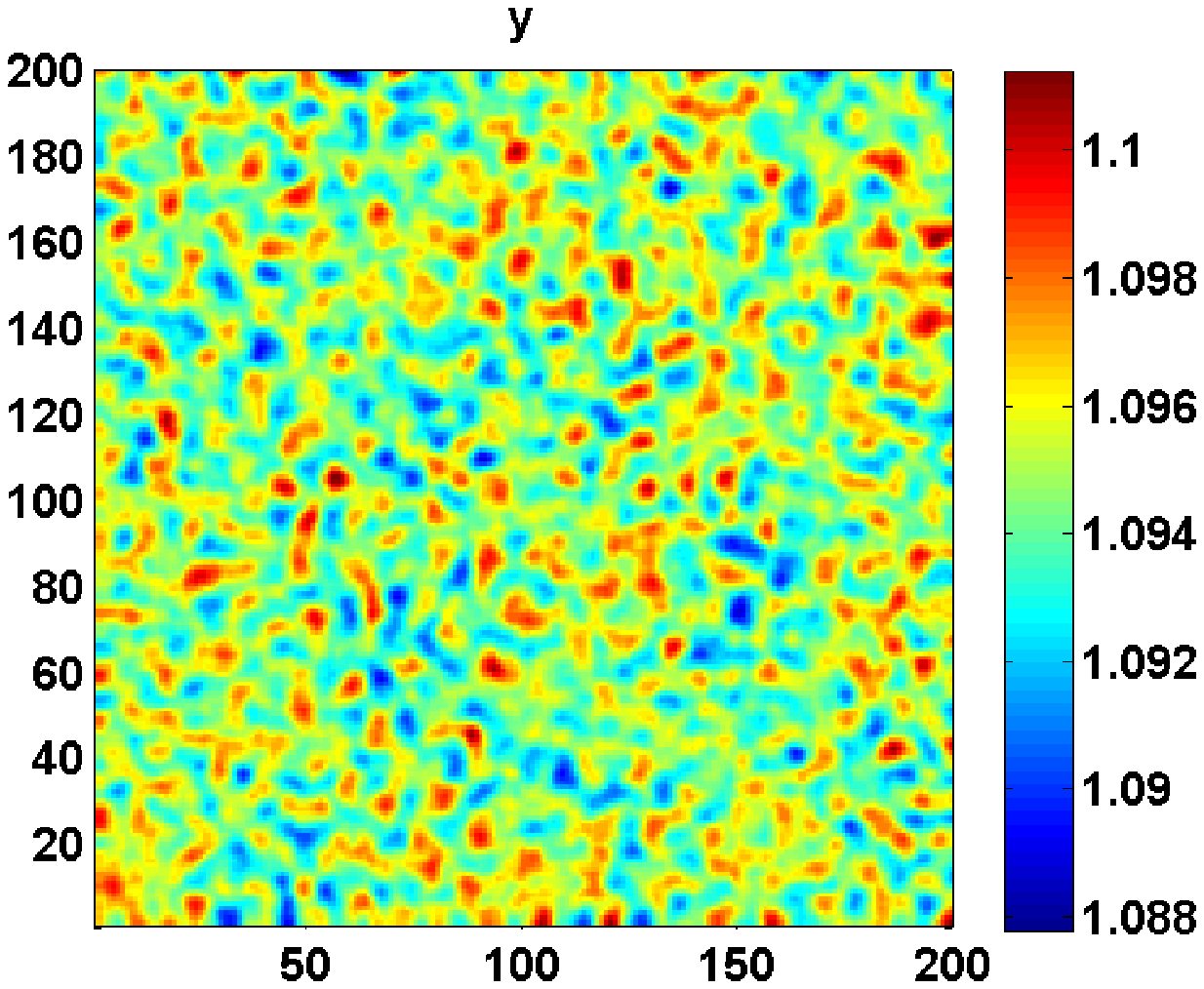}\label{Fig.16(a)}}
\qquad
\subfloat[]{\includegraphics[height=7cm,width=8.1cm]{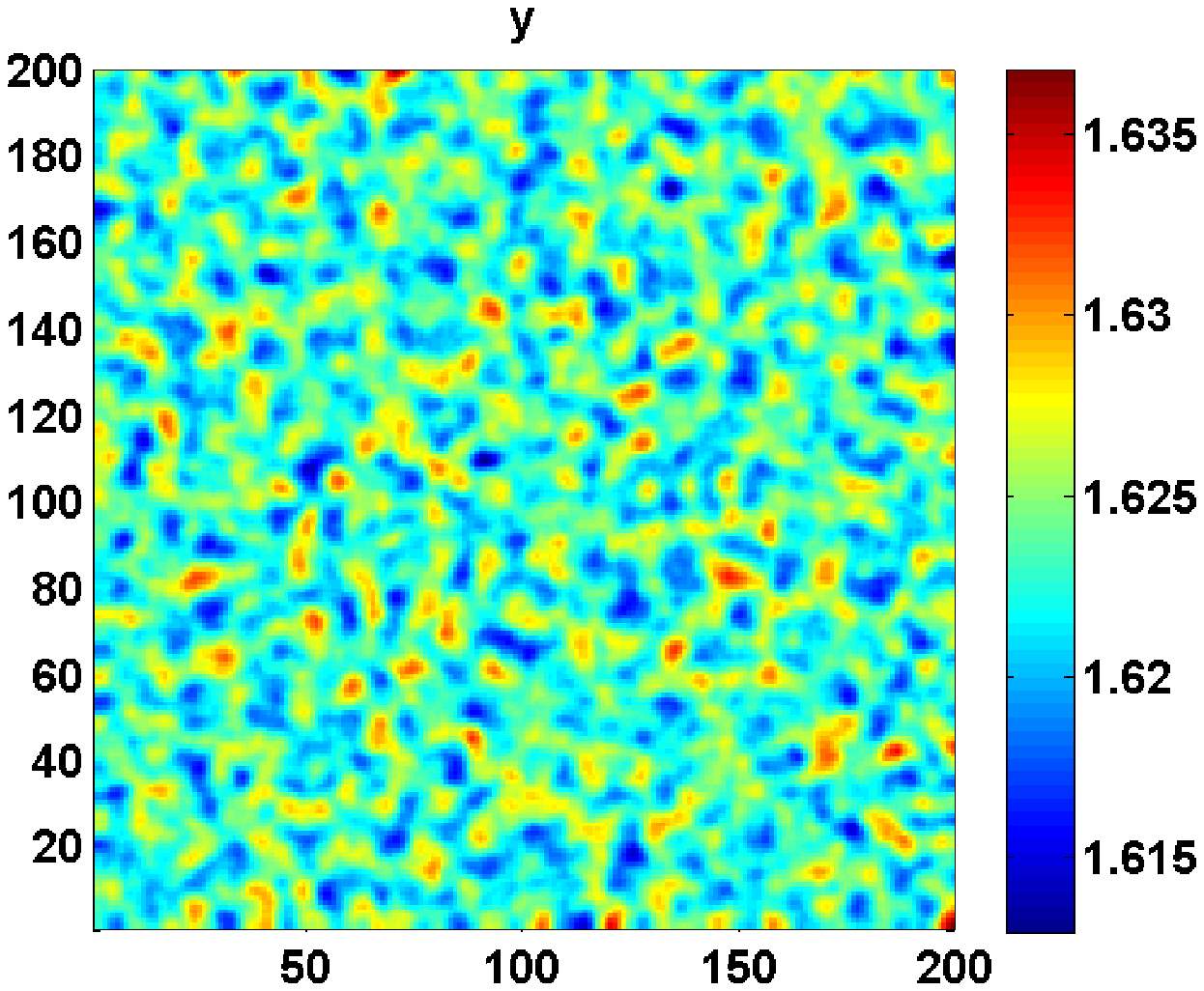}\label{Fig.16(b)}}
\begin{center}
\qquad
\subfloat[]{\includegraphics[height=7cm,width=8.1cm]{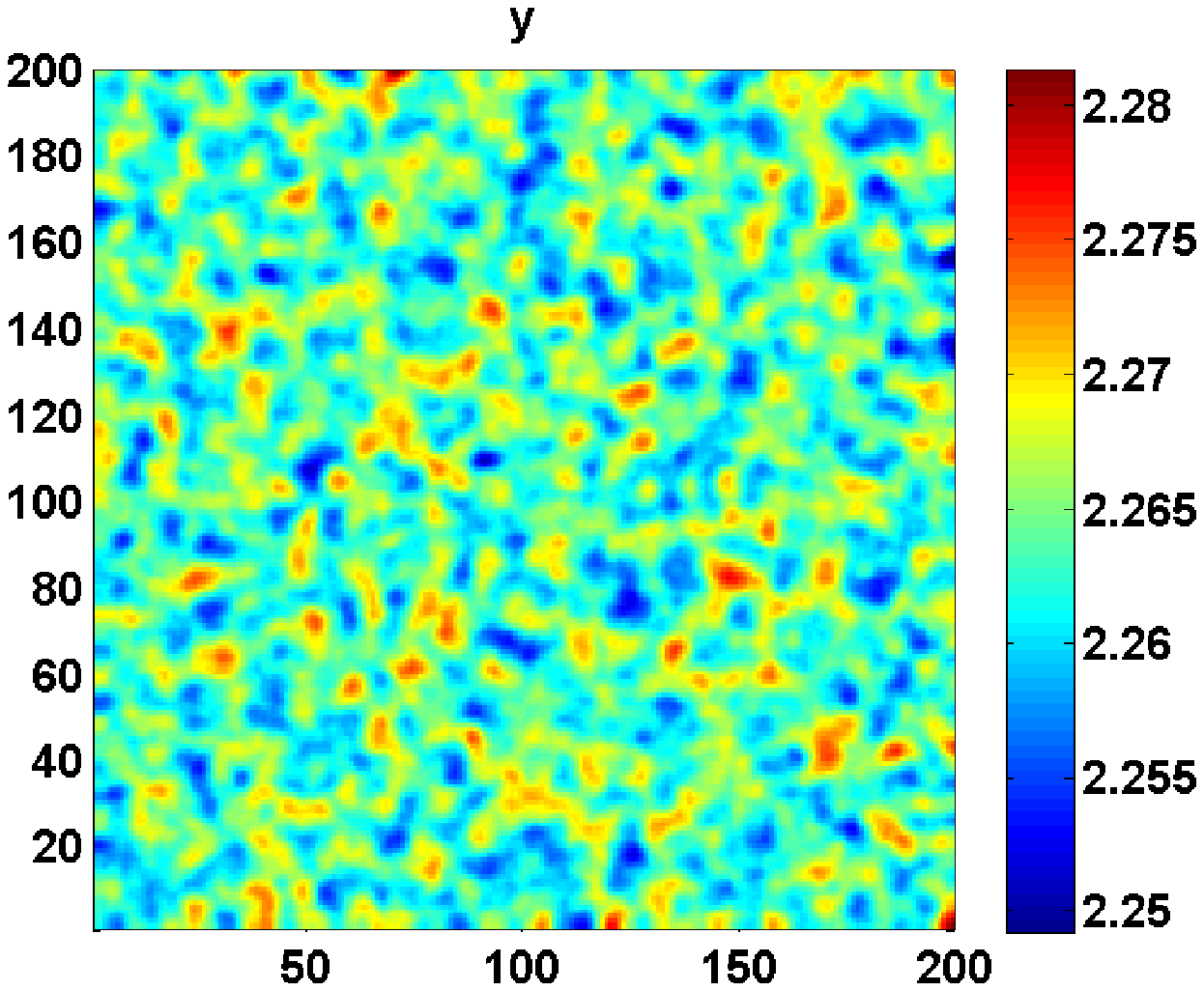}\label{Fig.16(c)}}
\caption{Spatial distribution of predator $y$ in $u-v$ plane at time iteration $5000$ and different values of prey refuge parameter (a) $m=0.1$, (b) $m=0.5$  and (c) $m=0.7$ for model system \eqref{18}. For the values of other parameters, refer to Eq. \eqref{29}.}\label{Fig.16}
\end{center}
\end{figure}
\section{Discussion}\label{Sec.8}
A driving force behind theoretical ecology has been the study of prey-predator interactions with a variety of significant mechanisms, which may help us to comprehend ecological systems in the real world more broadly. Prey-predator models have become more realistic over the last few years as lots of factors have been introduced into them, like the fear effect, group defence, density-dependent death rates of predators and prey density-dependent predation rates\cite{wang2016modelling,sasmal2020dynamics}. Wang et al. \cite{wang2016modelling} studied a prey-predator model with fear effect incorporating functional response of Holling type \RomanNumeralCaps{2}. He explored that high level of fear stabilises the system by removing the periodic solutions. The author also found the existence of subcritical Hopf bifurcation which was not possible in traditional prey-predator model without fear effect. Zhang et al. \cite{xie2022influence} studied a prey-predator model incorporating fear effect, prey refuge, and Holling type \RomanNumeralCaps{3} functional response and found that the increasing value of fear stabilise the system by eliminating the oscillatory solutions. These studies have motivated us to investigate a prey-predator interaction incorporating the fear effect, prey refuge, and Holling type functional response \RomanNumeralCaps{4} to observe whether the fear effect has a stabilising effect on the system even when prey species show group defence. In our proposed model, the functional response of the Monod-Haldane type is meant to reflect the prey's group defence. The fear effect is assumed to result in a decrease in the birth rate of prey. Furthermore, the parameter of predator-taxis sensitivity is used to relate the group defence and decreased birth rate caused by the fear effect. Due to the prey's limited overall time and/or energy, if they devote more of it to group defence, their birth rate should decline\cite{sasmal2020dynamics}.

In mathematical analysis, we have systematically examined all attainable dynamics of the model system \eqref{7} including the positivity and boundedness of solutions, some scenarios for species persistence and extinction, and local stability and bifurcation around all potential equilibrium states. We have identified the adequate parametric criterion for the extinction of prey species in Theorem \ref{thm.extinction}\ref{thm.3}. It demonstrates that the population of prey could  extinct in spite of the absence of predator species.
The inclusion of a death rate of predators caused by intra-predator competition in the considered model system promotes the extinction of predator species. Theorem \ref{thm.extinction}\ref{thm.6} gives a sufficient parametric criterion for the extinction of predator species only. It demonstrates that as a result of the higher natural death rate of predators and the prey's death rate caused by intra-prey competition, predator species become extinct in the system. However, only prey species survive if they exist at all. The adequate criterion for the persistence of prey species has been discussed in Theorem \ref{thm.7}\ref{thm.prey persistence}. We have shown numerically that there exists a region in $(\alpha,\beta)$ space in which prey species persist in the system (refer Fig. \ref{Fig.11}). Further, we discussed the existence of each and every probable equilibrium along with their stability analysis, analytically  and numerically. According to Theorem \ref{thm.8}\ref{thm. trivial stability}, only the population extinct equilibrium point exists and is surrounded by a global basin of attraction when the prey's natural death rate is higher than its birth rate. In this circumstance, both species perish. We found a sufficient condition for the stability of the only predator extinct equilibrium point in Theorem \ref{thm.8}\ref{thm. boundary stability}. As a result, there's a basin of attraction around it, and depending on the size of the initial population, predator species may become extinct.

Further, we have also explored all possible bifurcations including transcritical bifurcation, saddle-node bifurcation, and Hopf bifurcation with respect to different  bifurcation parameters such as the level of fear, the prey's death rate caused by intra-prey competition, and the prey's birth rate as bifurcation parameters. We have numerically validated all of the derived bifurcations and stability results through one/ two parameter bifurcation diagram. We have found that our model's Hopf bifurcation is subcritical in contrast to traditional prey-predator models, which ignore the cost of fear and have supercritical Hopf bifurcations in general. We acknowledge the conclusion of Wang et al. \cite{wang2016modelling} demonstrating numerically that a high amount of fear can stabilise the prey-predator system by eliminating the occurrence of periodic solutions. For both bifurcation parameters, the prey's death rate caused by the intra-prey competition and the prey's birth rate, there exists a region of bi-stability. The region occurs between the predator extinct equilibrium point and an interior equilibrium point. It may cause the extinction of predator species depending on the  size of initial population. We have obtained that there is a possibility for the coexistence of the species for all values of the death rate caused by intraprey competition (refer the Fig. \ref{Fig.2}). In addition, we have also obtained that predators become extinct for higher values of the prey birth rate ( refer the Fig. \ref{Fig.3}). We have also investigated the impact of the level of fear and the birth rate of prey on the equilibrium prey density and equilibrium predator densities. We have found some conditions for which both prey and predator equilibrium densities decrease as the fear level increases. Conditions have also been found under which both prey and predator equilibrium densities increase as the prey's birth rate increases.

A diverse range of patterns has been found in biological systems. However decades of research, detailed understanding of how these patterns arise still away from us \cite{maini2019turing}. In this study, the dynamical behaviour of proposed spatial model has been discussed analytically as well as numerically. We have also obtained the condition for occurrence of Turing instability and found that increasing value of parameter level of fear factor $k$, the system moves toward spatial instability, whereas the non spatial model shows stable dynamics (refer to Fig. \ref{Fig.Insta(b)}). The fear factor has an important impact on the spatially inhomogeneous distribution of the two species \cite{han2020cross}. We have observed that although varying value of fear factor does not much change in the spatiotemporal pattern but upper bound of density of predator population decreases with increasing value of this parameter. In case of temporal prey-predator model, one can observe that as the level of fear factor $k$ increases, the predator density gradually decreases and agrees with those in \cite{zhang2019impact,sasmal2018population}. Thus fear factor shows the same  impact on the density distribution for spatial system as in temporal system. The predator density is distributed in the form of mixture of spot and stripes and higher density of predator is concentrated near the axes of domain of consideration. Diffusion coefficient of predator population has also great impact on the spatial density  distribution of the predator population, as varying value of $D_{2}$, we have found the irregular stripe and mixture of stripe and spots. Time evolution plays very important role to understand the density distribution of species over space and variation time leads the cluster of spot pattern.\\
\textbf{Ecological Implications:}
Extinction and persistence are two important phenomena of population ecology, as they allow population ecologists to understand whether a species will survive or become extinct in a given ecosystem. A general scenario is that when the natural death rate of a prey species exceeds its birth rate, the prey species is wiped out. It means that irrespective of whether any other species exist or not in the system, the prey species may go extinct. In this ecosystem, if the growth rate of predator species is entirely dependent on prey species, they will also go extinct. In this case, whatever the initial population size of each species be, both the species will disappear from the system. In a prey-predator interaction, when the predator density is higher than a certain threshold value, it may drive the prey to extinction not only by killing them directly but also by inducing fear in them. They create a substantially high level of fear in the prey population. Therefore, prey populations change their natural behaviour to keep themselves safe from predators. In this process, they lose their time for reproduction. Predators may also go extinct when their densities are higher than a certain threshold value. In this case, two ecological scenarios may occur: (\romannumeral 1 ) intra predator competition increases; and (\romannumeral 2) because of the fear of predators, the prey population makes their group defence tight and does not allow predators to kill them. In both cases, predator population declines and predators go extinct in the system. Normally, we observe that as prey number increases, predator number also increases. However, when prey species show group defence, predator species may become extinct in the system. Increased prey numbers strengthen the group defence of prey and make it rare to predate on themselves. Therefore, predators starve because of the unavailability of food and eventually die out of the system. Another thing that we observe is that as intra-prey competition increases, prey number decreases. Because of the increased intra-prey competition, prey populations are wiped out of the system, and as a result, predator species are also wiped out. However, in the presence of group defence, we observe something different. Both populations may coexist at their lower population densities.
Moreover different bifurcations occur in our proposed model which have the following sense in the context of ecology: a saddle node bifurcation in our study can lead to critical shifts in population dynamics, extinction risks, and changes in ecosystem stability. Understanding these implications is vital for conservation and ecosystem management efforts, as well as for predicting and mitigating potential ecological crises. The occurrence of transcritical bifurcation signifies a change in the ecosystem's dominant species. Before this bifurcation, either the predator or prey species may hold greater abundance and dominance. Following the bifurcation, these roles can swap, resulting in an altered ecosystem structure. For instance, if the prey had been dominant previously, the predator could become more abundant and influential within the ecosystem. Further, the system's resilience, or its ability to recover from disturbances, can be influenced by the Hopf bifurcation. Understanding the conditions under which the system exhibits Hopf bifurcation can help identify when it is most resilient or vulnerable to external perturbations, such as habitat loss, climate change, or invasive species.
 
In conclusion, the fear effect induced group defence, the prey's birth rate, and the prey's death rate caused by intra-prey competition are probably what determine the dynamics of the system. The inclusion of the Holling type \RomanNumeralCaps{4} functional response, the fear factor, and the prey refuge results in rich dynamics via occurence of different equilibrium points. The stability of these equilibrium points may alter when local bifurcations take place. The dynamics of our suggested model under various realistic time delays, various types of group defence functions, as well as diverse prey refuges, would be fascinating to observe in near future.\\

\noindent{\bf Acknowledgments} \\
 The research work of Jai Prakash Tripathi is supported by the Science and Engineering Research Board (SERB), India [File No. MTR/2022/001028]. The research work of Satish Kumar Tiwari is supported by RGIPT, Jais, Amethi, India. The research of YK is partially supported by NSF-DMS (1716802 \& 2052820); and The James S. McDonnell665
Foundation 21st Century Science Initiative in Studying Complex Systems Scholar Award (DOI-27
10.37717/22002047).

\appendix
\section{Appendix}
\subsection{Proof of Theorem \ref{thm.1}}\label{appA1}
\textbf{Positivity} 
\begin{proof}

Considering $x$ and $y$ as functions of $s$, we integrate first equation of the model system (\ref{7}) to get,
\begin{align*}
  \int_{0}^{t} \ \frac{d x}{x}&= \int_{0}^{t}\left(\left(\frac{r}{1+k\alpha y}\right)-d_1 - d_2 x - \frac{ \left(-m+1 \right) \beta y}{a +   \left(-m+1 \right) \alpha b x + \left(-m+1\right)^2 x^2}\right)\ d s,\\
   \ln{x(t)}-\ln{x(0)}&=\int_{0}^{t}\left(\left(\frac{r}{1+k\alpha y}\right)-d_1 -  d_2 x - \frac{ \left(-m+1 \right) \beta y}{a +   \left(-m+1 \right) \alpha b x + \left(-m+1\right)^2 x^2}\right)\ d s,\\
 x(t)& = x(0) exp\left( \int_{0}^{t}\left(\left(\frac{r}{1+k\alpha y}\right)-d_1 - d_2 x  - \frac{ \left(-m+1 \right) \beta y}{a +   \left(-m+1 \right) \alpha b x + \left(-m+1\right)^2 x^2}\right)\ d s\right).
\end{align*}
Since $x(0)>0$ and exponential is always positive, hence $x(t)>0~\forall ~t.$
Similarly, using the second equation of the model system (\ref{7}), we get
\begin{align*}
\int_{0}^{t} \ \frac{d y}{y}&= \int_{0}^{t}\left(-d - e y + \frac{\left(-m+1 \right)  \beta c  x}{a +\left(1-m \right)  \alpha  b x + \left(1-m\right)^2 x^2}\right)\ d s,\\
 \ln{y(t)}-\ln{y(0)}&=\int_{0}^{t}\left(-d - e y + \frac{\left(-m+1 \right)  \beta c  x}{a + \left(1-m \right) \alpha  b x + \left(1-m\right)^2 x^2}\right)\ d s,\\
  y(t)&= y(0) exp\left(\int_{0}^{t}\left(-d - e y + \frac{\left(-m+1 \right)  \beta c  x}{a +  \left(1-m \right) \alpha  b x + \left(1-m\right)^2 x^2}\right)\ d s\right).
 \end{align*}
 Since $y(0)>0$ and exponential is always positive. Hence ~$y(t)>0~~  \forall ~t$.
\end{proof}
\noindent{\textbf{Boundedness}}
\begin{proof}
 \begin{enumerate}[label=(\roman*)]
 \item Using the first equation of the model system (\ref{7}), we get
  \begin{align*}
  \frac{d  x}{d t}& \leq \frac{r x}{1+k\alpha y} -d_1 x -d_2 x^2
    \leq r x - d_1 x -d_2 x^2
    = x \left(\left(r- d_1\right)  -d_2 x\right). \\
    \end{align*}
    Now we use the comparison lemma (\cite{chen2005nonlinear}) to obtain
    \begin{align*} 
   \limsup\limits_{t\to+\infty}x(t)\leq \left(\frac{r-d_1}{d_2}\right);~~provided~~ r>d_1.
   \end{align*}
    \item Using the second equation of the model system (\ref{7}), we get 
    \begin{align*}
\frac{d y}{d t} &\leq  -d y -e y^2 +\frac{c\beta \left(1-m\right) x y}{ b \alpha \left(1-m \right) x}
=  -d y -e y^2 +\frac{c\beta  y}{ b \alpha }
 = y \left(\frac{c\beta-db\alpha}{b\alpha}-e y\right).
\end{align*}
Now we use the comparison lemma (\cite{chen2005nonlinear}) to obtain
\begin{align*} 
\limsup\limits_{t\to+\infty}y(t)\leq\left(\frac{c\beta-b d\alpha}{b\alpha e}\right);~ provided~ c\beta>b d\alpha.
\end{align*}
\end{enumerate}
\end{proof}
\subsection{Proof of Theorem \ref{thm.8}}\label{appA3}
\begin{proof}
In order to analyze the stability of equilibrium points, we use the jacobian matrix of model system (\ref{7}) at any point ~$(x,y)$~ which is given by  
$$J=\left(\begin{array}{cc}
     g_x&g_y  \\
     \\
     h_x&h_y 
     \end{array}\right),$$
    where
    \begin{equation*}
    \begin{aligned}
    g_x=&\frac{(-m+1) \beta  x y \left((-m+1) \alpha  b +2 (-m+1)^2 x\right)}{\left(a+(-m+1) \alpha  b  x+(-m+1)^2 x^2\right)^2}-\frac{  (-m+1) \beta y}{a+(-m+1) \alpha  b x+(-m+1)^2 x^2}-2 d_2 x-d_1\\
    &+\frac{r}{\alpha  k y+1},\\
     g_y=&-\frac{(-m+1) \beta  x}{a+(-m+1) \alpha  b x+(-m+1)^2 x^2}-\frac{\alpha  k r x}{(\alpha  k y+1)^2},\\
      h_x=&\frac{(m-1) \beta  c y \left(-a + (m-1)^2 x^2\right)}{(a+(m-1) x (-\alpha  b + (m-1) x))^2},\\
      h_y=&\frac{(-m+1) \beta  c x}{a+(-m+1) \alpha  b x+(-m+1)^2 x^2}-d-2 e y.
    \end{aligned}
      \end{equation*}
    \begin{enumerate}[label=(\roman*)] 
    \item The jacobian matrix at $E_0(0,0)$ is given by
    \begin{equation*}
    J(E_0)= (J)_{(E_0)}=\left(
\begin{array}{cc}
-d_1 +  r& 0 \\
\\
 0 & -d \\
\end{array}
\right).
\end{equation*}
Here, $ J(E_0) $ is diagonal matrix whose eigenvalue are given by diagonal entries.
    $\lambda_1=r-d_1$~  and ~$\lambda_2=-d$~ are eigenvalues of~$ J(E_0) $. So, $\lambda_2$ is always negative because $d>0$. Therefore, $E_0(0,0)$ is saddle if  $\lambda_1>0$ $(i.e.~-d_1 + r>0)$ and $E_0(0,0)$ is locally asymptotically stable if  $\lambda_1<0$ $(i.e.~-d_1 + r<0)$.
   The global asymptotic stability of $E_0(0,0)$ under the condition $r-d_1 < 0$ follows from the Theorem \ref{thm.1} and
the local stability of $E_0(0,0)$.
    \item The jacobian matrix at ~$E_1(\frac{r-d_1}{d_2},0)$~is given by
    \begin{equation*}
       J(E_1)= (J)_{E_1}=\left(
\begin{array}{cc}
 d_1 -r & -\frac{k r \alpha  \left(r-d_1\right)}{d_2}-\frac{(-m+1) \beta  \left(r-d_1\right)}{\left(\frac{(-m+1)^2 \left(r-d_1\right){}^2}{d_2^2}+\frac{b (-m+1) \alpha  \left(r-d_1\right)}{d_2}+a\right) d_2} \\
 \\
 0 & \frac{\left(-m+1\right) \beta c \left(r-d_1\right)d_2-d\left(ad_2^2 +\left(-m+1\right) \alpha b \left(r-d_1\right)d_2 +\left(-m+1\right)^2 \left(r-d_1\right)^2\right)}{\left( ad_2^2 +\left(-m+1\right) \alpha  b \left(r-d_1\right)d_2 +\left(-m+1\right)^2 \left(r-d_1\right)^2 \right)} 
\end{array}
\right).
 \end{equation*}
Here, $J(E_1)$ is an upper triangular matrix, so eigenvalue of   $J(E_1)$ are given by diagonal entries.
Hence, $\lambda_1=d_1-r$ and ~$\lambda_2=\frac{\left(-m+1\right) \beta  c \left(r-d_1\right)d_2-d\left(ad_2^2 +\left(-m+1\right) \alpha  b \left(r-d_1\right)d_2 +\left(-m+1\right)^2 \left(r-d_1\right)^2\right)}{\left( ad_2^2 +\left(-m+1\right) \alpha b \left(r-d_1\right)d_2 +\left(-m+1\right)^2 \left(r-d_1\right)^2 \right)}$~ are eigenvalues of~ $J(E_1)$. So when $r>d_1$, and~$d\left(ad_2^2 +\left(-m+1\right) \alpha b \left(r-d_1\right)d_2 +\left(-m+1\right)^2 \left(r-d_1\right)^2\right)>\left(-m+1\right) \beta  c \left(r-d_1\right)d_2$,  the predator extinct equilibrium point $E_1(\frac{r-d_1}{d_2},0)$~is asymptotically stable.


 \end{enumerate}
\end{proof}
\subsection{Proof of Theorem \ref{thm.15}}\label{appA4}
\begin{proof}
At the interior equilibrium $E^*$, we take into account the Jacobian matrix (\ref{9}). The following is the characteristic equation at $E^*$ : 
\begin{equation}\label{10}
\lambda^2-S(k)\lambda+T(k)=0,     
\end{equation}
where $S=Tr(J(E^*))$ and $T=Det(J(E^*)).$ Eq. (\ref{10}) has two simple complex roots
\begin{equation}\label{11}
    \lambda_{1,2}=u(k)\pm\iota v(k).
\end{equation}
The characteristic equation (\ref{10}) becomes
\begin{equation}\label{12}
\lambda^2+T(k)=0,     
\end{equation}
at $k=k_H$ since $S(k_H)=0$. Computing for the roots of Eq. (\ref{12}) yields $ \lambda_{1,2}=\pm\iota \sqrt{T}$. As a result, we have two eigenvalues that are entirely imaginary. We also confirm the transversality criterion. In Eq. (\ref{11}), for any $k$ in the neighborhood of $k_H$ , we have $u(k)=Re\left[\lambda_i(k)\right]=\frac{1}{2}S(k)$ and $v(k)=\sqrt{T(k)-\frac{[S(k)]^2}{4}}$. Thus,
\begin{align*}
    u(k)=&\frac{1}{2}\left(\frac{ c ( -m+1) \beta x^*}{a+( -m+1) b \alpha  x^*+ (x^*)^2 ( -m+1)^2 }+\frac{\beta  ( -m+1) x^* y^* \beta  \left(( -m+1)  b \alpha  +2 ( -m+1)^2 x^*\right)}{\left(a+( -m+1)  b \alpha  x^*+ (x^*)^2 ( -m+1)^2 \right)^2}\right.\\
    &-\left.\frac{( -m+1) \beta  y^*}{a+( -m+1)  b \alpha  x^*+( -m+1)^2 (x^*)^2}-2 d_2 x^*-d-d_1-2 e y^*+\frac{r}{\alpha  k y^*+1}\right).
\end{align*}
When $k=k_H$, it is evident that $u(k)=0$. Now we differentiate $u(k)$ with respect to $k$ to get
\begin{align*}
    \dfrac{d u(k)}{d k}\Big|_{k=k_H}=&\frac{1}{2}\left[\dfrac{d}{d k}\left(\frac{ c ( -m+1) \beta x^*}{a+( -m+1)  b \alpha x^*+  (x^*)^2 ( -m+1)^2}+\frac{ ( -m+1) x^* y^* \beta \left(( -m+1)  b \alpha +2 ( -m+1)^2 x^*\right)}{\left(a+( -m+1)  b \alpha x^*+ (x^*)^2 ( -m+1)^2 \right)^2}\right.\right.\\
    &-\left.\left.\frac{( -m+1) \beta  y^*}{a+( -m+1)  b \alpha x^*+ (x^*)^2 ( -m+1)^2 }-2 d_2 x^*-d-d_1-2 e y^*+\frac{r}{\alpha  k y^*+1}\right)\right]\Bigg|_{k=k_H}
    \end{align*}
    We point out that since $x^*$ and $y^*$ can be dependent on the critical parameter K, it is impossible to determine their exact derivatives because we lack their explicit expression. Thus, if $\dfrac{d}{d k}\left(S(k)\right)\ne 0$ at $k=k_H$, the transversality condition is satisfied. As a result, the model system (\ref{7}) experiences Hopf-bifurcation at $k=k_H$ at the coexisting equilibrium point $E^*$.
\end{proof}

\section{Appendix}
\subsection{Coefficients of Eq. (\ref{8})}\label{appB1}
\begin{align*}
A_1=&d_2 (-e) (m-1)^6 (e-\alpha  d k),\\
A_2=&e (m-1)^5 \left(\alpha  d_2 (-3 \alpha  b d k+3 b e+\beta  c k)-d_1 (-1+m) (e-\alpha  d k)+e (-1+m) r\right),\\
A_3=&-e (m-1)^4 \left(d_2 \left(3 \left(a+\alpha ^2 b^2\right) (e-\alpha  d k)+2 \alpha ^2 b \beta  c k\right)+\alpha  d_1 (m-1) (3 \alpha  b d k-3 b e-\beta  c k)+3 \alpha  b e (m-1) r\right),\\
A_4=&(m-1)^3 \left(\alpha  d_2 e \left(\beta  c k \left(2 a+\alpha ^2 b^2\right)+b \left(6 a+\alpha ^2 b^2\right) (e-\alpha  d k)\right)-d_1 e (m-1) \left(3 \left(a+\alpha ^2 b^2\right) (e-\alpha  d k)\right.\right.\\
&\left.\left.+2 \alpha ^2 b \beta  c k\right)+(m-1) \left(3 e^2 r \left(a+\alpha ^2 b^2\right)+\beta  d (m-1) (\alpha  d k-e)\right)\right),\\
A_5=&-(m-1)^2 \left((m-1) \left(\alpha  b e^2 r \left(6 a+\alpha ^2 b^2\right)+2 \alpha  b \beta  d (m-1) (\alpha  d k-e)+\beta ^2 c (m-1) (e-2 \alpha  d k)\right)+a d_2 e \right.\\
&\left.\left(3 \left(a+\alpha ^2 b^2\right) (e-\alpha  d k)+2 \alpha ^2 b \beta  c k\right)-\alpha  d_1 e (m-1) \left(\beta  c k \left(2 a+\alpha ^2 b^2\right)+b \left(6 a+\alpha ^2 b^2\right) (e-\alpha  d k)\right)\right),\\
A_6=&(m-1) \left((m-1) \left(3 a^2 e^2 r+3 a \alpha ^2 b^2 e^2 r+2 a \beta  d (m-1) (\alpha  d k-e)+\alpha  \beta  (m-1) (\alpha  b d-\beta  c) \right.\right.\\
&\left.\left.(\alpha  b d k-b e-\beta  c k)\right)+a e \left(a \alpha  d_2 (3 b (e-\alpha  d k)+\beta  c k)-d_1 (m-1) \left(3 \left(a+\alpha ^2 b^2\right) (e-\alpha  d k)+2 \alpha ^2 b \beta  c k\right)\right)\right),\\
A_7=&-a \left((m-1) \left(3 a \alpha  b e^2 r+\beta  (m-1) (2 \alpha  b d (\alpha  d k-e)+\beta  c (e-2 \alpha  d k))\right)+a e \left(a d_2 (e-\alpha  d k)+\alpha  d_1 (m-1)\right.\right.\\
&\left.\left. (3 \alpha  b d k-3 b e-\beta  c k)\right)\right),\\
A_8=&a^2 \left(a d_1 e (\alpha  d k-e)+a e^2 r+\beta  d (m-1) (\alpha  d k-e)\right).
\end{align*}
\subsection{Expressions of differentiation of functions \lq f' and \lq g':}\label{appB2}
\begin{equation*}
\begin{aligned}
    g_x=&\frac{(-m+1) \beta  x y \left((-m+1) b  \alpha  +2 (-m+1)^2 x\right)}{\left(a+(-m+1) b  \alpha  x+  x^2 (-m+1)^2\right)^2}-\frac{(-m+1) \beta  y}{a+(-m+1)  b  \alpha  x+ x^2 (-m+1)^2 }-2 d_2 x-d_1\\
    &+\frac{r}{\alpha  k y+1},~
     g_y=-\frac{(-m+1) \beta  x}{a+(-m+1)  b  \alpha x+ x^2 (-m+1)^2 }-\frac{\alpha  k r x}{(\alpha  k y+1)^2},  \\
    g_{xx}=&(-(-m+1)) \beta  x y \left(\frac{2 \left((-m+1) b  \alpha  +2 (-m+1)^2 x \right)^2}{\left(a+(-m+1) b \alpha   x+ x^2 (-m+1)^2 \right)^3}-\frac{2 (-m+1)^2}{\left(a+(-m+1) b \alpha  x+ x^2 (-m+1)^2 \right)^2}\right)\\
    &-\frac{2(-1+m)  \beta  y \left((-m+1)  b \alpha  +2 (-m+1)^2 x\right)}{\left(a+ (-m+1)  b \alpha   x+ x^2 (-m+1)^2 \right)^2}-2 d_2,  \\
    g_{x y}=&\frac{(-m+1) \beta  x \left((-m+1)  b \alpha +2 (-m+1)^2 x\right)}{\left(a+(-m+1) b \alpha x+ x^2 (-m+1)^2 \right)^2}-\frac{ (-m+1) \beta }{a+(-m+1) b \alpha  x+ x^2 (-m+1)^2 }-\frac{\alpha  k r}{(\alpha  k y+1)^2},  \\
    g_{y x}=&\frac{(-m+1) \beta  x \left((-m+1)  b \alpha +2 (-m+1)^2 x\right)}{\left(a+(-m+1)  b \alpha  x+ x^2 (-m+1)^2 \right)^2}-\frac{ (-m+1) \beta }{a+(-m+1)  b \alpha x+ x^2 (-m+1)^2 }-\frac{\alpha  k r}{(\alpha  k y+1)^2},\\
    g_{y y}=&\frac{2 \alpha ^2 k^2 r x}{(\alpha  k y+1)^3},~~h_y=\frac{(-m+1) \beta  c x}{a+(-m+1) b \alpha  x+(-m+1)^2 x^2}-d-2 e y,~
      h_{y y}= -2 e\\
     h_x=&\frac{(-m+1) \beta  c  y}{a+(-m+1)   b \alpha  x+ x^2 (-m+1)^2 }-\frac{ (-m+1) \beta  c x y \left((-m+1)   b \alpha  +2 (-m+1)^2 x\right)}{\left(a+(-m+1)   b \alpha  x+ x^2 (-m+1)^2 \right)^2}\\
     \end{aligned}
     \end{equation*}

     \begin{equation*}
\begin{aligned}
    h_{xx}=&(-m+1) \beta  c  x y \left(\frac{2 \left((-m+1) b \alpha  +2 (-m+1)^2 x\right)^2}{\left(a+(-m+1) b \alpha   x+ x^2 (-m+1)^2 \right)^3}-\frac{2 (-m+1)^2}{\left(a+(-m+1) b \alpha  x+  x^2 (-m+1)^2\right)^2}\right)\\
    &-\frac{2 (-m+1) \beta  c y \left((-m+1)  b \alpha  +2 (-m+1)^2 x\right)}{\left(a+(-m+1)  b \alpha   x+ x^2 (-m+1)^2 \right)^2},  \\
    h_{x y}=&\frac{c (-m+1) \beta }{a+(-m+1) b \alpha   x+ x^2 (-m+1)^2 }-\frac{ c (-m+1) \beta   x \left((-m+1)   b \alpha +2 (-m+1)^2 x\right)}{\left(a+(-m+1)   b \alpha x+ x^2 (-m+1)^2 \right)^2},  \\
    h_{y x}=&\frac{c (-m+1) \beta }{a+(-m+1) b \alpha   x+ x^2 (-m+1)^2 }-\frac{ c (-m+1) \beta   x \left((-m+1)   b \alpha +2 (-m+1)^2 x\right)}{\left(a+(-m+1)   b \alpha x+ x^2 (-m+1)^2 \right)^2}.  
    \end{aligned}
\end{equation*}
\subsection{}\label{appB3}
\textbf{Expressions of $(\alpha_{ij})$ and $(\beta_{ij})$}
\begin{equation*}
 \begin{aligned}
\alpha_{10} =&\frac{ (-m+1) \beta x y \left((-m+1) b \alpha   +2 (-m+1)^2 x\right)}{\left(a+ (-m+1)  b \alpha x+(-m+1)^2 x^2\right)^2}-\frac{\beta  (-m+1) y}{a+(-m+1)  b \alpha  x+(-m+1)^2 x^2}-2 d_2 x-d_1+\frac{r}{\alpha  k y+1},\\
\alpha_{01}= &-\frac{\beta  (-m+1) x}{a+(-m+1) \alpha  b x+(-m+1)^2 x^2}-\frac{\alpha  k r x}{(\alpha  k y+1)^2},\\
\alpha_{11} =&\frac{1}{2} \left(\frac{(-m+1) \beta  x \left(  b (-m+1) \alpha +2 (-m+1)^2 x\right)}{\left(a+(-m+1)  \alpha  b  x+  x^2 (-m+1)^2\right)^2}-\frac{(-m+1) \beta  }{a+ (-m+1)  \alpha  b  x+ x^2 (-m+1)^2 }-\frac{\alpha  k r}{(\alpha  k y+1)^2}\right),\\
\alpha_{12}=&\frac{\alpha ^2 k^2 r}{3 (\alpha  k y+1)^3},
~\alpha_{02}=\frac{\alpha ^2 k^2 r x}{(\alpha  k y+1)^3},
~\alpha_{03}=-\frac{\alpha ^3 k^3 r x}{(\alpha  k y+1)^4},\\
\alpha_{21}=&\frac{1}{6} \left((-(-m+1)) \beta  x \left(\frac{2 \left((-m+1) b \alpha   +2 (-m+1)^2 x\right)^2}{\left(a+(-m+1)  \alpha  b x+ x^2 (-m+1)^2 \right)^3}-\frac{2 (-m+1)^2}{\left(a+(-m+1)  \alpha  b  x+ x^2 (-m+1)^2 \right)^2}\right)\right.\\
&\left.-\frac{2 (m-1) \beta   \left((-m+1)  b +2 (-m+1)^2 x\right)}{\left(a+(-m+1)  \alpha  b  x+(-m+1)^2 x^2\right)^2}\right),\\
\alpha_{20}=&\frac{1}{2} \left(\frac{2 (m-1)^2 \beta  y \left(a \alpha  b-3 a (m-1) x+(m-1)^3 x^3\right)}{(a+(m-1) x ((m-1) x-\alpha  b))^3}-2 d_2\right),\\
\alpha_{30}=&-\frac{(m-1)^3 \beta  y \left(a^2+a \left(4 b (m-1) \alpha  x-\alpha ^2 b^2-6  x^2 (m-1)^2\right)+x^4 (m-1)^4 \right)}{(a+(m-1) x (-\alpha  b +(m-1) x))^4},\\
\beta_{10}=&\frac{c (-m+1) \beta  y}{a+ b (-m+1) \alpha x+ x^2 (-m+1)^2}-\frac{ c (-m+1) \beta x y \left( b (-m+1) \alpha +2 (-m+1)^2 x\right)}{\left(a+ b (-m+1) \alpha x+ x^2 (-m+1)^2 \right)^2},\\
\beta_{01}=&\frac{c (-m+1) \beta   x}{a+ b (-m+1) \alpha x+ x^2 (-m+1)^2 }-d-2 e y,~\beta_{03}=0,~\beta_{12}=0,~\beta_{02}=-e,\\
\beta_{20}=&\frac{1}{2} \left( c (-m+1) \beta x y \left(\frac{2 \left( b (-m+1) \alpha +2 (-m+1)^2 x\right)^2}{\left(a+(-m+1) b \alpha  x+ x^2 (-m+1)^2 \right)^3}-\frac{2 (-m+1)^2}{\left(a+(-m+1) b \alpha  x+ x^2 (-m+1)^2 x^2\right)^2}\right)\right.\\
&\left.-\frac{2  (-m+1) c \beta   y \left((-m+1) b \alpha +2  x (-m+1)^2\right)}{\left(a+(-m+1) b \alpha  x+ x^2 (-m+1)^2 \right)^2}\right),\\
\beta_{30}=&\frac{ c (m-1)^3 \beta y \left(a^2+a \left(4  b (m-1) \alpha  x-\alpha ^2 b^2-6 x^2 (m-1)^2 \right)+ x^4 (m-1)^4 \right)}{(a+(m-1) x (-\alpha  b + (m-1) x))^4} ,\\
\beta_{21}=&\frac{1}{6} \left(  c (-m+1) \beta x \left(\frac{2 \left( b (-m+1) \alpha +2 (-m+1)^2 x\right)^2}{\left(a+(-m+1) b \alpha   x+ x^2 (-m+1)^2 \right)^3}-\frac{2 (-m+1)^2}{\left(a+(-m+1) b \alpha x+  x^2 (-m+1)^2\right)^2}\right)\right.\\
&\left.-\frac{2  c (-m+1) \beta  \left( b (-m+1) \alpha +2 x (-m+1)^2 \right)}{\left(a+ (-m+1) b \alpha   x+ x^2 (-m+1)^2 \right)^2}\right),\\
\beta_{11}=&\frac{1}{2} \left(\frac{ c (-m+1) \beta }{a+ b (-m+1) \alpha x+ x^2(-m+1)^2 }-\frac{ c (-m+1) \beta x \left( b (-m+1) \alpha +2 x (-m+1)^2 \right)}{\left(a+ b (-m+1) \alpha  x+ x^2 (-m+1)^2 \right)^2}\right).
 \end{aligned}   
\end{equation*}
\end{document}